\documentclass{article}

\usepackage{arxiv}

\usepackage{amsmath}

\newcommand{\bfbeta}{\mbox{\boldmath $\beta$}}

	\newcommand{\bfSigma}{\mbox{\boldmath $\Sigma$}}

\newcommand{\bfa}{{\bf a}}		\newcommand{\bfb}{{\bf b}}	
		\newcommand{\bfd}{{\bf d}}	
\newcommand{\bfe}{{\bf e}}				
			
		\newcommand{\bfI}{{\bf I}}

		\newcommand{\bft}{{\bf t}}		
			
		\newcommand{\bfx}{{\bf x}}	\newcommand{\bfX}{{\bf X}}
\newcommand{\bfy}{{\bf y}}		\newcommand{\bfz}{{\bf z}}	

\newcommand{\bfone}{{\bf 1}}
\newcommand{\bfzero}{{\bf 0}}

\newcommand{\bfvarepsilon}{\mbox{\boldmath $\varepsilon$}}

\usepackage[utf8]{inputenc}
\usepackage[T1]{fontenc}
\usepackage{microtype}
\usepackage{textcomp}

\usepackage{amsmath}
\usepackage{amssymb}
\usepackage{amsthm}
\usepackage{mathrsfs}
\usepackage{nicefrac}

\usepackage{graphicx}
\usepackage{xcolor}
\usepackage{adjustbox}

\usepackage{booktabs}
\usepackage{multirow}
\usepackage{longtable}

\usepackage{algorithm}
\usepackage{algpseudocode}

\usepackage{setspace}
\usepackage{changepage}
\usepackage{pdflscape}
\usepackage[title]{appendix}
\usepackage{manyfoot}

\usepackage{listings}

\usepackage{natbib}
\usepackage{doi}

\usepackage{cleveref}

\newtheorem{theorem}{Theorem}[section]
\newtheorem{lemma}[theorem]{Lemma}
\newtheorem{corollary}[theorem]{Corollary}
\newtheorem{proposition}[theorem]{Proposition}
\newtheorem{assumption}[theorem]{Assumption}

\title{Structured Screen-and-Select for Ultra-High-Dimensional Variable Selection}

\newif\ifuniqueAffiliation
\uniqueAffiliationtrue

\ifuniqueAffiliation 
\author{
         \href{https://orcid.org/0000-0003-4814-7602}{\includegraphics[scale=0.06]{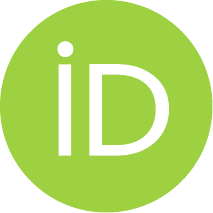}\hspace{1mm}Nilotpal~Sanyal\thanks{Corresponding author}} \\
	Department of Mathematical Sciences\\
	The University of Texas at El Paso\\
	El Paso, TX 79968\\
	\texttt{nsanyal@utep.edu} \\ 
	\And
	\href{https://orcid.org/0009-0007-0612-5188}{\includegraphics[scale=0.06]{orcid.pdf}\hspace{1mm}Padmore N.~Prempeh}\\
	Department of Mathematical Sciences\\
	The University of Texas at El Paso\\
	El Paso, TX 79968 \\
	\texttt{pprempeh@albany.edu}
}
\else
\usepackage{authblk}

\newbox{\orcid}\sbox{\orcid}{\includegraphics[scale=0.06]{orcid.pdf}} 
\author[1]{%
	\href{https://orcid.org/0000-0000-0000-0000}{\usebox{\orcid}\hspace{1mm}David S.~Hippocampus\thanks{\texttt{hippo@cs.cranberry-lemon.edu}}}%
}
\author[1,2]{%
	\href{https://orcid.org/0000-0000-0000-0000}{\usebox{\orcid}\hspace{1mm}Elias D.~Striatum\thanks{\texttt{stariate@ee.mount-sheikh.edu}}}%
}
\affil[1]{Department of Computer Science, Cranberry-Lemon University, Pittsburgh, PA 15213}
\affil[2]{Department of Electrical Engineering, Mount-Sheikh University, Santa Narimana, Levand}
\fi

\renewcommand{\headeright}{} 
\renewcommand{\undertitle}{} 
\renewcommand{\shorttitle}{Structured Screen-and-Select for Ultra-High-Dimensional Variable Selection}

\hypersetup{
pdftitle={Structured Screen-and-Select for Ultra-High-Dimensional Variable Selection},
pdfsubject={stat.ME, stat.TH},
pdfauthor={Nilotpal~Sanyal, Padmore N.~Prempeh},
pdfkeywords={High-dimensional variable selection, Structured screening and selection, Correlated predictors and proxy variables, Genomics and transcriptomics analysis},
}

\begin{document}
\maketitle

\addtocontents{toc}{\protect\setcounter{tocdepth}{-1}} 

\begin{abstract}
Ultra-high-dimensional data, in which $p$ greatly exceeds $n$, are common in genomics and biomedical research. Many screening methods rely mainly on marginal predictor--outcome associations and can miss active predictors with weak marginal signals, especially when predictors are strongly correlated. We develop Structured Screen-and-Select Variable Selection (S3VS), an iterative framework combining outcome-based screening with correlation-based local predictor sets. At each iteration, S3VS selects leading variables, forms local sets using predictor--predictor associations, applies a model-specific selector, aggregates selected and nonselected variables, and updates the candidate set and, when appropriate, the outcome representation. The framework supports multiple criteria for leading variables and leading sets, flexible aggregation rules, and implementations for linear, generalized linear, accelerated failure-time, and Cox models. For one specified one-step linear configuration, we establish a sure-screening result under conditions on proxy coverage, correlation separation, within-set retention, and active-preserving aggregation. Simulations compare full and first-iteration S3VS with one-pass SIS-based procedures in linear, binary-logistic, and Cox settings. S3VS can improve recovery or prediction when correlated predictors provide proxy information, but gains depend on predictor structure and the selector. In ovarian-cancer data, full-iteration S3VS with clinical variables had the strongest internal discrimination and early prediction, whereas SIS--Cox--LASSO with clinical variables had the strongest external discrimination. Neither molecular approach consistently reduced prediction error, and no gene was selected in all five outer folds. S3VS provides a framework for using predictor dependence before model-specific selection while clarifying when it helps. The methodology is implemented in the public \textsf{R} package \texttt{S3VS} on CRAN.
\end{abstract}


Ultra-high-dimensional data, in which the number of candidate predictors $p$ greatly exceeds the sample size $n$, are now common in genomics, transcriptomics, proteomics, medical imaging, and other data-rich fields. Scientific interest usually centers on a relatively small subset of relevant variables. Recovering that subset is difficult when the sample size is small relative to $p$, predictors are strongly correlated, and fitting a model to the full predictor set is computationally demanding. A relevant variable may be especially difficult to detect when its marginal association with the outcome is weak or when it is highly correlated with other predictors.

Penalized procedures such as LASSO \citep{tibshirani_1996}, adaptive LASSO \citep{zou_2006}, SCAD \citep{fan_li_2001}, and the Dantzig selector \citep{candes_tao_2007} are widely used for high-dimensional estimation and variable selection. When $p$ is extremely large, however, fitting these procedures directly can still be expensive, and strong collinearity can make the selected model unstable. Screening methods reduce the number of predictors before model fitting. Sure independence screening (SIS) \citep{fan_lv_2008}, for example, ranks predictors by marginal association and has a sure-screening property under suitable conditions. Its reliance on marginal association can also be a limitation. An active predictor may be missed when its individual association is weak even though it contributes jointly or is strongly correlated with another predictor that has a clearer marginal signal.

One way to address this problem is to use predictor dependence during screening rather than treating correlation only as an obstacle to estimation. The GWASinlps procedure of \citet{sanyal_et_al_2019} provided an important precursor to this perspective and introduced the terms \emph{leading variable} and \emph{leading set}. It combined outcome-based ranking with correlation-defined local neighborhoods and non-local-prior-based variable selection in a genomic setting. In the notation used here, its screening and neighborhood construction are represented by top-$k$ and top-$q$ rules, respectively, together with a single aggregation rule. In the present paper, we develop S3VS as a general iterative framework that builds on this insight and extends it to a unified methodology for linear, generalized linear, and survival models. S3VS makes leading-variable and leading-set construction, within-set selection, aggregation, and outcome updating modular, allowing these components to be matched to the statistical model and scientific objective. At each iteration, S3VS chooses leading variables according to their association with the current outcome quantity, forms a leading set around each leader using predictor--predictor association, applies a specified selector within each set, combines the selected and nonselected variables, and updates the outcome quantity and candidate predictor set.

S3VS provides a broad and explicitly configurable design space. Leading variables can be selected by Top-$k$, Fixed-Threshold, or Relative-to-Maximum criteria, and leading sets can be constructed by the analogous Top-$q$, Fixed-Threshold, or Relative-to-Maximum criteria. Selected variables can be aggregated by Conservative Selection or Liberal Selection, while nonselected variables can be aggregated by Conservative Exclusion from the Beginning, Conservative Exclusion from the End, or Liberal Exclusion. Within-set selection can use different model-specific procedures, including LASSO, non-local-prior selection \citep{johnson_rossell_2012}, bridge penalization \citep{huang_ma_2010}, AFTGEE \citep{aftgee}, and PVAFT \citep{afthd}. We give implementations for linear models, generalized linear models, and censored survival outcomes under accelerated failure time and Cox proportional hazards formulations. The paper also develops an asymptotic sure-screening result for a specified one-step linear configuration, conducts extensive simulations in linear, binary-logistic, and Cox settings, and provides the publicly available \textsf{R} package \texttt{S3VS} \citep{S3VS} on CRAN.

Our theoretical analysis focuses on a specified one-step linear S3VS configuration with top-$d_L$ marginal-score leading variables, top-$m_n$ absolute-correlation leading sets, a within-set selector satisfying the stated conditional sure-screening property, and an active-preserving aggregation map. This focused configuration permits a transparent analysis of proxy-assisted screening. We use the term proxy for a predictor that is correlated with an active predictor and has a stronger marginal association with the outcome. This use of proxy refers only to screening information and does not imply causal substitution or biological equivalence. Under the stated conditions, an active predictor with a weak marginal signal can be recovered through a more strongly associated proxy if the active predictor lies in the proxy's leading set and the within-set selector retains it. When the proxy signal is stronger than the active predictor's own marginal signal, the derived failure-probability bound can decrease faster than the corresponding SIS bound, and the local candidate sets can be smaller than an SIS candidate set. We also give post-screening estimation comparisons and a counterexample with no proxy-signal advantage. The theory therefore identifies the mechanism through which structured screening can improve recovery while clarifying the conditions under which that improvement is expected.

The numerical studies evaluate S3VS in linear, binary-logistic, and Cox models at dimensions up to $p=20{,}000$. The linear study compares full S3VS, the selection from its first iteration, and one-pass SIS across 13 dependence structures and three signal levels. The logistic and Cox studies focus on correlated-proxy and no-proxy settings to examine when predictor dependence is useful. Across these studies, the benefit of S3VS depends on the dependence structure and the within-set selector rather than being uniform across settings. In the ovarian-cancer application, we use event-stratified five-fold outer validation in TCGA followed by external validation in GSE9891. Under the primary $\lambda_{1\mathrm{se}}$ fit, full-iteration S3VS combined with clinical variables achieved the highest internal C-index (0.640) and 1-year AUC (0.754). SIS--Cox--LASSO combined with clinical variables achieved the highest external C-index (0.672). Neither molecular approach consistently reduced overall prediction error, and no gene was selected in all five outer folds.

The rest of the paper defines the general S3VS algorithm in Section~\ref{sec:general-s3vs}, gives the model-specific implementations in Section~\ref{sec:model-specific}, develops the theoretical results in Section~\ref{sec:asymptotic-analysis}, presents the simulation studies in Section~\ref{sec:simulation}, applies the method to ovarian-cancer survival data in Section~\ref{sec:real-data-analysis}, and concludes in Section~\ref{sec:conclusion}. Detailed proofs are given in the appendix, and additional simulation results are provided in the supplementary material.

\section{General S3VS Framework} 
\label{sec:general-s3vs}

\subsection{Overview} 
\label{sec:s3vs-overview}

Let $\bfX=(\bfx_1,\ldots,\bfx_p)\in\mathbb{R}^{n\times p}$ denote the predictor matrix, where $\bfx_r\in\mathbb{R}^n$ is the $r$th predictor, and let $\bfy$ denote the observed outcome representation for the statistical model under consideration. The representation of $\bfy$ and the model-specific treatment of the outcome are given in Section~\ref{sec:model-specific}. We consider the ultra-high-dimensional setting $p\gg n$ and assume that only a relatively small subset of the predictors is associated with the outcome. The columns of $\bfX$ are required to have unique names, which are used to identify predictors throughout the implementation.

S3VS builds on sure independence screening \citep{fan_lv_2008} and on the structured screen-and-select formulation introduced in GWASinlps \citep{sanyal_et_al_2019}. That work established the leading-variable and leading-set perspective for incorporating local predictor dependence into high-dimensional screening. The present paper develops this perspective into a common and modular framework for linear, generalized linear, and survival models. The top-size screening idea and non-local-prior selection used in the earlier formulation are incorporated as particular components of a broader framework that also supports threshold-based screening, multiple aggregation operators, and model-specific selection procedures. A preliminary arXiv report also considered an iterative screen-and-select construction for accelerated failure time models. Unlike purely marginal screening, S3VS combines predictor--outcome association with dependence among predictors. It first identifies leading variables, forms a local leading set around each leader, applies a selector within each set, combines the resulting decisions, and updates the candidate set and, when specified, the outcome quantity. The model-specific screening score, within-set selector, and outcome update are defined in Section~\ref{sec:model-specific}.

At the beginning of iteration $i$, let $\mathcal{A}^{(i)}\subseteq\{1,\ldots,p\}$ denote the current candidate index set, let $\bfX^{(i)}=\bfX_{\mathcal{A}^{(i)}}\in\mathbb{R}^{n\times|\mathcal{A}^{(i)}|}$ denote the corresponding predictor submatrix, let $\mathcal{M}^{(i)}$ denote the cumulative set of predictors selected before iteration $i$, and let $\bfy^{(i)}$ denote the current outcome quantity. Throughout, $p$ denotes the number of predictors remaining after deterministic preprocessing. A predictor with zero sample variance is removed before initialization and cannot be selected. All scores are oriented so that larger finite values indicate stronger association, and nonfinite scores are ineligible for the corresponding screening or grouping operation.

The algorithm is initialized by \[ \mathcal{A}^{(1)}=\{1,\ldots,p\},\qquad \mathcal{M}^{(1)}=\varnothing,\qquad \bfy^{(1)}=\bfy,\qquad c_1=0, \] where $c_i$ is the cumulative number of iterations before iteration $i$ in which no variables were accepted into the selected set. The counter is not reset after a successful iteration. All sets below contain the original predictor indices, and all submatrices retain the original predictor names and column order.

At each iteration, S3VS performs five steps: it chooses leading variables, constructs leading sets, selects variables within each leading set, combines the resulting decisions, and updates the outcome quantity and candidate predictor set.

\subsection{S3VS Algorithm} \label{sec:s3vs-algorithm}

Each iteration of S3VS consists of the following five steps.

\begin{enumerate}

\item \textbf{Determine leading variables.}

For $r\in\mathcal{A}^{(i)}$, define the generic predictor--outcome screening score 

\begin{equation} 
A_i(r)=A\!\left(\bfx_r,\bfy^{(i)};\mathcal{M}^{(i)}\right). 
\label{eq:generic-screening-score} 
\end{equation} 

The leading variables are selected according to an association measure between each covariate and the current outcome quantity $\bfy^{(i)}$. The specific association measure depends on the statistical model and is described in Section~\ref{sec:model-specific}. A candidate with a nonfinite score, or with a failed model fit used to compute that score, is ineligible for leading-variable selection.

A collection of leading variables is obtained using one of the criteria in Section~\ref{sec:leading-variable-criteria}. The retained indices are ordered by decreasing $A_i(r)$: \[ \mathcal{L}^{(i)}=\{\ell_1^{(i)},\ldots,\ell_{k_i}^{(i)}\}\subseteq\mathcal{A}^{(i)}, \] and the corresponding leading-variable vectors are 

\begin{equation} 
\mathcal{Z}^{(i)}=\{\bfz_1^{(i)},\ldots,\bfz_{k_i}^{(i)}\},\qquad \bfz_j^{(i)}=\bfx_{\ell_j^{(i)}}\in\mathbb{R}^n. 
\label{eq:leading-variable-vectors} 
\end{equation} 

The ordering of $\mathcal{L}^{(i)}$ is retained in all subsequent operations because the conservative aggregation rules depend on leader rank. Ties are resolved by the original predictor-column order. If no finite screening score is available, then $k_i=0$ and the algorithm terminates.

\item \textbf{Determine leading sets.}

For leader $\bfz_j^{(i)}=\bfx_{\ell_j^{(i)}}$, define the generic predictor--leader association score 

\begin{equation} 
B_{ij}(r)=B\!\left(\bfx_r,\bfz_j^{(i)}\right),\qquad r\in\mathcal{A}^{(i)}. 
\label{eq:generic-leading-set-score} 
\end{equation} 

For each $j=1,\ldots,k_i$, a leading set \[ \mathcal{S}_j^{(i)}\subseteq\mathcal{A}^{(i)} \] is constructed using one of the criteria in Section~\ref{sec:leading-set-criteria}. The association function $B$ and the treatment of nonfinite or degenerate pairwise scores are implementation-specific and are stated in Section~\ref{sec:model-specific}. The leader is included explicitly whenever it is a valid member of the current candidate set.

Leading sets are reconstructed at every iteration exclusively from the predictors still contained in $\mathcal{A}^{(i)}$. Variables selected or excluded in earlier iterations are not eligible to appear in later leading sets. The same predictor may occur in multiple leading sets. Such overlaps are retained. Each leading set is analyzed independently even when two leading sets overlap or coincide. The current implementation does not combine, cache, or deduplicate overlapping within-set fits.

The terms \emph{leading variable} and \emph{leading set}, together with the principle of using local predictor dependence during screening, were introduced by \citet{sanyal_et_al_2019}. Here, this principle is developed into a more general iterative framework with multiple leading-variable and leading-set criteria, flexible aggregation choices, and model-specific selection procedures. It decomposes the original ultra-high-dimensional problem into smaller subproblems while preserving local dependence among predictors. At iteration $i$, the within-set stage requires at most $k_i$ independent local fits, with the dimension of the $j$th fit equal to $|\mathcal{S}_j^{(i)}|$. The exact computational cost depends on the model-specific score and selection operator.

\item \textbf{Perform variable selection within leading sets.}

A fully specified selection operator $\mathcal{V}$ is fitted separately within each leading set: 

\begin{equation} 
\mathcal{V}\!\left(\bfy^{(i)},\bfX_{\mathcal{S}_j^{(i)}}\right)=\left(\mathcal{S}_{j,\mathrm{sel}}^{(i)},\mathcal{S}_{j,\mathrm{nosel}}^{(i)}\right),\qquad j=1,\ldots,k_i, 
\label{eq:selection-operator} 
\end{equation} 

where $\mathcal{S}_{j,\mathrm{sel}}^{(i)}\subseteq\mathcal{S}_j^{(i)}$ and 

\begin{equation} 
\mathcal{S}_{j,\mathrm{nosel}}^{(i)}=\mathcal{S}_j^{(i)}\setminus\mathcal{S}_{j,\mathrm{sel}}^{(i)}. 
\label{eq:within-set-nonselection} 
\end{equation} 

The operator may be a penalized-likelihood, Bayesian, or other model-specific selection procedure. Its fitting model, tuning rule, fold assignment, scaling, convergence controls, and coefficient-decision rule are fixed inputs to a run and are specified in Section~\ref{sec:model-specific}.

The fit in \eqref{eq:selection-operator} uses only the current outcome quantity and the predictors in $\mathcal{S}_j^{(i)}$. Previously selected predictors are not included automatically in the within-set fitting model, and local coefficient estimates from previous iterations or other leading sets are neither retained nor used as offsets or starting values. Predictor centering and scaling, when used by a selector, are performed separately within each local fit according to the declared selector settings.

A failed local fit or a nonfinite coefficient or decision statistic is treated as an empty selected set and a full nonselected set. It therefore follows the same aggregation and irreversible-exclusion rules as any other local fit that selects nothing. For reproducibility, the predictor ordering, software versions, and execution order should be fixed or recorded. Parallel runs should also record the random-number stream and task schedule.

\item \textbf{Aggregate the within-set decisions.}

The leading sets are ordered according to the decreasing screening scores of their leaders. The current implementation uses the highest-ranked leading set, $\mathcal{S}_1^{(i)}$, to determine whether iteration $i$ is a selection iteration or an exclusion iteration. Define

\begin{equation} 
D_i=I\!\left\{\mathcal{S}_{1,\mathrm{sel}}^{(i)}\neq\varnothing\right\}. 
\label{eq:iteration-branch-indicator} 
\end{equation} 

Under the aggregation rules defined below, $D_i=1$ is equivalent to the iteration having a nonempty selected aggregate, because the first selected set is nonempty in the selection branch and the aggregate always retains its first-set contribution.

If $D_i=1$, the sets $\{\mathcal{S}_{j,\mathrm{sel}}^{(i)}\}_{j=1}^{k_i}$ are aggregated using the selected-variable rule specified in Section~\ref{sec:aggregation-selected}, yielding $\mathcal{S}_{\mathrm{sel}}^{(i)}$. No variables are excluded in that iteration: 

\begin{equation} 
\mathcal{S}_{\mathrm{nosel}}^{(i)}=\varnothing. 
\label{eq:no-exclusion-selection-iteration} 
\end{equation}

If $D_i=0$, no selected variables from any leading set are accepted during that iteration, including variables selected within lower-ranked leading sets. Instead, the nonselected sets $\{\mathcal{S}_{j,\mathrm{nosel}}^{(i)}\}_{j=1}^{k_i}$ are aggregated using the specified exclusion rule in Section~\ref{sec:aggregation-nonselected}, and 

\begin{equation} 
\mathcal{S}_{\mathrm{sel}}^{(i)}=\varnothing. 
\label{eq:no-selection-exclusion-iteration} 
\end{equation}

In particular, if every within-set selector returns the empty set, then $D_i=0$, all variables in each leading set are nonselected, and the selected set is empty. The chosen exclusion rule is then applied exactly as written. If it also returns the empty set, no candidate is removed, but the unsuccessful-iteration counter still increases. Thus, selection and exclusion are mutually exclusive within an iteration. There is no concurrent selection--exclusion conflict requiring a secondary conflict-resolution rule.

\item \textbf{Update the outcome quantity and algorithmic state.}

The cumulative selected set is updated by 

\begin{equation} 
\mathcal{M}^{(i+1)}=\mathcal{M}^{(i)}\cup\mathcal{S}_{\mathrm{sel}}^{(i)}. 
\label{eq:cumulative-selected-update} 
\end{equation} 

The candidate index set is updated by 

\begin{equation} 
\mathcal{A}^{(i+1)}=\mathcal{A}^{(i)}\setminus\left(\mathcal{S}_{\mathrm{sel}}^{(i)}\cup\mathcal{S}_{\mathrm{nosel}}^{(i)}\right),\qquad \bfX^{(i+1)}=\bfX_{\mathcal{A}^{(i+1)}}. 
\label{eq:candidate-update} 
\end{equation} 

Consequently, both selected and excluded predictors are permanently removed from the candidate set. The current implementation has no re-entry mechanism: a predictor excluded at one iteration cannot be reconsidered at a later iteration.

The unsuccessful-iteration counter satisfies 

\begin{equation} 
c_{i+1}=c_i+I(D_i=0). 
\label{eq:no-selection-counter} 
\end{equation} 

The counter records the cumulative, rather than consecutive, number of unsuccessful iterations and is not reset after a successful iteration.

Define the iteration-specific action set by 

\begin{equation} 
\mathcal{Q}^{(i)}=\begin{cases}\mathcal{S}_{\mathrm{sel}}^{(i)},&D_i=1,\\\mathcal{S}_{\mathrm{nosel}}^{(i)},&D_i=0.\end{cases} 
\label{eq:iteration-action-set} 
\end{equation} 

The generic outcome update is 

\begin{equation} 
\bfy^{(i+1)}=\mathcal{U}\!\left(\bfy^{(i)},\mathcal{S}_{\mathrm{sel}}^{(i)},\mathcal{S}_{\mathrm{nosel}}^{(i)},D_i\right). 
\label{eq:generic-update} 
\end{equation} 

The operator $\mathcal{U}$ may preserve or transform the current outcome representation and is specified for each supported model in Section~\ref{sec:model-specific}. No model-specific outcome update is implied by the general framework.

\end{enumerate}

\subsection{Leading Variable Criteria} \label{sec:leading-variable-criteria}

S3VS allows three rules for constructing the leading-index set $\mathcal{L}^{(i)}$. For compactness, let \( \operatorname{Top}_q\{c_r:r\in\mathcal{I}\} \) denote the indices of the $\min(q,N_{\mathcal I})$ largest finite values among $\{c_r:r\in\mathcal I\}$, ordered by decreasing value with ties resolved by original predictor-column order, where $N_{\mathcal I}=|\{r\in\mathcal I:c_r\in\mathbb{R}\}|$. If $N_{\mathcal I}=0$, the result is empty. The resulting leading indices are always ordered by decreasing screening score.

\subsubsection{Top-$k$ Criterion}

For a prespecified positive integer $k$, define 

\begin{equation} 
\mathcal{L}^{(i)}=\operatorname{Top}_{k}\left\{A_i(r):r\in\mathcal{A}^{(i)}\right\}. 
\label{eq:topk_leading} 
\end{equation} 

Thus, $k_i=\min\{k,N_{\mathcal A^{(i)}}\}$, where $N_{\mathcal A^{(i)}}$ is the number of finite screening scores in the current candidate set. This rule generalizes the ranking step in SIS \citep{fan_lv_2008}. The value of $k$ controls the maximum number of leaders retained at each iteration.

\subsubsection{Fixed-Threshold Criterion}

For a prespecified threshold $\tau_A$, define 

\begin{equation} 
\mathcal{L}^{(i)}=\left\{r\in\mathcal{A}^{(i)}:A_i(r)\geq\tau_A\right\}. 
\label{eq:fixed_threshold_leading} 
\end{equation} 

The retained indices are ordered by decreasing $A_i(r)$, and $k_i=|\mathcal{L}^{(i)}|$. If the set in \eqref{eq:fixed_threshold_leading} is empty, the S3VS procedure terminates.

\subsubsection{Relative-to-Maximum Criterion}

Let $A_{\max}^{(i)}=\max\{A_i(r):r\in\mathcal A^{(i)},\,A_i(r)\in\mathbb R\}$ when this set is nonempty. The software accepts a percentage parameter $\pi_A\in(0,100]$ and sets $\alpha_A=\pi_A/100$. The relative-to-maximum rule is 

\begin{equation} 
\mathcal{L}^{(i)}=\left\{r\in\mathcal{A}^{(i)}:A_i(r)\geq\alpha_A A_{\max}^{(i)}\right\}. 
\label{eq:percentile_leading} 
\end{equation} 

The retained indices are ordered by decreasing $A_i(r)$, and $k_i=|\mathcal{L}^{(i)}|$. This is a fraction-of-maximum rule, not an empirical percentile rule. If no finite score exists or no predictor satisfies the relative threshold, the algorithm terminates.

\subsection{Leading Set Criteria} \label{sec:leading-set-criteria}

For each $\bfz_j^{(i)}=\bfx_{\ell_j^{(i)}}\in\mathcal{Z}^{(i)}$, S3VS forms $\mathcal{S}_j^{(i)}$ from the scores $B_{ij}(r)$. The generic criteria below apply to any valid nonnegative association score for which larger values indicate stronger predictor--leader dependence. Nonfinite scores are ignored, and the leader is included explicitly whenever it is valid.

\subsubsection{Top-$q$ Criterion}

For a prespecified leading-set size $q\geq1$, define 

\begin{equation} 
\mathcal{S}_j^{(i)}=\left\{\ell_j^{(i)}\right\}\cup\operatorname{Top}_{q-1}\left\{B_{ij}(r):r\in\mathcal{A}^{(i)}\setminus\{\ell_j^{(i)}\}\right\},\qquad j=1,\ldots,k_i. 
\label{eq:topk-leadingset} 
\end{equation} 

Thus, $|\mathcal{S}_j^{(i)}|=1+\min\{q-1,N_{\mathcal A^{(i)}\setminus\{\ell_j^{(i)}\}}\}$, where the second term counts only finite association scores. This rule includes the leader and directly controls the maximum dimension of each within-set selection problem.

\subsubsection{Fixed-Threshold Criterion}

For a prespecified threshold $\tau_B$, define 

\begin{equation} 
\mathcal{S}_j^{(i)}=\left\{\ell_j^{(i)}\right\}\cup\left\{r\in\mathcal{A}^{(i)}:B_{ij}(r)\geq\tau_B\right\}. 
\label{eq:fixed-leadingset} 
\end{equation} 

The leading-set size therefore adapts to the observed predictor--leader association scores.

\subsubsection{Relative-to-Maximum Criterion}

Let $B_{\max,ij}=\max\{B_{ij}(r):r\in\mathcal A^{(i)},\,B_{ij}(r)\in\mathbb R\}$ when this set is nonempty. The software accepts a percentage parameter $\pi_B\in(0,100]$ and sets $\alpha_B=\pi_B/100$. The leading set is 

\begin{equation} 
\mathcal{S}_j^{(i)}=\left\{\ell_j^{(i)}\right\}\cup\left\{r\in\mathcal{A}^{(i)}:B_{ij}(r)\geq\alpha_B B_{\max,ij}\right\}. 
\label{eq:percentile-leadingset} 
\end{equation} 

The size of $\mathcal{S}_j^{(i)}$ depends on the observed predictor--leader association scores. This is a fraction-of-maximum rule, not an empirical percentile rule. The association score used in the implementations considered here, and the resulting equivalence between this rule and a fixed threshold when applicable, are stated in Section~\ref{sec:model-specific}.

\subsection{Aggregation Strategies} \label{sec:aggregation}

A predictor may occur in several leading sets and receive different within-set selection decisions. Aggregation reconciles these decisions. Let \( \left\{\mathcal{S}_{j,\mathrm{sel}}^{(i)}\right\}_{j=1}^{k_i} \) and \( \left\{\mathcal{S}_{j,\mathrm{nosel}}^{(i)}\right\}_{j=1}^{k_i} \) denote the ordered selected and nonselected sets at iteration $i$.

Aggregation is a central modular component of S3VS. The earlier GWASinlps formulation used a single aggregation rule. S3VS distinguishes the aggregation of selected variables from the aggregation of nonselected variables and provides multiple options for each. This makes the balance between variable recovery and candidate-space reduction explicit and allows the aggregation behavior to be selected according to the statistical objective.

\subsubsection{Aggregation of Selected Variables} \label{sec:aggregation-selected}

The selected-variable aggregation rule is applied only when \( \mathcal{S}_{1,\mathrm{sel}}^{(i)}\neq\varnothing. \)

\paragraph{Conservative Selection.}

Starting with the highest-ranked leading set, define the successive prefix intersections \[ \mathcal{I}_{m,\mathrm{sel}}^{(i)}=\bigcap_{j=1}^{m}\mathcal{S}_{j,\mathrm{sel}}^{(i)},\qquad m=1,\ldots,k_i. \] Because $\mathcal{S}_{1,\mathrm{sel}}^{(i)}\neq\varnothing$, define 

\begin{equation} 
m_{\mathrm{sel}}^{(i)}=\max\left\{m\in\{1,\ldots,k_i\}:\mathcal{I}_{m,\mathrm{sel}}^{(i)}\neq\varnothing\right\}. 
\label{eq:conservative-selection-prefix} 
\end{equation} 

The conservative aggregate is 

\begin{equation} 
\mathcal{S}_{\mathrm{sel}}^{(i)}=\mathcal{I}_{m_{\mathrm{sel}}^{(i)},\mathrm{sel}}^{(i)}=\bigcap_{j=1}^{m_{\mathrm{sel}}^{(i)}}\mathcal{S}_{j,\mathrm{sel}}^{(i)}. 
\label{eq:cons-select} 
\end{equation} 

Operationally, the implementation starts with the selected set from the highest-ranked leader, successively intersects it with lower-ranked selected sets, and returns the last nonempty prefix intersection. Once a prefix intersection becomes empty, no lower-ranked sets are considered further under this rule. If $k_i=1$, conservative and liberal selected-variable aggregation are identical. If selected sets are disjoint after the first leading set, conservative selection returns the selected set from the highest-ranked leader, whereas liberal selection returns their union. These labels describe only the set operations and their dependence on leader order. They do not imply an ordering of model size, false-discovery proportion, or formal error control.

\paragraph{Liberal Selection.}

The liberal rule retains every predictor selected in at least one leading set: 

\begin{equation} 
\mathcal{S}_{\mathrm{sel}}^{(i)}=\bigcup_{j=1}^{k_i}\mathcal{S}_{j,\mathrm{sel}}^{(i)}. 
\label{eq:lib-select} 
\end{equation}

\subsubsection{Aggregation of Nonselected Variables} \label{sec:aggregation-nonselected}

The nonselected-variable aggregation rule is applied only when \( \mathcal{S}_{1,\mathrm{sel}}^{(i)}=\varnothing. \) In this branch, all selected-variable decisions are discarded and $\mathcal{S}_{\mathrm{sel}}^{(i)}=\varnothing$.

\paragraph{Conservative Exclusion from the Beginning.}

Define the prefix intersections \[ \mathcal{I}_{m,\mathrm{beg}}^{(i)}=\bigcap_{j=1}^{m}\mathcal{S}_{j,\mathrm{nosel}}^{(i)},\qquad m=1,\ldots,k_i. \] If $\mathcal{S}_{1,\mathrm{nosel}}^{(i)}=\varnothing$, set $\mathcal{S}_{\mathrm{nosel}}^{(i)}=\varnothing$. Otherwise, define 

\begin{equation} 
m_{\mathrm{beg}}^{(i)}=\max\left\{m\in\{1,\ldots,k_i\}:\mathcal{I}_{m,\mathrm{beg}}^{(i)}\neq\varnothing\right\} 
\label{eq:conservative-exclusion-begin-index} 
\end{equation} 

and set 

\begin{equation} 
\mathcal{S}_{\mathrm{nosel}}^{(i)}=\mathcal{I}_{m_{\mathrm{beg}}^{(i)},\mathrm{beg}}^{(i)}=\bigcap_{j=1}^{m_{\mathrm{beg}}^{(i)}}\mathcal{S}_{j,\mathrm{nosel}}^{(i)}. 
\label{eq:cons-excl-begin} 
\end{equation} 

Thus, the algorithm returns the last nonempty intersection obtained while proceeding from the highest-ranked leading set toward lower-ranked sets.

\paragraph{Conservative Exclusion from the End.}

Let 

\begin{equation} 
j_{\mathrm{end}}^{(i)}=\max\left\{j\in\{1,\ldots,k_i\}:\mathcal{S}_{j,\mathrm{nosel}}^{(i)}\neq\varnothing\right\}, 
\label{eq:rightmost-nonempty-nonselection} 
\end{equation} 

provided that at least one nonselected set is nonempty. Trailing empty nonselected sets with indices greater than $j_{\mathrm{end}}^{(i)}$ are ignored. For $r=1,\ldots,j_{\mathrm{end}}^{(i)}$, define the suffix intersection \[ \mathcal{I}_{r,\mathrm{end}}^{(i)}=\bigcap_{j=r}^{j_{\mathrm{end}}^{(i)}}\mathcal{S}_{j,\mathrm{nosel}}^{(i)}. \] Let 

\begin{equation} 
r_{\mathrm{end}}^{(i)}=\min\left\{r\in\{1,\ldots,j_{\mathrm{end}}^{(i)}\}:\mathcal{I}_{r,\mathrm{end}}^{(i)}\neq\varnothing\right\}, 
\label{eq:conservative-exclusion-end-index} 
\end{equation} 

and set 

\begin{equation} 
\mathcal{S}_{\mathrm{nosel}}^{(i)}=\mathcal{I}_{r_{\mathrm{end}}^{(i)},\mathrm{end}}^{(i)}=\bigcap_{j=r_{\mathrm{end}}^{(i)}}^{j_{\mathrm{end}}^{(i)}}\mathcal{S}_{j,\mathrm{nosel}}^{(i)}. 
\label{eq:cons-excl-end} 
\end{equation} 

If every nonselected set is empty, then $\mathcal{S}_{\mathrm{nosel}}^{(i)}=\varnothing$. If $k_i=1$, the two conservative exclusion rules and liberal exclusion are identical.

\paragraph{Liberal Exclusion.}

The liberal rule excludes every predictor that is nonselected in at least one leading set: 

\begin{equation} 
\mathcal{S}_{\mathrm{nosel}}^{(i)}=\bigcup_{j=1}^{k_i}\mathcal{S}_{j,\mathrm{nosel}}^{(i)}. 
\label{eq:lib-excl} 
\end{equation}

Because selection and exclusion occur in separate branches, the implementation does not subsequently subtract $\mathcal{S}_{\mathrm{sel}}^{(i)}$ from $\mathcal{S}_{\mathrm{nosel}}^{(i)}$. Conservative aggregation emphasizes agreement among ranked leading sets, whereas liberal aggregation favors broader selection or more rapid reduction of the candidate space. These rules are generally order dependent and do not, by themselves, provide formal false-discovery-rate control or protect an active variable from exclusion.

\subsection{Stopping Rule and Returned Model} \label{sec:stopping-and-output}

Let $m_{\max}\geq1$ denote the prespecified selected-model-size limit and let $n_{\mathrm{skip}}\geq1$ denote the permitted cumulative number of unsuccessful iterations. Before beginning iteration $i$, S3VS evaluates 

\begin{equation} 
|\mathcal{M}^{(i)}|<m_{\max},\qquad |\mathcal{A}^{(i)}|>0,\qquad c_i<n_{\mathrm{skip}}. 
\label{eq:loop-conditions} 
\end{equation} 

The algorithm continues only while all three conditions hold. It also terminates inside an iteration if no leading variable is returned. Let $i_{\mathrm{stop}}$ denote the iteration at which the procedure stops. If the loop condition fails before iteration $i$, then $i_{\mathrm{stop}}=i$. If iteration $i$ terminates because $k_i=0$, then $i_{\mathrm{stop}}=i$. Otherwise, after a completed iteration, $i$ is incremented, and the first subsequent failed loop condition defines $i_{\mathrm{stop}}$.

Because the selected-model-size condition is evaluated before an iteration begins, an iteration may add several variables and produce \( |\mathcal{M}^{(i+1)}|>m_{\max}. \) Thus, $m_{\max}$ is a pre-iteration stopping limit rather than a hard truncation of the returned model. For example, $m_{\max}=100$ means that an iteration is started only while fewer than 100 variables have already been selected. An iteration can therefore take the total above 100. Likewise, $n_{\mathrm{skip}}=3$ means that the procedure stops after three cumulative unsuccessful iterations, whether or not they are consecutive. All structural choices, aggregation rules, update flags, and tuning values are fixed inputs to a run. S3VS does not select them using the unknown active set. If they are chosen from the data, that choice must be made within the corresponding training sample and included in the stated validation procedure.

The principal S3VS output is the selected-variable set 

\begin{equation} 
\widehat{\mathcal{M}}=\mathcal{M}^{(i_{\mathrm{stop}})}. 
\label{eq:final-selected-model} 
\end{equation} 

The implementation additionally returns the variables selected at each successful iteration and the elapsed runtime. It does not retain the coefficient estimates from the within-set fits or from temporary response-update fits, and it does not automatically return a final joint regression model.

For estimation or prediction after S3VS, a separate post-selection model is fitted jointly using $\bfX_{\widehat{\mathcal{M}}}$. The post-selection fitting method and its tuning rule must therefore be specified separately from the S3VS screening procedure. S3VS itself returns no coefficient vector. If coefficient estimates are required, a prespecified model is fitted jointly to $\bfy$ and $\bfX_{\widehat{\mathcal{M}}}$ after selection. Coefficients obtained from different S3VS iterations are not accumulated to form that post-selection fit.

The exclusion decision is irreversible under the current implementation. Therefore, there is no separate mechanism that protects a truly active predictor from permanent removal after it enters $\mathcal{S}_{\mathrm{nosel}}^{(i)}$. The exclusion rule therefore determines how aggressively predictors are removed from further consideration.

The complete S3VS procedure is summarized in Algorithm~\ref{alg:S3VS}. The S3VS methodology is implemented in the \textsf{R} package \texttt{S3VS} \citep{S3VS}, available from CRAN. The package provides implementations for linear, generalized linear, and survival models, including Cox proportional hazards and accelerated failure time models.

\begin{algorithm}
\caption{Structured Screen-and-Select Variable Selection (S3VS)} 
\label{alg:S3VS} \begin{algorithmic}[1]
\Require Outcome $\bfy$, predictor matrix $\bfX$, model family, leading-variable criterion and parameters, leading-set criterion and parameters, within-set selector $\mathcal{V}$ and its tuning parameters, selected-variable aggregation rule, nonselected-variable aggregation rule, update options, $m_{\max}$, $n_{\mathrm{skip}}$, and optional random-number seed
\State Set $\mathcal{A}^{(1)}\gets\{1,\ldots,p\}$, $\mathcal{M}^{(1)}\gets\varnothing$, $\bfy^{(1)}\gets\bfy$, $c_1\gets0$, $i\gets1$, and $\mathsf{terminate}\gets\mathsf{false}$
\While{$|\mathcal{M}^{(i)}|<m_{\max}$, $|\mathcal{A}^{(i)}|>0$, $c_i<n_{\mathrm{skip}}$, and $\mathsf{terminate}=\mathsf{false}$}
    \State Compute $\{A_i(r):r\in\mathcal{A}^{(i)}\}$
    \State Construct the ordered leading-index set $\mathcal{L}^{(i)}=\{\ell_1^{(i)},\ldots,\ell_{k_i}^{(i)}\}$
    \If{$k_i=0$}
        \State Set $\mathsf{terminate}\gets\mathsf{true}$
    \Else
        \For{$j=1,\ldots,k_i$}
            \State Set $\bfz_j^{(i)}\gets\bfx_{\ell_j^{(i)}}$
            \State Construct $\mathcal{S}_j^{(i)}\subseteq\mathcal{A}^{(i)}$ from $\{B_{ij}(r):r\in\mathcal{A}^{(i)}\}$
            \State Fit $\mathcal{V}$ independently to $\left(\bfy^{(i)},\bfX_{\mathcal{S}_j^{(i)}}\right)$
            \State Obtain $\left(\mathcal{S}_{j,\mathrm{sel}}^{(i)},\mathcal{S}_{j,\mathrm{nosel}}^{(i)}\right)$
        \EndFor
        \State Set $D_i\gets I\{\mathcal{S}_{1,\mathrm{sel}}^{(i)}\neq\varnothing\}$
        \If{$D_i=1$}
            \State Aggregate $\{\mathcal{S}_{j,\mathrm{sel}}^{(i)}\}_{j=1}^{k_i}$ to obtain $\mathcal{S}_{\mathrm{sel}}^{(i)}$
            \State Set $\mathcal{S}_{\mathrm{nosel}}^{(i)}\gets\varnothing$ and $c_{i+1}\gets c_i$
        \Else
            \State Set $\mathcal{S}_{\mathrm{sel}}^{(i)}\gets\varnothing$
            \State Aggregate $\{\mathcal{S}_{j,\mathrm{nosel}}^{(i)}\}_{j=1}^{k_i}$ to obtain $\mathcal{S}_{\mathrm{nosel}}^{(i)}$
            \State Set $c_{i+1}\gets c_i+1$
        \EndIf
        \State Set $\mathcal{M}^{(i+1)}\gets\mathcal{M}^{(i)}\cup\mathcal{S}_{\mathrm{sel}}^{(i)}$
        \State Set $\bfy^{(i+1)}\gets\mathcal{U}\!\left(\bfy^{(i)},\mathcal{S}_{\mathrm{sel}}^{(i)},\mathcal{S}_{\mathrm{nosel}}^{(i)},D_i\right)$ according to the model-specific update options
        \State Set $\mathcal{A}^{(i+1)}\gets\mathcal{A}^{(i)}\setminus\left(\mathcal{S}_{\mathrm{sel}}^{(i)}\cup\mathcal{S}_{\mathrm{nosel}}^{(i)}\right)$
        \State Set $i\gets i+1$
    \EndIf
\EndWhile
\State Set $i_{\mathrm{stop}}\gets i$
\State \Return the cumulative selected set $\widehat{\mathcal{M}}=\mathcal{M}^{(i_{\mathrm{stop}})}$, the iteration-specific selection history, and the runtime
\end{algorithmic} \end{algorithm}

\section{Model-Specific Implementations}
\label{sec:model-specific}

Having defined the common S3VS algorithm, we now specify the components that depend on the statistical model. Three components are model-specific: the predictor--outcome screening score $A_i(r)$, the within-leading-set selector $\mathcal{V}$, and the outcome-update operator $\mathcal{U}$. Throughout this section, $\bfX_{l\cdot}$ denotes the $l$th row of $\bfX$, and $\bfbeta$ denotes the coefficient vector for the model under consideration.

For all implementations considered here, the predictor--leader association is the absolute sample Pearson correlation 

\begin{equation} 
B_{ij}(r)=\left|R_{\bfx_r,\bfz_j^{(i)}}\right|=\left|\operatorname{Cor}\!\left(\bfx_r,\bfz_j^{(i)}\right)\right|,\qquad r\in\mathcal{A}^{(i)},\quad j=1,\ldots,k_i, 
\label{eq:leading-set-correlation} 
\end{equation} 

where $R_{\bfa,\bfb}$ denotes the sample Pearson correlation between vectors $\bfa$ and $\bfb$. Because the leader is in the current candidate set and has $B_{ij}(\ell_j^{(i)})=1$ when it is nonconstant, the relative-to-maximum leading-set rule has $B_{\max,ij}=1$ and is identical to the fixed-threshold rule with $\tau_B=\alpha_B$. It is therefore not a distinct configuration when absolute Pearson correlation is used.

\subsection{Shared Within-Set Fitting Conventions}
The selection operator $\mathcal{V}$ is fitted independently for every leading set using only the current outcome quantity and the predictors in that set. For LASSO and elastic-net selection, a prespecified mixing parameter $\alpha\in[0,1]$ is used, with $\alpha=1$ giving LASSO. Unless an analysis specifies otherwise, predictors are standardized within each local fit, and the penalty parameter is selected by minimizing the cross-validated loss. Alternatively, the one-standard-error rule may be used to choose a more parsimonious penalty. For SCAD and MCP selection, the penalty is selected using the one-standard-error rule, which chooses the most parsimonious penalty whose cross-validated loss is no more than one standard error above the minimum. For linear and binary-outcome fits, an intercept is included and is not penalized. Cox fits have no separate intercept. All centering and scaling are performed separately within each local fit. Cross-validation folds are generated separately for each local fit unless an explicit fold assignment is supplied. Unless otherwise specified, ten-fold cross-validation is used.
\subsection{Linear Models}

Consider the linear regression model 

\begin{equation} 
\bfy=\beta_0\bfone_n+\bfX\bfbeta+\bfvarepsilon,\qquad \mathrm{E}(\bfvarepsilon)=\bfzero,\qquad \mathrm{Var}(\bfvarepsilon)=\sigma^2\bfI_n. 
\label{eq:lm-implementation} 
\end{equation} 

At iteration $i$, the predictor--outcome screening score is 

\begin{equation} 
A_i(r)=\left|R_{\bfx_r,\bfy^{(i)}}\right|,\qquad r\in\mathcal{A}^{(i)}. 
\label{eq:lm-scores} 
\end{equation} 

Leading sets use the common predictor--leader score in \eqref{eq:leading-set-correlation}.

The selection operator $\mathcal{V}$ is fitted independently within each leading set using only the current outcome representation and the predictors in that set. For LASSO and elastic-net selection, a prespecified mixing parameter $\alpha\in[0,1]$ is used, with $\alpha=1$ corresponding to LASSO. Centering and scaling are performed within each local fit, with predictor standardization used unless otherwise specified. The penalty parameter is selected by minimizing cross-validated loss unless the one-standard-error rule is specified, in which case the most parsimonious penalty whose loss is no more than one standard error above the minimum is chosen. For SCAD and MCP, the one-standard-error rule is used. Linear and binary-outcome fits include an unpenalized intercept, whereas Cox fits have no separate intercept. Cross-validation folds are generated independently for each local fit unless explicit fold assignments are supplied, and ten-fold cross-validation is used by default. Analyses using explicit fold assignments should record the fold labels. 

For a completed iteration, let $\mathcal{Q}^{(i)}$ be the action set defined in \eqref{eq:iteration-action-set}. The response update is obtained by fitting an ordinary least-squares model with an intercept to the current outcome quantity and the action set: 

\begin{equation} 
\left(\widehat{\beta}_0^{(i)},\widehat{\bfbeta}_{\mathcal{Q}^{(i)}}^{(i)}\right)=\arg\min_{\beta_0,\bfbeta}\left\|\bfy^{(i)}-\beta_0\bfone_n-\bfX_{\mathcal{Q}^{(i)}}\bfbeta\right\|_2^2. 
\label{eq:linear-update-fit} 
\end{equation} 

The updated outcome is the residual vector 

\begin{equation} 
\bfy^{(i+1)}=\bfy^{(i)}-\widehat{\beta}_0^{(i)}\bfone_n-\bfX_{\mathcal{Q}^{(i)}}\widehat{\bfbeta}_{\mathcal{Q}^{(i)}}^{(i)}. 
\label{eq:update-lm} 
\end{equation} 

If $D_i=0$, the update based on the excluded set is performed only when the optional regression-out step is enabled. Otherwise, $\bfy^{(i+1)}=\bfy^{(i)}$. This update regresses the current outcome quantity only on the variables selected or excluded at iteration $i$ and does not jointly refit all predictors in $\mathcal{M}^{(i+1)}$.

\subsection{Generalized Linear Models}

Suppose that $Y_l$ follows an exponential-family distribution with conditional mean $\mu_l=\mathrm{E}(Y_l\mid\bfX_{l\cdot})$ and link function $g$, so that 

\begin{equation} 
g(\mu_l)=\beta_0+\bfX_{l\cdot}\bfbeta,\qquad l=1,\ldots,n. 
\label{eq:glm-implementation} 
\end{equation} 

For the binary implementation, the predictor--outcome screening score is the squared point-biserial correlation between the current binary outcome and predictor $r$:

\begin{equation}
A_i(r)=\widehat{\operatorname{Cor}}\!\left(\bfy^{(i)},\bfx_r\right)^2,\qquad r\in\mathcal{A}^{(i)}.
\label{eq:glm-screening-score}
\end{equation}

This scale-invariant score measures the strength of the marginal linear association. It is used only for screening and is not obtained by fitting a separate marginal logistic model. Under an unweighted logistic null model containing an intercept, the squared point-biserial correlation is proportional to the marginal logistic score statistic for testing the addition of predictor $r$. Thus, the implemented score provides the same screening ranking as a marginal logistic score test while avoiding separate marginal logistic fits. Leading sets use the common predictor--leader score in \eqref{eq:leading-set-correlation}.

Within each leading set, the binary implementation supports LASSO, elastic-net, SCAD, MCP, and nonlocal-prior (NLP) selection. LASSO and elastic-net use cross-validated penalty selection, whereas SCAD and MCP use the one-standard-error tuning rule. NLP selection is available only when a response-specific Bayesian procedure, prior specification, and decision rule have been defined for the binary outcome. For a one-predictor leading set, the LASSO, elastic-net, SCAD, and MCP branches use an unpenalized logistic regression and retain the predictor when the two-sided test of its coefficient has a $p$-value less than $0.01$.

For binary outcomes, response updating is optional. If $D_i=1$, the update is performed only when the optional selected-set regression-out step is enabled. If $D_i=0$, the update is performed only when the optional excluded-set regression-out step is enabled. Whenever an update is performed, a logistic regression with an intercept is fitted using the current binary outcome and the iteration-specific action set:

\begin{equation} 
\operatorname{logit}\!\left\{\widehat{\mu}_l^{(i)}\right\}=\widehat{\beta}_0^{(i)}+\bfX_{l,\mathcal{Q}^{(i)}}\widehat{\bfbeta}_{\mathcal{Q}^{(i)}}^{(i)},\qquad l=1,\ldots,n. 
\label{eq:binary-update-fit} 
\end{equation} 

For a prespecified threshold $h\in[0,1]$, the updated binary outcome is 

\begin{equation} 
Y_l^{(i+1)}=\begin{cases}Y_l^{(i)},&\left|Y_l^{(i)}-\widehat{\mu}_l^{(i)}\right|>h,\\\operatorname{round}\!\left(\widehat{\mu}_l^{(i)}\right),&\left|Y_l^{(i)}-\widehat{\mu}_l^{(i)}\right|\leq h,\end{cases}\qquad l=1,\ldots,n. 
\label{eq:glm-update} 
\end{equation} 

The simulations use $h=0.5$. With a binary response, this value leaves the response unchanged except in the theoretically possible exact-tie case $\widehat{\mu}_l^{(i)}=0.5$ under the rounding convention used by \textsf{R}. Consequently, the implemented logistic procedure retains a valid binary response across iterations while the candidate predictor set changes. If the relevant updating option is false, then $\bfy^{(i+1)}=\bfy^{(i)}$. The logistic update contains only the current action set and is not a joint refit of the cumulative selected set. Intercepts and action-set coefficients are re-estimated at every invoked update, and no fitted coefficients or working weights are carried into the next iteration.

The general GLM notation permits other exponential-family responses only after a response-specific screening score, within-set selector, and update rule have been declared. The simulation studies in this paper focus on binary logistic regression.

\subsection{Survival Models}

For subject $l$, let $T_l$ and $C_l$ denote the event and censoring times, respectively, and define the observed time and event indicator by $Y_l=\min(T_l,C_l)$ and $\delta_l=I(T_l\leq C_l)$, $l=1,\ldots,n$. Thus, for survival outcomes, $\bfy=(\bft,\bfd)$ with $\bft=(Y_1,\ldots,Y_n)^T$ and $\bfd=(\delta_1,\ldots,\delta_n)^T$. The observed times and event indicators are retained as the outcome representation throughout S3VS, while the model-specific screening score may condition on the cumulative selected set $\mathcal{M}^{(i)}$.

\subsubsection{Accelerated Failure Time Models}

The accelerated failure time (AFT) model assumes 

\begin{equation} 
\log(T_l)=\mu+\bfX_{l\cdot}\bfbeta+\varepsilon_l,\qquad l=1,\ldots,n, 
\label{eq:aft-implementation} 
\end{equation} 

where $\mu$ is an intercept and $\varepsilon_l$ is an error term. The first S3VS iteration uses a marginal utility score $\mathrm{MU}_r$, whereas subsequent iterations use a conditional utility score $\mathrm{CU}_r^{(i)}$ that measures the additional contribution of predictor $r$ after adjustment for the cumulative selected set. This iterative construction was considered in preliminary form for accelerated failure time models in \citet{sanyal_2023}. Thus, 

\begin{equation} 
A_i(r)=\begin{cases}\mathrm{MU}_r,&i=1,\\\mathrm{CU}_r^{(i)},&i\geq2,\end{cases}\qquad r\in\mathcal{A}^{(i)}. 
\label{eq:aft-screening-score} 
\end{equation} 

The AFT screening utility is based on the maximized likelihood under a Weibull accelerated failure-time model, with observed events contributing their event likelihoods and right-censored observations contributing their survival probabilities. At the first iteration, each candidate predictor is evaluated marginally. At later iterations, each candidate is evaluated conditionally after adjustment for the cumulative selected predictors $\mathcal{M}^{(i)}$. Thus, $\mathrm{MU}_r$ and $\mathrm{CU}_r^{(i)}$ are marginal and conditional log-likelihood-based utility scores, respectively. Leading sets use the common predictor--leader score in \eqref{eq:leading-set-correlation}. Within each leading set, variable selection may be performed using bridge penalization \citep{huang_ma_2010}, AFTGEE \citep{aftgee}, PVAFT \citep{afthd}, or another explicitly declared AFT procedure, with its tuning and decision rule fixed before fitting. After aggregation, the survival outcome representation remains unchanged, and the cumulative selected set is used in the conditional-utility calculation at the next iteration.

\subsubsection{Cox Proportional Hazards Models}

The Cox proportional hazards model specifies 

\begin{equation} 
\lambda(t\mid\bfX_{l\cdot})=\lambda_0(t)\exp\!\left(\bfX_{l\cdot}\bfbeta\right), 
\label{eq:cox-implementation} 
\end{equation} 

where $\lambda_0(t)$ is the baseline hazard function. For candidate predictor $r$, let $\mathcal{L}_r(\beta_r)$ denote the marginal Cox partial likelihood, let $\widehat{\beta}_r$ be its maximizer, and define the marginal deviance $D_r^{\mathrm{marg}}=-2\log\mathcal{L}_r(\widehat{\beta}_r)$. If $D_0$ denotes the null-model deviance, the first-iteration screening score is the marginal deviance reduction 

\begin{equation} 
A_1(r)=D_0-D_r^{\mathrm{marg}}. 
\label{eq:cox-marginal-score} 
\end{equation} 

For iteration $i\geq2$, let $D_{\mathrm{base}}^{(i)}$ denote the deviance of the Cox model containing the cumulative selected predictors $\mathcal{M}^{(i)}$ and let $D_{\mathrm{new},r}^{(i)}$ denote the deviance after adding candidate predictor $r$. The conditional screening score is 

\begin{equation} 
A_i(r)=\mathrm{CU}_r^{(i)}=D_{\mathrm{base}}^{(i)}-D_{\mathrm{new},r}^{(i)}. 
\label{eq:cox-conditional-score} 
\end{equation} 

Thus, larger values of $A_i(r)$ indicate greater incremental improvement in model fit. All deviances used in a comparison are computed from the corresponding Cox partial-likelihood fits on the same survival data, and a failed or nonfinite fit yields an ineligible score.

Within each leading set, the Cox implementation uses the declared Cox-model selection procedure. For Cox LASSO or elastic-net selection, the local fit has no separate intercept and follows the shared cross-validation and scaling conventions, with the penalty choice recorded for the analysis. If a Cox leading set contains only one predictor, the predictor is retained when its Cox-model Wald test has $p$-value less than $0.01$. The survival outcome representation is not modified after aggregation: 

\begin{equation} 
\bfy^{(i+1)}=\bfy^{(i)}=(\bft,\bfd). 
\label{eq:update-survival} 
\end{equation} 

Instead, $\mathcal{M}^{(i+1)}$ enters the subsequent conditional screening calculations. Cumulative selected variables are not automatically included in each within-leading-set selection fit. The settings used for the ovarian-cancer analysis and its SIS--Cox--LASSO comparator are given in Section~\ref{sec:real-data-analysis}.

\section{Asymptotic Analysis}
\label{sec:asymptotic-analysis}

We now study a simplified one-step linear S3VS operator to isolate the role of proxy variables. The analysis covers top-$d_L$ leading-variable screening, top-$m_n$ absolute-correlation leading sets, a within-set selector with a stated conditional sure-screening property, and an aggregation rule that preserves active variables selected in relevant leading sets. The theory does not cover later iterations, outcome updates, data-dependent stopping, irreversible exclusion, cross-validated tuning, conservative aggregation, or the GLM, AFT, and Cox implementations. 

To state this one-step result precisely, we first specify the model, the screening score, and the one-step operator.
We work at the first iteration, so no selection-history conditioning is used. Relabel the initial candidate indices as $\{1,\ldots,p\}$ and write $\bfX=(\bfx_1,\ldots,\bfx_p)\in\mathbb{R}^{n\times p}$ and $\bfy=\bfy^{(1)}$. The rows of $\bfX$ are i.i.d. copies of $\bfx^T$, and we assume that the variables have been centered before considering the linear model

\begin{equation} 
\bfy=\bfX\bfbeta+\bfvarepsilon,\qquad \bfvarepsilon\sim N(\bfzero,\sigma^2\bfI_n),\qquad \bfvarepsilon\ \text{independent of}\ \bfX. 
\label{eq:model} 
\end{equation} 

The predictors have zero population means and unit population variances. The theoretical screening score is the sample covariance score $\widehat\omega_r=n^{-1}\bfx_r^T\bfy$, with population counterpart $\omega_r^0=\operatorname{Cov}(X_r,Y)$. When the predictor columns are empirically standardized, ranking $|\widehat\omega_r|$ is equivalent to ranking the absolute sample correlations used by the linear implementation. Otherwise, the results below apply to the covariance-score version of the first screening step.

Define the active set and its size by 

\begin{equation} 
\mathcal{M}_*=\{r\in\{1,\ldots,p\}:\beta_r\neq0\},\qquad s=|\mathcal{M}_*|. 
\label{eq:active-set} 
\end{equation} 

All sets in this section contain original predictor indices. The theoretical one-step operator is defined as follows. Let $\mathcal{L}^{(1)}=\operatorname{Top}_{d_L}\{|\widehat\omega_r|:1\le r\le p\}$, with deterministic tie-breaking, and write its ordered indices as $\ell_1,\ldots,\ell_{d_L}$. For each leader, define 

\begin{equation} 
\mathcal{S}_r^{(1)}=\{\ell_r\}\cup\operatorname{Top}_{m_n-1}\{|\widehat\rho_{\ell_r k}|:k\in\{1,\ldots,p\}\setminus\{\ell_r\}\},\qquad r=1,\ldots,d_L, 
\label{eq:one-step-leading-set} 
\end{equation} 

where $1\le m_n\le p$ and ties are broken deterministically. The operator applies $\mathcal V$ independently to each leading set and applies an aggregation map to the within-set selected variables. The aggregation map is required to satisfy the active-preserving property in \eqref{eq:active-preserving}. No outcome update, candidate-set update, stopping rule, or exclusion operation is applied in this one-step analysis.

For clarity, $\operatorname{Top}_k$ denotes the indices of the $k$ largest listed values, $\operatorname{Top}_0=\varnothing$, and all ties are broken deterministically.

For each $r$, write the within-set output as 

\begin{equation} 
\mathcal V\!\left(\bfy,\bfX_{\mathcal S_r^{(1)}}\right)=\left(\mathcal S_{r,\mathrm{sel}}^{(1)},\mathcal S_{r,\mathrm{nosel}}^{(1)}\right),\qquad \mathcal S_{r,\mathrm{sel}}^{(1)}\subseteq\mathcal S_r^{(1)},\qquad \mathcal S_{r,\mathrm{nosel}}^{(1)}=\mathcal S_r^{(1)}\setminus\mathcal S_{r,\mathrm{sel}}^{(1)}. 
\label{eq:one-step-selector-output} 
\end{equation}

We next state the assumptions that connect marginal proxy screening, correlation-based leading-set coverage, within-set retention, and aggregation.

\subsection{Assumptions}
\label{sec:asymptotic-assumptions}

\subsubsection{Dimensionality and marginal-score premise}

\begin{assumption}[Dimensionality]
\label{ass:dim}
The dimensionality satisfies $p>n$ and $\log p=O(n^\xi)$ for some $\xi\in(0,1)$. The exponent $\tau\geq0$ below is the covariance-growth exponent appearing in a Fan--Lv-type marginal-screening condition.
\end{assumption}

\begin{assumption}[Fan--Lv-type marginal-score condition]
\label{ass:marginal}
For every deterministic target set $\mathcal{T}_n\subseteq\{1,\ldots,p\}$ satisfying $|\mathcal{T}_n|\le d_n$ and $\min_{j\in\mathcal{T}_n}|\operatorname{Cov}(X_j,Y)|\geq c_A n^{-\kappa_A}$ for constants $c_A>0$ and $\kappa_A\geq0$, and for every $d_n\asymp n^{\theta_A}$ with $2\kappa_A+\tau<\theta_A<1$ and $\xi<1-2\kappa_A$, the top-$d_n$ covariance-score rule satisfies, for constants $C_A,C_A'>0$, 

\begin{equation} 
\mathbb{P}\left\{\mathcal{T}_n\subseteq\operatorname{Top}_{d_n}\{|\widehat\omega_r|:1\leq r\leq p\}\right\}\geq1-C_A\exp\left\{-C_A'\frac{n^{1-2\kappa_A}}{\log n}\right\}. \label{eq:marginal-score-condition} 
\end{equation} 

This condition is the marginal-screening input to the present composition theorem. It is a Fan--Lv-type condition rather than a claim about the later S3VS iterations; verifying it for a particular design and score is not delegated to conditioning on a selection history.
\end{assumption}

The preceding condition provides the generic marginal-screening input. We next identify the proxy set that will serve as the target set in that condition.

\subsubsection{Proxy marginal signals}

For $j=1,\ldots,p$, let $\bfe_j\in\mathbb{R}^p$ denote the $j$th standard basis vector. The population marginal score is 

\begin{equation} 
\omega_j^0=\operatorname{Cov}(X_j,Y)=\bfe_j^T\bfSigma\bfbeta, \label{eq:population-score} 
\end{equation} 

where $\bfSigma=\operatorname{Cov}(\bfx)$. The active predictors need not all have strong marginal scores.

\begin{assumption}[Proxy marginal strength]
\label{ass:proxy}
There exists a deterministic map $q:\mathcal{M}_*\longrightarrow\{1,\ldots,p\}$ with image $\mathcal{P}_*=q(\mathcal{M}_*)$ and constants $c_L>0$ and $\kappa_L\geq0$ such that 

\begin{equation} 
\min_{\ell\in\mathcal{P}_*}|\operatorname{Cov}(X_\ell,Y)|\geq c_L n^{-\kappa_L}. \label{eq:proxy-strength} 
\end{equation} 

For each $j\in\mathcal{M}_*$, $X_{q(j)}$ is called a proxy leading variable for $X_j$. The proxy may coincide with the active predictor, $q(j)=j$, and the map need not be one-to-one.
\end{assumption}

If $q(j)=j$ for every $j\in\mathcal{M}_*$, one may take $\mathcal{P}_*=\mathcal{M}_*$ and $\kappa_L=\kappa_M$, where $\kappa_M$ denotes the corresponding SIS marginal-signal exponent, so the leading-variable step has no signal-strength advantage over SIS. The proxy mechanism considered here is the regime in which $\kappa_L<\kappa_M$ may occur.

Marginal strength places proxy variables among the leading variables, but it does not by itself place the corresponding active predictors in their leading sets. The next assumptions impose the predictor--predictor correlation conditions needed for that coverage.

\subsubsection{Uniform sample-correlation control and leading-set coverage}

For $j,k\in\{1,\ldots,p\}$, define the population and sample correlations by 

\begin{equation} 
\rho_{jk}=\operatorname{Corr}(X_j,X_k),\qquad \widehat\rho_{jk}=\operatorname{Corr}_n(\bfx_j,\bfx_k), \label{eq:correlations} 
\end{equation} 

where $\operatorname{Corr}_n$ denotes the usual sample correlation computed after empirical centering.

\begin{assumption}[Sub-Gaussian predictors]
\label{ass:subg}
The coordinates of $\bfx$ are uniformly sub-Gaussian, namely, for some constant $K<\infty$, 

\begin{equation} 
\max_{1\leq j\leq p}\|X_j\|_{\psi_2}\leq K,\qquad \|U\|_{\psi_2}=\inf\left\{t>0:\mathbb{E}\left[\exp(U^2/t^2)\right]\leq2\right\}. \label{eq:subgaussian} 
\end{equation}

\end{assumption}

\begin{assumption}[Correlation separation for leading-set inclusion]
\label{ass:gap}
There exists a deterministic sequence $a_n\downarrow0$ satisfying $a_n\sqrt{n/\log p}\longrightarrow\infty$ such that, for every $j\in\mathcal{M}_*$ with $q(j)\neq j$, at most $m_n-1$ indices $k\neq j$ satisfy 

\begin{equation} 
|\rho_{k,q(j)}|\geq|\rho_{j,q(j)}|-2a_n. \label{eq:gap} 
\end{equation} 

The leading-set size satisfies $1\le m_n\le p$ and $m_n=O(n^\alpha)$ for some $\alpha\in[0,1)$.
\end{assumption}

Assumption~\ref{ass:gap} places each active predictor among the $m_n$ predictors most strongly associated with its proxy at the population level, with a margin sufficient to withstand uniform sample-correlation error. No separation condition is needed when $q(j)=j$, because a leader is included in its own leading set by construction.

Once a proxy is selected as a leader and its corresponding active predictor is covered by the leading set, the within-set selector must retain that active predictor. We now state the required selector and aggregation conditions.

\subsubsection{Within-set selector and aggregation}

Let $\mathfrak{S}^{(1)}=\{\mathcal{S}_1^{(1)},\ldots,\mathcal{S}_{d_L}^{(1)}\}$ denote the random family of leading sets generated by the one-step operator.

\begin{assumption}[Conditional setwise sure-screening property]
\label{ass:selector}
There exist constants $C_V,c_V>0$ and $\zeta>0$ such that, almost surely, for every generated leading set $\mathcal{S}_r^{(1)}$ satisfying $\mathcal{S}_r^{(1)}\cap\mathcal{M}_*\neq\varnothing$, 

\begin{equation} 
\mathbb{P}\left\{\mathcal{S}_r^{(1)}\cap\mathcal{M}_*\not\subseteq\mathcal{S}_{r,\mathrm{sel}}^{(1)}\,\middle|\,\mathfrak{S}^{(1)}\right\}\leq C_V\exp(-c_Vn^\zeta). \label{eq:selector} 
\end{equation}

\end{assumption}

The conditional formulation in Assumption~\ref{ass:selector} is a strong uniform requirement over the data-dependent family of leading sets. It is not implied merely by conditioning on a selection history and is not claimed here for cross-validated LASSO without additional assumptions such as sample splitting, uniform design control, and a suitable signal condition.

The one-step aggregation map is called active preserving if, for every predictor index $v$, 

\begin{equation} 
v\in\bigcup_{j=1}^{d_L}\mathcal{S}_{j,\mathrm{sel}}^{(1)}\quad\Longrightarrow\quad v\in\mathcal{S}_{\mathrm{sel}}^{(1)}\ \text{and}\ v\notin\mathcal{S}_{\mathrm{nosel}}^{(1)}. \label{eq:active-preserving} 
\end{equation} 

The liberal selected-variable union 

\begin{equation} 
\mathcal{S}_{\mathrm{sel}}^{(1)}=\bigcup_{j=1}^{d_L}\mathcal{S}_{j,\mathrm{sel}}^{(1)} \label{eq:liberal-union} 
\end{equation} 

satisfies this property when no nonselected aggregate is allowed to override a selected variable. The theorem below assumes this property directly. It does not assert that the first-leading-set branch of the full iterative algorithm is active-preserving.

The assumptions above provide the four ingredients of the one-step argument: marginal screening reaches the proxy set, correlation-based leading sets cover the corresponding active predictors, within-set selection retains them, and aggregation preserves them. The resulting sure-screening bound is stated next.

\subsection{Sure-Screening Result}
\label{sec:s3vs-sure-screening}

\begin{theorem}[One-step S3VS sure screening under proxy coverage]
\label{thm:s3vs}
Suppose Assumptions~\ref{ass:dim}, \ref{ass:marginal}, \ref{ass:proxy}, \ref{ass:subg}, \ref{ass:gap}, and \ref{ass:selector} hold. Assume 

\begin{equation} 
\xi<1-2\kappa_L,\qquad 2\kappa_L+\tau<\theta_L<1,\qquad d_L\asymp n^{\theta_L},\qquad |\mathcal{P}_*|\le d_L. \label{eq:s3vs-rate-conditions} 
\end{equation} 

Let the one-step operator use the top-$d_L$ marginal-score rule for leading variables, the leading sets in \eqref{eq:one-step-leading-set}, a within-set selector satisfying Assumption~\ref{ass:selector}, and an aggregation map satisfying \eqref{eq:active-preserving}. Then there exist constants $C_1,C_2,C_3,c_2>0$ such that 

\begin{equation} 
\mathbb{P}\left\{\mathcal{M}_*\subseteq\mathcal{S}_{\mathrm{sel}}^{(1)}\right\}\geq1-C_1\exp\left\{-C_2\frac{n^{1-2\kappa_L}}{\log n}\right\}-C_3p^2\exp(-c_2na_n^2)-C_Vd_L\exp(-c_Vn^\zeta). \label{eq:s3vs-bound} 
\end{equation} 

Under Assumption~\ref{ass:gap}, $p^2\exp(-c_2na_n^2)\longrightarrow0$. If also $d_L\exp(-c_Vn^\zeta)\longrightarrow0$, then 

\begin{equation} 
\mathbb{P}\left\{\mathcal{M}_*\subseteq\mathcal{S}_{\mathrm{sel}}^{(1)}\right\}\longrightarrow1. \label{eq:s3vs-consistency} 
\end{equation} 

If the two remainder terms in \eqref{eq:s3vs-bound} are $o\!\left(\exp\{-C_2n^{1-2\kappa_L}/\log n\}\right)$, then the leading term in this derived one-step upper bound is 

\begin{equation} 
O\!\left(\exp\left\{-C_2\frac{n^{1-2\kappa_L}}{\log n}\right\}\right). \label{eq:s3vs-leading-rate} 
\end{equation}

\end{theorem}

Theorem~\ref{thm:s3vs} combines these sufficient conditions to give a one-step sure-screening result. Its conclusion applies only to the one-step operator defined above and does not automatically extend to later S3VS iterations, residualized outcomes, conservative aggregation, exclusion branches, or model-specific implementations for which the assumptions have not been established.

\subsection{Comparison with SIS}
The following corollaries compare the one-step S3VS bounds with analogous SIS quantities under additional assumptions. They are not guarantees for the full iterative algorithm. The first comparison concerns only the displayed upper bounds and does not compare the unknown true failure probabilities unless matching lower bounds are also available.

\begin{corollary}[Comparison of derived failure-probability upper bounds]
\label{cor:prob}
Suppose ordinary SIS has a derived upper-bound term of the form 

\begin{equation} 
\exp\left\{-C_M\frac{n^{1-2\kappa_M}}{\log n}\right\},\qquad C_M>0, \label{eq:sis-failure-bound} 
\end{equation} 

with $\kappa_M>\kappa_L$, and suppose the conditions of Theorem~\ref{thm:s3vs} hold with the remainder terms in \eqref{eq:s3vs-bound} negligible relative to its first term. Then, for every fixed $C_L,C_M>0$, 

\begin{equation} 
\exp\left\{-C_L\frac{n^{1-2\kappa_L}}{\log n}\right\}=o\!\left(\exp\left\{-C_M\frac{n^{1-2\kappa_M}}{\log n}\right\}\right). \label{eq:probability-bound-improvement} 
\end{equation} 

Thus, under the stated premises, the leading term in the one-step S3VS upper bound decays faster than the corresponding leading term in the SIS upper bound. This is not a claim that the actual S3VS failure probability is asymptotically smaller than the actual SIS failure probability.
\end{corollary}

The preceding corollary compares the leading failure-probability terms. A separate consequence concerns the dimension of the candidate pool passed to the post-screening selector.

\begin{corollary}[Post-screening candidate-pool comparison]
\label{cor:dim}
Let ordinary SIS retain a candidate model of size $d_{\mathrm{SIS}}\asymp n^{\theta_M}$ with $2\kappa_M+\tau<\theta_M<1$. Let the one-step S3VS operator use $d_L\asymp n^{\theta_L}$ leading variables and leading sets of size at most $m_n=O(n^\alpha)$, and define 

\begin{equation} 
\mathcal{S}_{\mathrm{union}}^{(1)}=\bigcup_{j=1}^{d_L}\mathcal{S}_j^{(1)},\qquad D_{\mathrm{S3VS}}=|\mathcal{S}_{\mathrm{union}}^{(1)}|. \label{eq:s3vs-candidate-union} 
\end{equation} 

Then 

\begin{equation} 
D_{\mathrm{S3VS}}\leq d_Lm_n=O(n^{\theta_L+\alpha}). \label{eq:s3vs-candidate-bound} 
\end{equation} 

Consequently, $D_{\mathrm{S3VS}}=o(d_{\mathrm{SIS}})$ whenever $\theta_L+\alpha<\theta_M$. An admissible $\theta_L$ satisfying this inequality exists whenever 

\begin{equation} 
2\kappa_L+\tau+\alpha<\theta_M. \label{eq:dim-improve-existence} 
\end{equation} 

If the SIS and S3VS exponents are chosen with the same positive slack above their respective lower boundaries, the sufficient condition reduces to 

\begin{equation} 
\alpha<2(\kappa_M-\kappa_L). \label{eq:alpha-cond} 
\end{equation} 

Moreover, each individual one-step S3VS selection problem has dimension at most $m_n$, so $m_n=o(d_{\mathrm{SIS}})$ whenever $\alpha<\theta_M$. These are deterministic dimension comparisons. They do not establish lower runtime, lower memory use, or lower computational complexity.
\end{corollary}

A smaller candidate pool can also reduce the logarithmic dimension factor in a subsequent estimation bound. The next corollary states this implication conditionally, without deriving the post-screening estimator's rate.

\begin{corollary}[Conditional post-screening estimation implication]
\label{cor:est}
Let $q_n$ and $d_{\mathrm{SIS}}$ be deterministic candidate-model-size sequences with $2\leq q_n<d_{\mathrm{SIS}}<n$. Suppose the same post-screening estimator, applied after the one-step S3VS and SIS procedures, satisfies the respective assumed bounds 

\begin{equation} 
\|\widehat{\bfbeta}_{\mathrm{S3VS}}-\bfbeta\|_2=O_p\!\left(\sqrt{\frac{s\log q_n}{n}}\right),\qquad \|\widehat{\bfbeta}_{\mathrm{SIS}}-\bfbeta\|_2=O_p\!\left(\sqrt{\frac{s\log d_{\mathrm{SIS}}}{n}}\right). \label{eq:conditional-estimation-rates} 
\end{equation} 

If $q_n=n^{a+o(1)}$ and $d_{\mathrm{SIS}}=n^{b+o(1)}$ for constants $0<a<b$, then 

\begin{equation} 
\frac{\log q_n}{\log d_{\mathrm{SIS}}}\longrightarrow\frac{a}{b}<1. \label{eq:log-constant-improvement} 
\end{equation} 

In this polynomial case, the ratio of the square-root logarithmic factors in the two assumed rates tends to $\sqrt{a/b}$, a constant rather than an order improvement. If $\log q_n=o(\log d_{\mathrm{SIS}})$, the logarithmic factor has a smaller asymptotic order for the S3VS bound. This corollary is only a conditional comparison of assumed post-screening bounds. It neither derives those bounds from Theorem~\ref{thm:s3vs} nor establishes an estimator-specific S3VS rate.
\end{corollary}

These corollaries describe potential gains under additional assumptions rather than uniform advantages. The following proposition gives a design in which proxy structure provides no exponent advantage over SIS.

\subsection{Limits of Uniform Improvement}
\begin{proposition}[A design with no proxy-signal advantage]
\label{prop:no-free}
There exist linear-model designs for which no proxy set has a marginal-signal exponent smaller than that of the active predictors. For such designs, the proxy mechanism provides no structural basis for a uniformly stronger one-step sure-screening bound than SIS.
\end{proposition}

Together, these results show that the benefit of the one-step S3VS construction depends on both proxy coverage and the relationship between marginal signal and predictor dependence. The simulation studies that follow examine how these conditions and trade-offs appear in finite samples.

\subsection{Scope of the Theoretical Claims}
The theory gives a one-step result: proxy screening, correlation-based coverage, within-set retention, and active-preserving aggregation together imply retention of the active set. It does not show that these conditions hold for the full iterative procedure, nor does it give a guarantee after data-dependent outcome updates or stopping. The SIS and candidate-pool comparisons are, respectively, comparisons of derived upper-bound expressions and deterministic dimension bounds. No claim of uniform actual-probability improvement, runtime improvement, or prediction improvement is made.

\section{Simulation Study}
\label{sec:simulation}

\subsection{Linear Models}
\label{sec:simulation-lm}

\subsubsection{Simulation Design}
The linear-model study crossed sample sizes $n\in\{200,500,800\}$, dimensions $p\in\{10{,}000, 15{,}000, 20{,}000\}$, the 13 named design scenarios in Supplementary Table~S2
 and weak, moderate, and strong signal regimes with target signal-to-noise ratios (SNRs) $0.25$, $1$, and $4$, respectively. The reported results use 100 independently generated replicates for every combination of scenario, signal regime, $n$, and $p$, giving 35,100 replicates in total. Each replicate generates a new training design, an independent test design, a coefficient vector, and training and test responses. Thus, the simulation does not condition on one design matrix or one coefficient vector while varying only the response noise.

The design scenarios include independent predictors, unequal positive-correlation blocks, four prespecified proxy/no-proxy configurations, AR(1) dependence, latent factors, a sparse pairwise graph, overlapping factor modules, signed factor loadings with negative correlations, heavy-tailed predictors, and row/feature contamination. Each replicate contains 25 active coordinates. Their placement is modified in the proxy scenarios to create cancellation, shared-proxy, null-proxy, and misleading-null cases. Coefficients are generated with random signs and magnitudes and then rescaled so that the empirical training SNR, $\operatorname{var}(\bfX_{\rm train}\bfbeta)/\sigma^2$, equals the target for the selected signal regime, with $\sigma^2=1$. The calibration in Supplementary Figure~S1
 and Supplementary Table~S3
 confirms that the realized SNRs agree with these targets.

\subsubsection{Implementation and Performance Measures}
We compare 12 procedures. For each selector--LASSO, SCAD, MCP, and NLP--we compare full S3VS, the selection from the first S3VS iteration (labeled S3VS-one-step), and one-pass SIS followed by the same selector. The SIS procedures are one-pass benchmarks rather than ISIS. Full S3VS uses one leading variable, leading sets of size five, liberal aggregation, conservative-from-the-beginning removal, the pre-iteration selected-model-size limit $m_{\max}=100$, and the cumulative unsuccessful-iteration limit $n_{\mathrm{skip}}=3$. Complete method definitions are given in Supplementary Table~S1.
Because the S3VS-one-step results are taken from the first iteration of the full run, they show what is gained by later iterations and require no separate S3VS fit. Only full S3VS and SIS are timed separately.

Each selected set is evaluated against the replicate-specific active set using TPR, sure-screening probability, conditional FDP among nonempty selections, unconditional mean FDP with empty selections assigned FDP zero, precision, false-positive and false-negative counts, and F-score. We also report the probability of a nonempty selection, mean and median selected-model size, coefficient error, and independent-test MSE. Coefficient error and test MSE use a common post-selection LASSO refit on the selected variables, so these two measures compare selected sets rather than the four different selectors. All table entries are means with Monte Carlo standard errors in parentheses. The main figures show replicate means or distributions as indicated in their captions.

\begin{table}
\centering
\caption{Linear-model selection and prediction performance at $n=800$ and $p=20{,}000$. Entries are means with Monte Carlo standard errors (MCSEs) in parentheses, pooled over the named scenarios. Reps. is the number of replicates, $\lvert S\rvert$ is selected-model size, $P(\mathrm{nonempty})$ is the probability of a nonempty selection, and Sure is the probability of selecting all active variables. FDP $|$ nonempty is conditional on a nonempty selection, whereas Mean FDP is unconditional, with empty selections assigned FDP zero. TPR, FP, FN, MSE, and MCSE denote true positive rate, false positives, false negatives, mean squared error, and Monte Carlo standard error, respectively.}
\label{tab:LM-primary-performance}
{\fontsize{5.4}{6.4}\selectfont
\setlength{\tabcolsep}{1pt}
\begin{tabular}{lrrrrrrrrrrrrrr}
\toprule
Method & Signal & Reps. & \shortstack{$P$\\$(\mathrm{nonempty})$} & Mean $\lvert S\rvert$ & Med. $\lvert S\rvert$ & TPR & Sure & \shortstack{FDP $|$\\nonempty} & \shortstack{Mean\\FDP} & Precision & FP & FN & F-score & \shortstack{Test\\MSE} \\
\midrule
S3VS-full-LASSO & moderate & 1300 & 1.000(0.000) & 100.997(0.031) & 101.000(0.000) & 0.652(0.005) & 0.003(0.002) & 0.839(0.001) & 0.839(0.001) & 0.161(0.001) & 84.701(0.141) & 8.704(0.137) & 0.259(0.002) & 1.548(0.006) \\
S3VS-full-MCP & moderate & 1300 & 0.798(0.011) & 2.077(0.051) & 2.000(0.000) & 0.079(0.002) & 0.000(0.000) & 0.046(0.005) & 0.036(0.004) & 0.954(0.005) & 0.105(0.012) & 23.028(0.050) & 0.138(0.003) & 1.805(0.008) \\
S3VS-full-NLP & moderate & 1300 & 0.998(0.001) & 21.098(0.225) & 22.000(0.422) & 0.580(0.007) & 0.001(0.001) & 0.317(0.004) & 0.317(0.004) & 0.683(0.004) & 6.601(0.105) & 10.503(0.164) & 0.603(0.005) & 1.287(0.004) \\
S3VS-full-SCAD & moderate & 1300 & 0.798(0.011) & 2.613(0.072) & 2.000(0.000) & 0.092(0.003) & 0.000(0.000) & 0.089(0.007) & 0.071(0.005) & 0.911(0.007) & 0.302(0.024) & 22.688(0.069) & 0.154(0.004) & 1.792(0.008) \\
S3VS-one-step-LASSO & moderate & 1300 & 1.000(0.000) & 1.000(0.000) & 1.000(0.000) & 0.039(0.000) & 0.000(0.000) & 0.026(0.004) & 0.026(0.004) & 0.974(0.004) & 0.026(0.004) & 24.026(0.004) & 0.075(0.000) & 2.004(0.003) \\
S3VS-one-step-MCP & moderate & 1300 & 0.785(0.011) & 0.785(0.011) & 1.000(0.000) & 0.030(0.000) & 0.000(0.000) & 0.030(0.005) & 0.024(0.004) & 0.970(0.005) & 0.024(0.004) & 24.238(0.012) & 0.059(0.001) & 2.004(0.003) \\
S3VS-one-step-NLP & moderate & 1300 & 0.998(0.001) & 0.998(0.001) & 1.000(0.000) & 0.039(0.000) & 0.000(0.000) & 0.025(0.004) & 0.025(0.004) & 0.975(0.004) & 0.025(0.004) & 24.026(0.004) & 0.075(0.000) & 2.004(0.003) \\
S3VS-one-step-SCAD & moderate & 1300 & 0.785(0.011) & 0.785(0.011) & 1.000(0.000) & 0.030(0.000) & 0.000(0.000) & 0.030(0.005) & 0.024(0.004) & 0.970(0.005) & 0.024(0.004) & 24.239(0.012) & 0.059(0.001) & 2.004(0.003) \\
SIS-LASSO & moderate & 1300 & 1.000(0.000) & 131.826(1.864) & 175.000(0.826) & 0.605(0.007) & 0.003(0.002) & 0.840(0.004) & 0.840(0.004) & 0.160(0.004) & 116.712(1.711) & 9.885(0.177) & 0.204(0.002) & 1.586(0.005) \\
SIS-MCP & moderate & 1300 & 1.000(0.000) & 112.479(1.879) & 152.500(1.322) & 0.566(0.008) & 0.000(0.000) & 0.756(0.007) & 0.756(0.007) & 0.244(0.007) & 98.328(1.708) & 10.848(0.193) & 0.220(0.002) & 1.601(0.005) \\
SIS-NLP & moderate & 1300 & 0.998(0.001) & 15.830(0.254) & 18.000(0.536) & 0.494(0.008) & 0.000(0.000) & 0.197(0.004) & 0.196(0.004) & 0.803(0.004) & 3.479(0.094) & 12.649(0.189) & 0.557(0.007) & 1.348(0.006) \\
SIS-SCAD & moderate & 1300 & 1.000(0.000) & 122.432(1.810) & 162.000(1.425) & 0.581(0.007) & 0.002(0.001) & 0.837(0.004) & 0.837(0.004) & 0.163(0.004) & 107.902(1.649) & 10.470(0.185) & 0.206(0.002) & 1.592(0.005) \\
S3VS-full-LASSO & strong & 1300 & 1.000(0.000) & 101.004(0.031) & 101.000(0.000) & 0.898(0.006) & 0.505(0.014) & 0.778(0.001) & 0.778(0.001) & 0.222(0.001) & 78.542(0.154) & 2.538(0.150) & 0.357(0.002) & 1.434(0.013) \\
S3VS-full-MCP & strong & 1300 & 0.997(0.002) & 13.639(0.182) & 14.000(0.489) & 0.539(0.007) & 0.007(0.002) & 0.025(0.003) & 0.025(0.003) & 0.975(0.003) & 0.158(0.015) & 11.518(0.187) & 0.655(0.007) & 1.939(0.023) \\
S3VS-full-NLP & strong & 1300 & 1.000(0.000) & 29.225(0.217) & 31.000(0.374) & 0.869(0.007) & 0.508(0.014) & 0.264(0.004) & 0.264(0.004) & 0.736(0.004) & 7.505(0.103) & 3.280(0.166) & 0.781(0.005) & 1.266(0.010) \\
S3VS-full-SCAD & strong & 1300 & 0.998(0.001) & 14.417(0.180) & 16.000(0.389) & 0.559(0.008) & 0.006(0.002) & 0.051(0.004) & 0.051(0.004) & 0.949(0.004) & 0.430(0.031) & 11.013(0.191) & 0.668(0.007) & 1.915(0.024) \\
S3VS-one-step-LASSO & strong & 1300 & 1.000(0.000) & 1.000(0.000) & 1.000(0.000) & 0.040(0.000) & 0.000(0.000) & 0.002(0.001) & 0.002(0.001) & 0.998(0.001) & 0.002(0.001) & 24.002(0.001) & 0.077(0.000) & 4.992(0.009) \\
S3VS-one-step-MCP & strong & 1300 & 0.997(0.002) & 0.997(0.002) & 1.000(0.000) & 0.040(0.000) & 0.000(0.000) & 0.002(0.001) & 0.002(0.001) & 0.998(0.001) & 0.002(0.001) & 24.005(0.002) & 0.077(0.000) & 4.992(0.009) \\
S3VS-one-step-NLP & strong & 1300 & 1.000(0.000) & 1.000(0.000) & 1.000(0.000) & 0.040(0.000) & 0.000(0.000) & 0.002(0.001) & 0.002(0.001) & 0.998(0.001) & 0.002(0.001) & 24.002(0.001) & 0.077(0.000) & 4.992(0.009) \\
S3VS-one-step-SCAD & strong & 1300 & 0.998(0.001) & 0.998(0.001) & 1.000(0.000) & 0.040(0.000) & 0.000(0.000) & 0.002(0.001) & 0.002(0.001) & 0.998(0.001) & 0.002(0.001) & 24.004(0.002) & 0.077(0.000) & 4.992(0.009) \\
SIS-LASSO & strong & 1300 & 1.000(0.000) & 110.045(1.515) & 140.000(0.875) & 0.712(0.008) & 0.066(0.007) & 0.781(0.004) & 0.781(0.004) & 0.219(0.004) & 92.239(1.352) & 7.195(0.192) & 0.276(0.002) & 2.080(0.021) \\
SIS-MCP & strong & 1300 & 1.000(0.000) & 42.512(1.015) & 32.000(0.848) & 0.680(0.008) & 0.046(0.006) & 0.397(0.008) & 0.397(0.008) & 0.603(0.008) & 25.505(0.953) & 7.993(0.205) & 0.549(0.007) & 1.902(0.024) \\
SIS-NLP & strong & 1300 & 0.999(0.001) & 17.888(0.221) & 22.000(0.100) & 0.677(0.008) & 0.058(0.007) & 0.055(0.002) & 0.055(0.002) & 0.945(0.002) & 0.965(0.034) & 8.076(0.209) & 0.743(0.007) & 1.768(0.025) \\
SIS-SCAD & strong & 1300 & 1.000(0.000) & 52.572(0.829) & 52.000(0.781) & 0.689(0.008) & 0.048(0.006) & 0.601(0.006) & 0.601(0.006) & 0.399(0.006) & 35.354(0.751) & 7.782(0.200) & 0.456(0.005) & 1.931(0.023) \\
S3VS-full-LASSO & weak & 1300 & 1.000(0.000) & 101.015(0.032) & 101.000(0.000) & 0.220(0.003) & 0.000(0.000) & 0.945(0.001) & 0.945(0.001) & 0.055(0.001) & 95.509(0.091) & 19.494(0.085) & 0.087(0.001) & 1.614(0.004) \\
S3VS-full-MCP & weak & 1300 & 0.302(0.013) & 0.378(0.018) & 0.000(0.000) & 0.010(0.001) & 0.000(0.000) & 0.339(0.023) & 0.103(0.008) & 0.661(0.023) & 0.121(0.010) & 24.742(0.016) & 0.019(0.001) & 1.248(0.002) \\
S3VS-full-NLP & weak & 1300 & 0.998(0.001) & 9.981(0.141) & 10.000(0.409) & 0.137(0.002) & 0.000(0.000) & 0.643(0.006) & 0.641(0.006) & 0.357(0.006) & 6.560(0.111) & 21.579(0.059) & 0.188(0.003) & 1.222(0.003) \\
S3VS-full-SCAD & weak & 1300 & 0.305(0.013) & 0.574(0.030) & 0.000(0.000) & 0.016(0.001) & 0.000(0.000) & 0.373(0.023) & 0.114(0.008) & 0.627(0.023) & 0.177(0.014) & 24.603(0.026) & 0.029(0.002) & 1.238(0.002) \\
S3VS-one-step-LASSO & weak & 1300 & 1.000(0.000) & 1.000(0.000) & 1.000(0.000) & 0.030(0.000) & 0.000(0.000) & 0.243(0.012) & 0.243(0.012) & 0.757(0.012) & 0.243(0.012) & 24.243(0.012) & 0.058(0.001) & 1.254(0.002) \\
S3VS-one-step-MCP & weak & 1300 & 0.290(0.013) & 0.290(0.013) & 0.000(0.000) & 0.008(0.000) & 0.000(0.000) & 0.316(0.024) & 0.092(0.008) & 0.684(0.024) & 0.092(0.008) & 24.802(0.011) & 0.015(0.001) & 1.254(0.002) \\
S3VS-one-step-NLP & weak & 1300 & 0.998(0.001) & 0.998(0.001) & 1.000(0.000) & 0.030(0.000) & 0.000(0.000) & 0.241(0.012) & 0.241(0.012) & 0.759(0.012) & 0.241(0.012) & 24.243(0.012) & 0.058(0.001) & 1.254(0.002) \\
S3VS-one-step-SCAD & weak & 1300 & 0.292(0.013) & 0.292(0.013) & 0.000(0.000) & 0.008(0.000) & 0.000(0.000) & 0.327(0.024) & 0.095(0.008) & 0.673(0.024) & 0.095(0.008) & 24.804(0.011) & 0.015(0.001) & 1.254(0.002) \\
SIS-LASSO & weak & 1300 & 0.964(0.005) & 135.572(2.041) & 184.000(0.512) & 0.331(0.005) & 0.000(0.000) & 0.914(0.003) & 0.881(0.006) & 0.086(0.003) & 127.291(1.934) & 16.718(0.121) & 0.100(0.001) & 1.357(0.004) \\
SIS-MCP & weak & 1300 & 0.931(0.007) & 118.868(2.032) & 166.000(1.186) & 0.287(0.005) & 0.000(0.000) & 0.894(0.005) & 0.832(0.008) & 0.106(0.005) & 111.688(1.917) & 17.819(0.127) & 0.090(0.001) & 1.372(0.004) \\
SIS-NLP & weak & 1300 & 0.914(0.008) & 6.227(0.234) & 3.000(0.000) & 0.094(0.002) & 0.000(0.000) & 0.430(0.010) & 0.393(0.010) & 0.570(0.010) & 3.875(0.198) & 22.648(0.059) & 0.138(0.003) & 1.237(0.003) \\
SIS-SCAD & weak & 1300 & 0.964(0.005) & 126.525(1.986) & 173.000(0.847) & 0.303(0.005) & 0.000(0.000) & 0.912(0.003) & 0.879(0.006) & 0.088(0.003) & 118.939(1.879) & 17.415(0.123) & 0.096(0.001) & 1.362(0.004) \\
\bottomrule
\end{tabular}
}
\end{table}

\subsubsection{Results}
The primary comparison fixes $n=800$ and $p=20{,}000$ and pools the 13 named scenarios within each signal regime. Table~\ref{tab:LM-primary-performance} and Figure~\ref{fig:LM-tpr-fdp-pareto} show the trade-off between recovering active variables and keeping the selected model small. At the moderate signal level, full S3VS--LASSO has TPR $0.652$ but selects about 101 variables on average and has mean FDP $0.839$. Full S3VS--NLP has somewhat lower TPR ($0.580$) but a much smaller mean model size ($21.1$), lower mean FDP ($0.317$), and lower independent-test MSE. Full S3VS--MCP and full S3VS--SCAD are substantially sparser, with mean sizes $2.1$ and $2.6$ and mean FDPs $0.036$ and $0.071$, respectively, but their TPRs are only $0.079$ and $0.092$. The figure shows that no selector is uniformly best. The nonconvex selectors have low FDP but low TPR, LASSO has higher recovery but high FDP, and NLP lies between these extremes, with its performance depending on the design structure.

\begin{figure}[htbp]
\centering
\includegraphics[width=\textwidth]{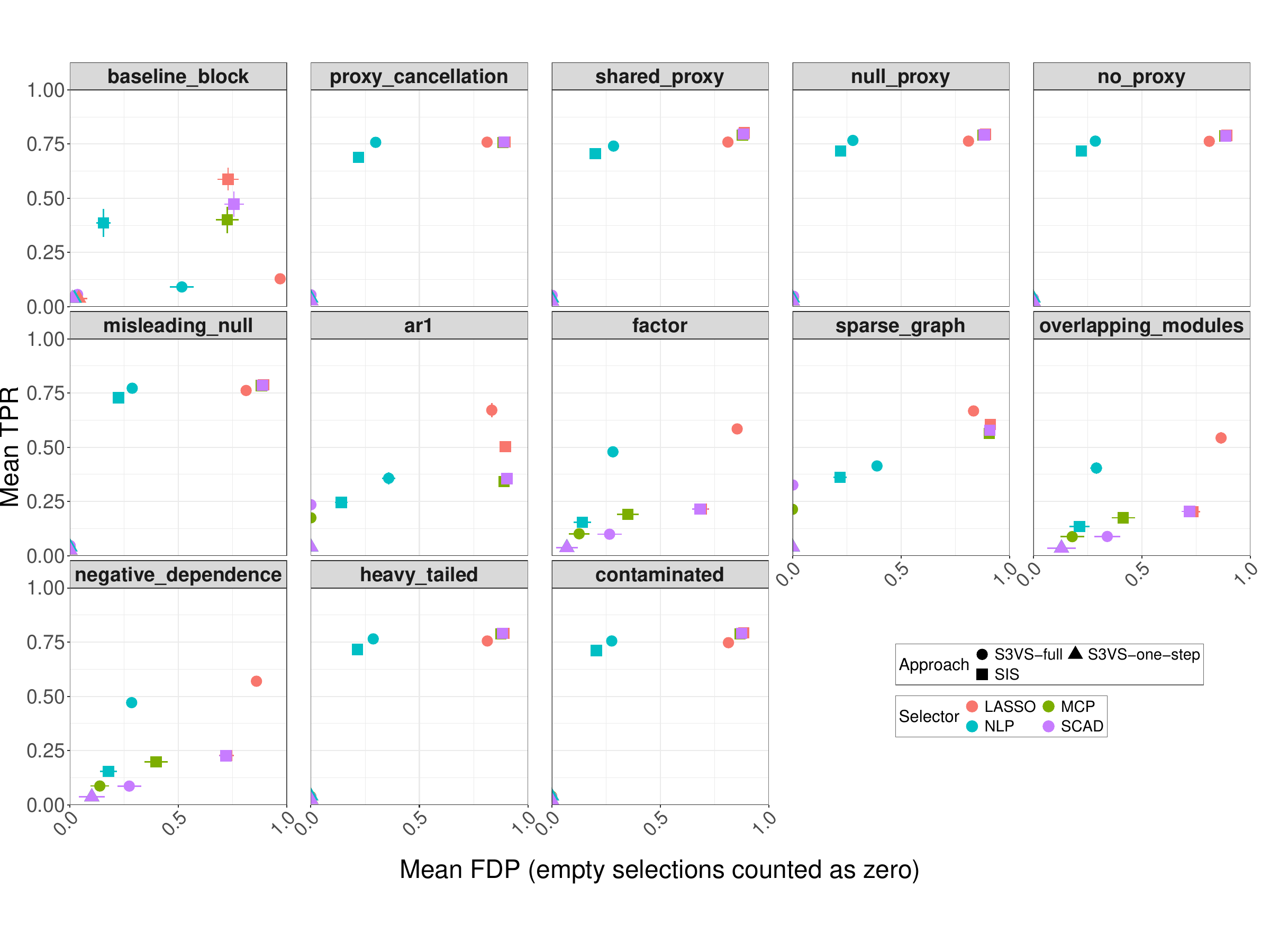}
\caption{Mean TPR and unconditional FDP at $n=800$, $p=20{,}000$, and the moderate signal level. Points are replicate means and bars are 95\% Monte Carlo intervals. Empty selections contribute zero to the unconditional FDP.}
\label{fig:LM-tpr-fdp-pareto}
\end{figure}

Recovery increased substantially with signal strength. For full S3VS--NLP, pooled TPR increases from $0.137$ under weak signals to $0.580$ under moderate signals and $0.869$ under strong signals. The corresponding full S3VS--LASSO values are $0.220$, $0.652$, and $0.898$. Figure~\ref{fig:LM-tpr-heatmap} shows that this improvement occurs across most scenarios, although the ordering of methods varies with the dependence structure. The weak-signal results also show why a small FDP alone is not sufficient: full S3VS--MCP and full S3VS--SCAD have mean model sizes below one and TPRs near zero in that regime. Figure~\ref{fig:LM-tpr-boxplots} displays the replicate-level distributions underlying these means and shows appreciable scenario-to-scenario variation, especially for correlated and non-Gaussian designs.

\begin{figure}[htbp]
\centering
\includegraphics[width=\textwidth]{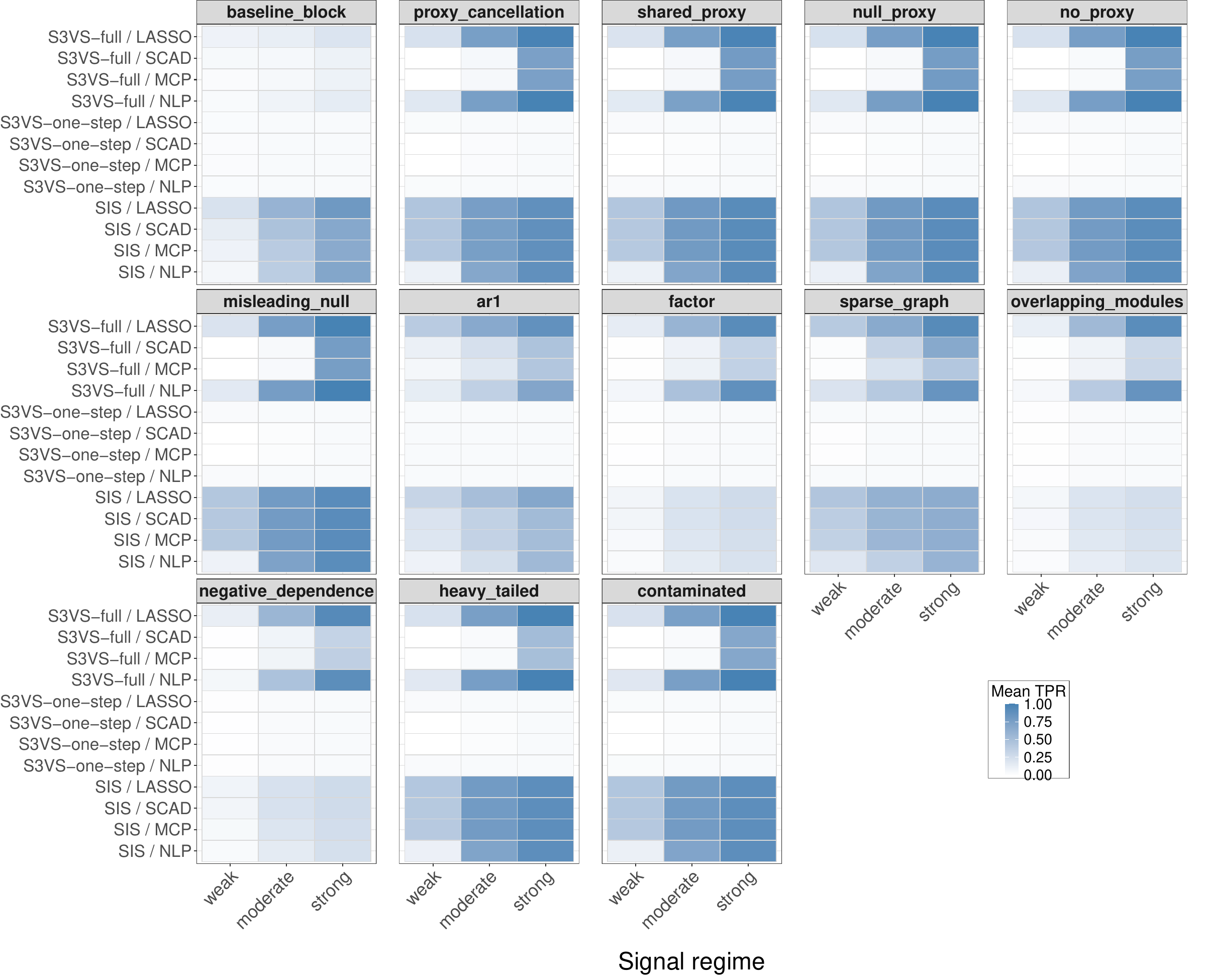}
\caption{Named-method TPR heat map at $n=800$ and $p=20{,}000$, stratified by signal regime and design scenario.}
\label{fig:LM-tpr-heatmap}
\end{figure}

\begin{figure}[htbp]
\centering
\includegraphics[width=\textwidth]{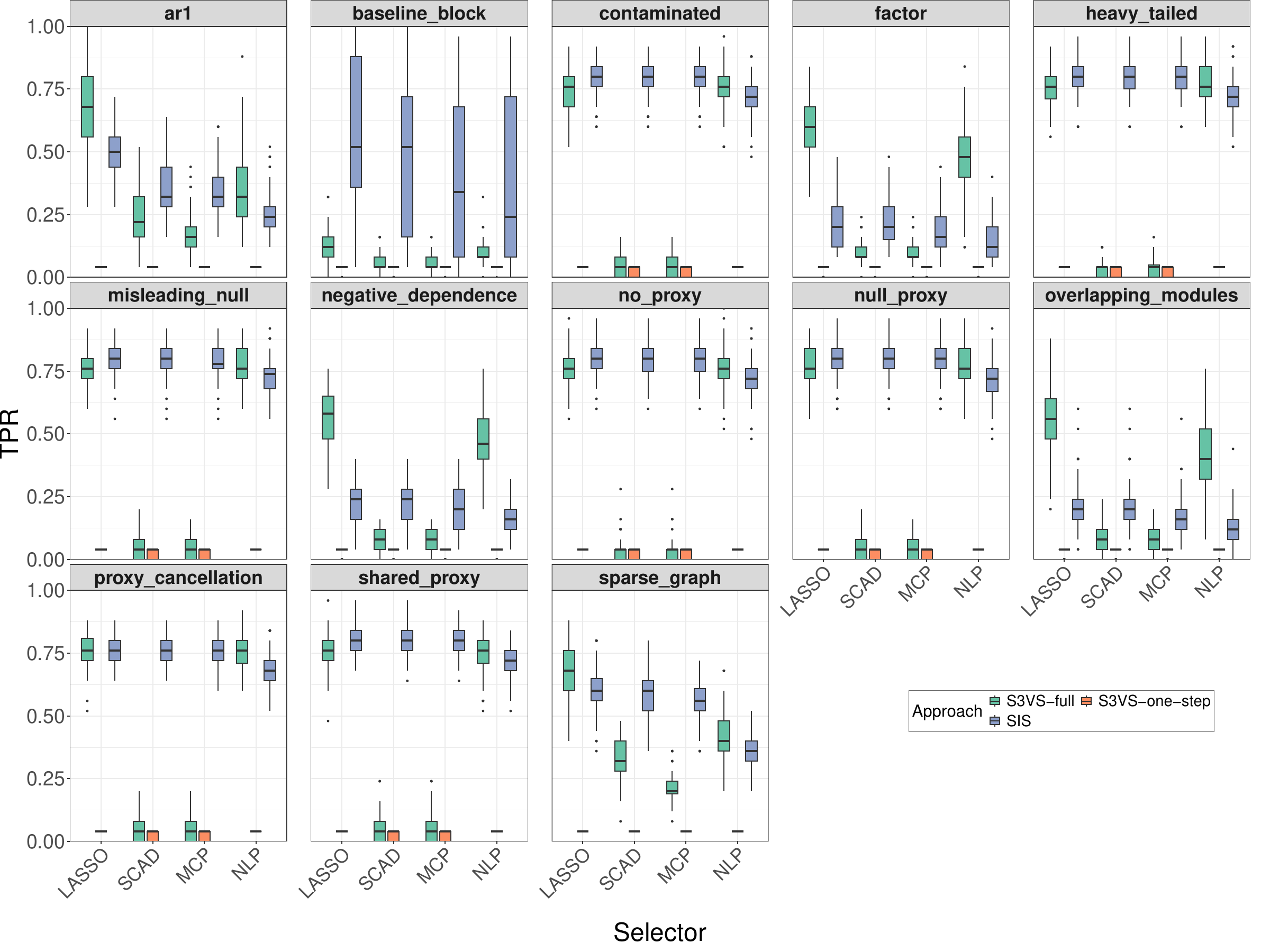}
\caption{Replicate-level TPR distributions at $n=800$, $p=20{,}000$, and the moderate signal level across the named design scenarios.}
\label{fig:LM-tpr-boxplots}
\end{figure}

Comparing first-iteration and full S3VS shows the effect of later iterations. As shown in Figure~\ref{fig:LM-selected-size}, the one-step S3VS procedures select roughly one variable on average and have TPRs between $0.008$ and $0.040$ in the primary table, whereas later S3VS iterations increase the selected model size and recovery. The SIS benchmarks also illustrate the effect of the selector: SIS--LASSO, SIS--SCAD, and SIS--MCP select more than 110 variables on average at the moderate signal level and have mean FDPs from $0.756$ to $0.840$, while SIS--NLP selects 15.8 variables with mean FDP $0.196$ and TPR $0.494$. These results distinguish the iterative S3VS search from a single screening-and-selection pass.

\begin{figure}[htbp]
\centering
\includegraphics[width=\textwidth]{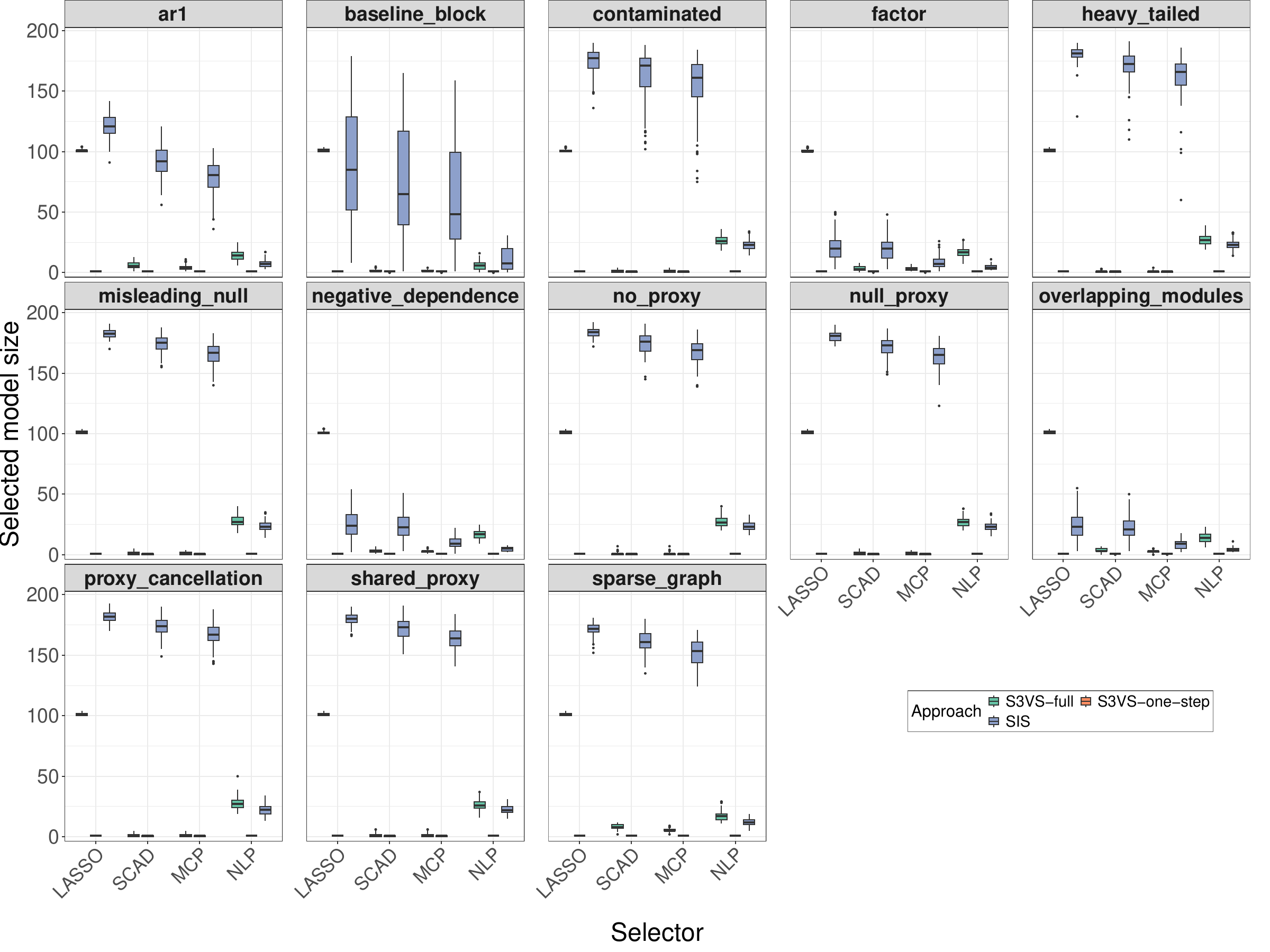}
\caption{Replicate-level selected-model sizes at $n=800$, $p=20{,}000$, and the moderate signal level across the named design scenarios.}
\label{fig:LM-selected-size}
\end{figure}

Independent prediction gives a complementary view of the selection trade-off. Figure~\ref{fig:LM-test-mse} and Table~\ref{tab:LM-primary-performance} show that full S3VS--NLP has the lowest or nearly lowest test error among the procedures in the primary comparison: its pooled MSE is $1.222$, $1.287$, and $1.266$ for weak, moderate, and strong signals, compared with $1.237$, $1.348$, and $1.768$ for SIS--NLP. The one-step procedures have test MSE close to the null-model benchmark in the moderate and strong settings because their selected sets contain little signal. Full S3VS--LASSO can recover more active variables than NLP, but its large models and high FDP do not yield a corresponding prediction advantage.

\begin{figure}[htbp]
\centering
\includegraphics[width=\textwidth]{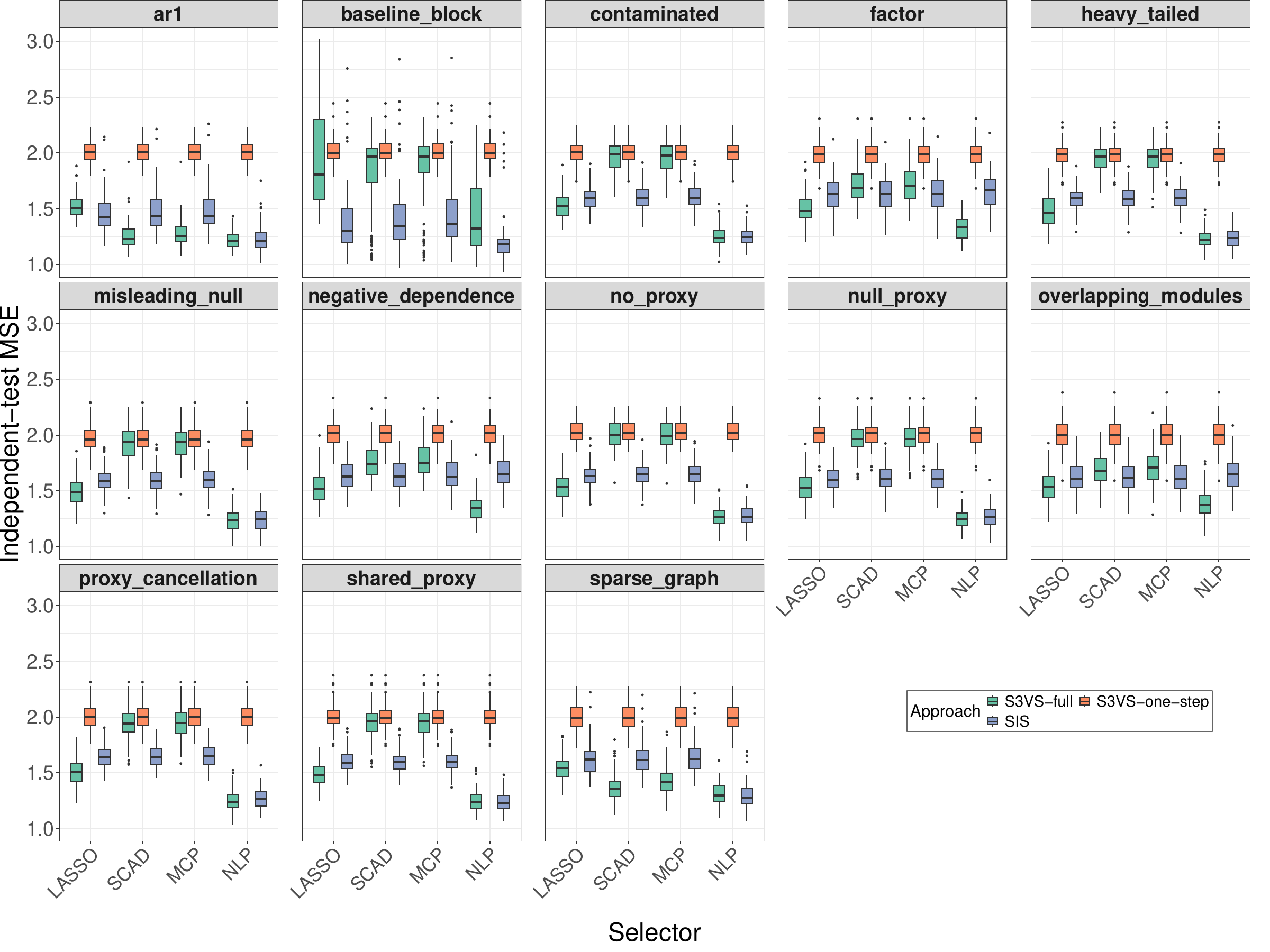}
\caption{Independent-test MSE distributions at $n=800$, $p=20{,}000$, and the moderate signal level across the named design scenarios. Each test set is independently generated from the same replicate-specific coefficient vector as its training set.}
\label{fig:LM-test-mse}
\end{figure}

Runtime generally increases with $p$ at fixed $n$, while its dependence on $n$ is not monotone because the cross-validation and optimization workloads vary across replicates. These runtime patterns are shown in Figure~\ref{fig:LM-runtime}. Full S3VS--SCAD and full S3VS--MCP are the least expensive S3VS variants, whereas full S3VS--NLP and SIS--NLP are more costly in several settings. The first-iteration S3VS results are not assigned a separate runtime because they are extracted from the corresponding full run. The supplementary tables and figures give the proxy/no-proxy results, covariance-structure results, diagnostics, coefficient-error summaries, and the complete sample-size and dimension comparisons.

\begin{figure}[htbp]
\centering
\includegraphics[width=\textwidth]{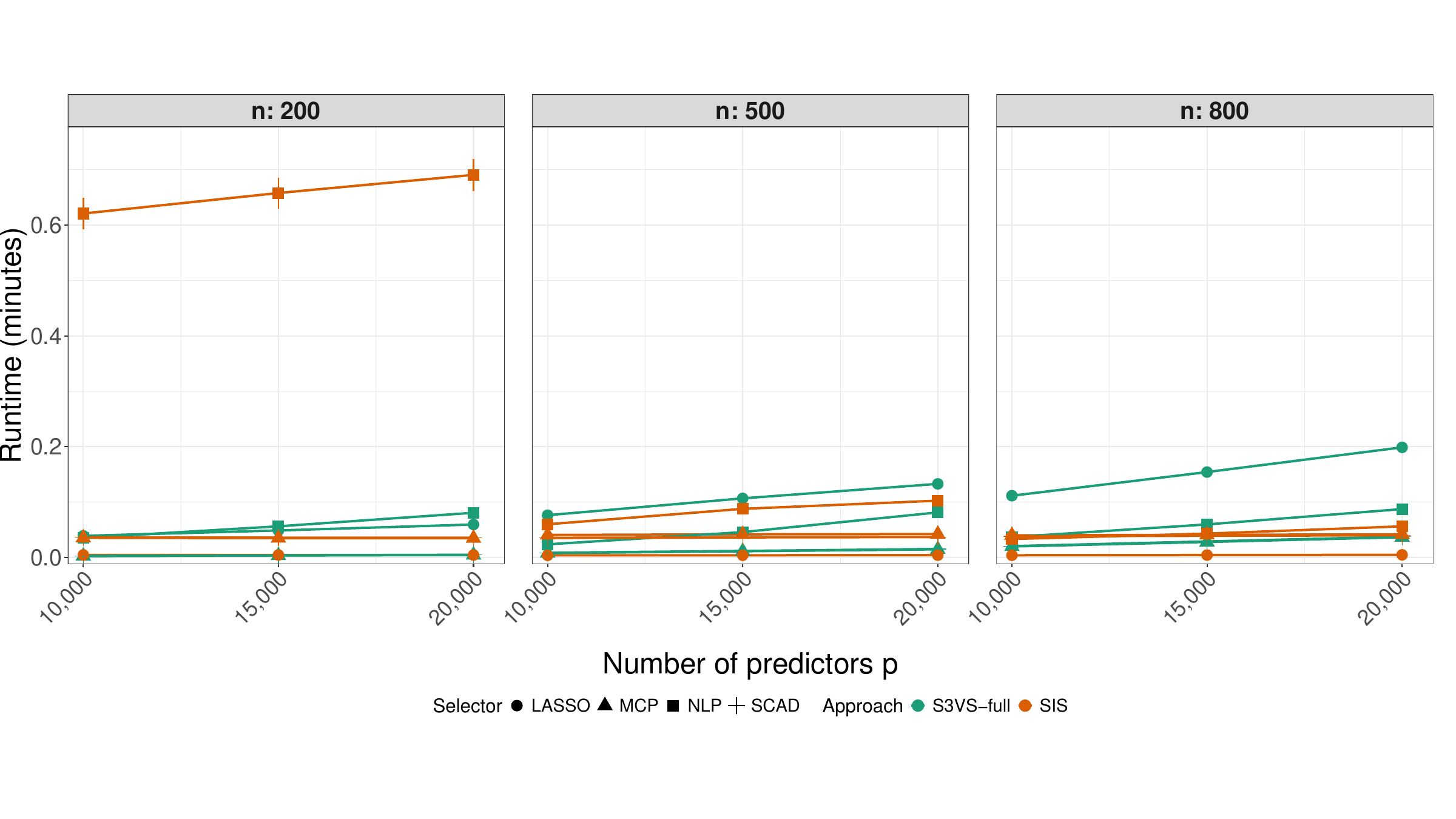}
\caption{Runtime scalability for separately timed full S3VS and one-pass SIS procedures across $n$ and $p$. First-iteration S3VS selections are derived from the shared full run and are not separately timed.}
\label{fig:LM-runtime}
\end{figure}

\subsection{Generalized Linear Models} \label{sec:simulation-glm}

\subsubsection{Simulation Design}
We evaluated the generalized-linear-model implementation using binary logistic regression. For each observation, the data-generating model was

\[   
Y_i \mid \mathbf{X}_i \sim \operatorname{Bernoulli}(\pi_i),   \qquad   \operatorname{logit}(\pi_i)   = \alpha + \mathbf{X}_i^{\mathsf T}\boldsymbol{\beta}. 
\] 

We considered sample sizes $n\in\{200,500,800\}$ and dimensions $p\in\{10{,}000,15{,}000,20{,}000\}$. Each $(n,p)$--scenario combination was replicated 100 times. A new coefficient vector, training design, independent test design, and pair of binary response vectors were generated for every replicate. The first 25 predictors were active, with independently sampled signs. The coefficient vector was scaled so that the linear predictor had standard deviation one, and the intercept was calibrated to a target marginal prevalence of 0.50.

We used two predictor structures.  In the no-proxy scenario, all predictors were independent.  In the correlated-proxy scenario, the first 100 predictors formed a block with pairwise correlation 0.60: the first 25 were active and the remaining 75 were null proxies for the active variables.  The other predictors were independent of this block.  Predictors were standardized within each generated data set.  The independent test sample was generated from the same coefficient vector and predictor structure as the training sample, so that prediction performance did not reuse the training responses.

\subsubsection{Methods and Performance Measures}
We compared S3VS with SIS followed by logistic LASSO.  The S3VS configuration used a single leading variable for the response--variable screen, a five-variable leading set for the predictor--predictor screen, liberal aggregation, conservative-from-the-beginning removal, and the LASSO selector with the one-standard-error tuning rule.  Each S3VS run used the pre-iteration selected-model-size limit $m_{\max}=100$ and the cumulative unsuccessful-iteration limit $n_{\mathrm{skip}}=3$. The comparator retained at most $\min(200,n-1)$ variables by marginal screening and then fit logistic LASSO using deviance cross-validation and the one-standard-error rule.  Both procedures were applied to the same training data within a replicate.

For each selected model, we recorded the probability of a nonempty selection, true-positive rate (TPR), false-discovery proportion (FDP), and selected-model size.  We also refit the logistic model on the selected variables and evaluated independent-test log-loss.  FDP was defined as zero when the selected set was empty. The probability of a nonempty selection was reported separately. All entries below are replicate means with Monte Carlo standard errors (MCSE).

\subsubsection{Results}
The primary comparison uses the largest sample-size and dimension setting, $n=800$ and $p=20{,}000$, and is reported in Table~\ref{tab:GLM-primary-performance}.  The corresponding comparison of mean TPR and mean unconditional FDP is shown in Figure~\ref{fig:GLM-primary-TPR-FDP}.

\begin{table}[htbp]
\centering
\small
\caption{Binary-logistic simulation performance at $n=800$ and $p=20{,}000$. Entries are means with Monte Carlo standard errors in parentheses.}
\label{tab:GLM-primary-performance}
{\fontsize{9}{10}\selectfont
\setlength{\tabcolsep}{1pt}
\begin{tabular}{llrrrrrrr}
\toprule
Scenario & Method & n & p & \shortstack[c]{P(nonempty),\\MCSE} & \shortstack[c]{TPR,\\MCSE} & \shortstack[c]{FDP,\\MCSE} & \shortstack[c]{Mean size,\\MCSE} & \shortstack[c]{Test log-loss,\\MCSE} \\
\midrule
Correlated proxy & S3VS & 800 & 20,000 & 1.000 (0.000) & 0.800 (0.029) & 0.764 (0.007) & 84.5 (2.8) & 0.664 (0.004) \\
Correlated proxy & SIS--logistic--LASSO & 800 & 20,000 & 1.000 (0.000) & 0.316 (0.014) & 0.933 (0.003) & 125.5 (2.3) & 0.862 (0.006) \\
No proxy & S3VS & 800 & 20,000 & 1.000 (0.000) & 0.155 (0.008) & 0.858 (0.007) & 28.5 (1.2) & 0.747 (0.003) \\
No proxy & SIS--logistic--LASSO & 800 & 20,000 & 1.000 (0.000) & 0.377 (0.008) & 0.949 (0.001) & 186.0 (0.3) & 0.937 (0.006) \\
\bottomrule
\end{tabular}
}
\end{table}

\begin{figure}[htbp]   
\centering   
\includegraphics[width=\textwidth]{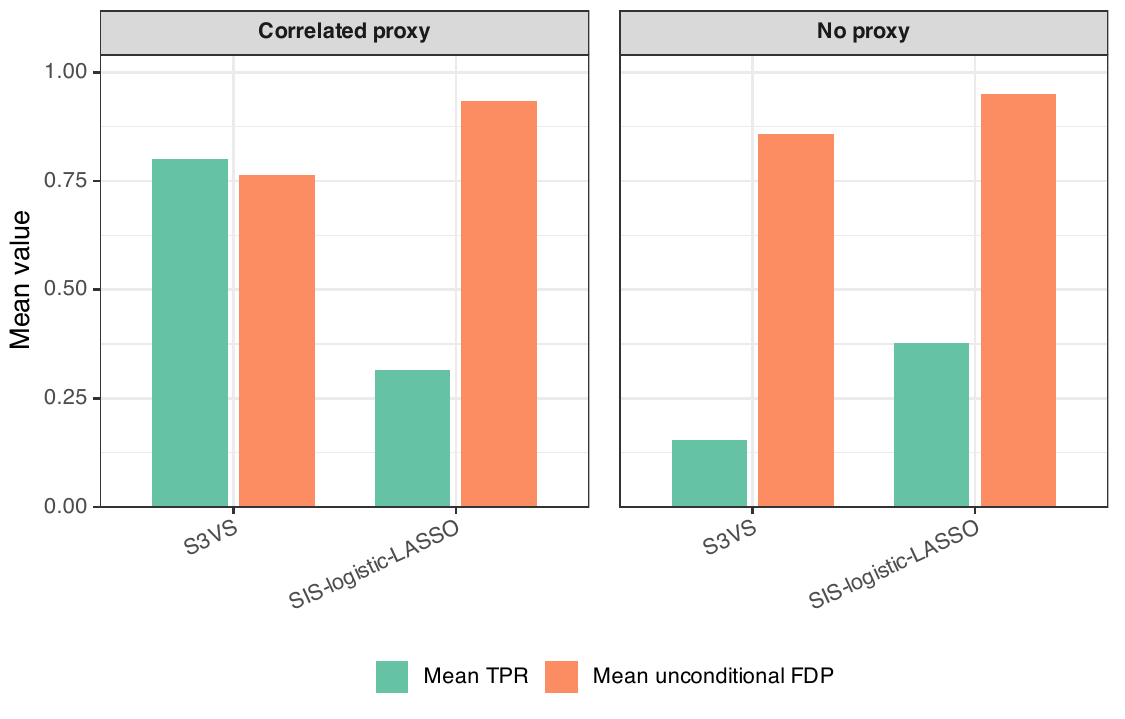}   
\caption{Binary-logistic simulation performance at $n=800$ and   $p=20{,}000$.  Bars show replicate means of TPR and unconditional FDP for   S3VS and SIS--logistic--LASSO in the correlated-proxy and   no-proxy scenarios.}   
\label{fig:GLM-primary-TPR-FDP} 
\end{figure}

In the correlated-proxy scenario, S3VS had a mean TPR of 0.800 (MCSE 0.029), compared with 0.316 (0.014) for SIS--logistic--LASSO.  It also selected a smaller model on average, 84.5 (2.8) versus 125.5 (2.3) variables, and had lower unconditional FDP, 0.764 (0.007) versus 0.933 (0.003).  These differences were accompanied by lower independent-test log-loss for S3VS, 0.664 (0.004) versus 0.862 (0.006).

The no-proxy scenario gives a different result.  SIS--logistic--LASSO had a higher TPR, 0.377 (0.008) versus 0.155 (0.008), whereas S3VS selected substantially fewer variables, 28.5 (1.2) versus 186.0 (0.3), and had lower unconditional FDP, 0.858 (0.007) versus 0.949 (0.001). Its independent-test log-loss was also lower, 0.747 (0.003) versus 0.937 (0.006). The pattern therefore depends on predictor structure. S3VS performed especially well when correlated null proxies were available. In the no-proxy setting, SIS--logistic--LASSO had higher TPR, while S3VS selected a much smaller model and had lower test log-loss. The full grid of results is provided in Supplementary Table~S9
 and Supplementary Figure~S8.

\subsection{Cox Proportional Hazards Models}
\label{sec:simulation-cox}

We also evaluated S3VS for censored survival outcomes using Cox proportional hazards models. This simulation is a computational assessment rather than a theoretical validation of the survival implementation. Because S3VS is iterative and SIS--Cox--LASSO is single-pass, the primary comparison limits S3VS to one iteration. Full-iteration S3VS is reported as a secondary analysis in Supplementary Table~S10
 and Supplementary Figure~S9. 
 Complete results over the $(n,p)$ grid are given in Supplementary Table~S11
  and Supplementary Figure~S10.

\subsubsection{Simulation design and implementation}
The simulation considered sample sizes $n\in\{200,500,800\}$, predictor dimensions $p\in\{10{,}000,15{,}000,20{,}000\}$, and two predictor scenarios.  In the no-proxy scenario, the predictors were independent standard Gaussian variables. In the correlated-proxy scenario, the first 25 predictors were active, and the next 75 were null proxies correlated with the active variables through a common Gaussian factor. The remaining predictors were independent Gaussian variables.

For every replicate and every $(n,p)$--scenario combination, we generated a new training design, an independent test design, and a new coefficient vector.  The first 25 coefficients were assigned independent random signs and were rescaled so that the standard deviation of the training linear predictor was 0.60.  Thus, the simulation evaluates performance across replicates rather than repeatedly generating outcomes from one fixed design and one fixed coefficient vector.  If $\boldsymbol{x}_i$ denotes a row of the design matrix, the event time was generated from
\[
 T_i\mid\boldsymbol{x}_i\sim
 \operatorname{Exponential}\!\left\{0.10\exp(\boldsymbol{x}_i^\mathsf{T}
 \boldsymbol{\beta})\right\}.
\]
An independent exponential censoring time $C_i$ was generated, and the observed outcome was
\[
 Y_i=\min(T_i,C_i),\qquad
 \Delta_i=I(T_i\leq C_i).
\]
The censoring rate was calibrated separately for each generated training design to target a censoring proportion of 0.30 and was then used for the corresponding test sample.  Training event and censoring proportions were recorded for every replicate.  We used 100 independent replicates for each cell of the simulation grid.

S3VS used top-$k$ outcome screening with $k=1$, the relative-to-maximum leading-set rule with $\pi_B=95$ (equivalently, the fixed threshold $\tau_B=0.95$ because the predictor--leader score is absolute Pearson correlation and the leader has score one), liberal aggregation of selected variables, and conservative-from-the-beginning removal of variables. Within each leading set, a Cox LASSO model was fitted using the one-standard-error value from ten-fold cross-validation. The full iterative implementation used the pre-iteration selected-model-size limit $m_{\max}=100$ and the cumulative unsuccessful-iteration limit $n_{\mathrm{skip}}=3$. For the primary comparison, the same S3VS settings were used with $m_{\max}=1$ and $n_{\mathrm{skip}}=1$, thereby restricting S3VS to one iteration.

The comparator ranked predictors by the absolute marginal Cox Wald statistic from separate univariate Cox proportional-hazards fits using the observed survival times and event indicators. Specifically, for candidate predictor $j$, the screening statistic was $|z_j|$, where $z_j=\widehat{\beta}_{j,\mathrm{marg}}/\widehat{\operatorname{se}}(\widehat{\beta}_{j,\mathrm{marg}})$ is the Wald statistic from the univariate Cox fit. It retained up to $\min\{200,n_{\mathrm{train}}-1,p\}$ highest-ranked predictors and fitted a Cox LASSO model to this screened set using ten-fold cross-validation and the one-standard-error tuning value. This comparator-specific Wald screening rule is distinct from the marginal-deviance score used by S3VS in \eqref{eq:cox-marginal-score}. We refer to this procedure as SIS--Cox--LASSO.

For both procedures, the selected variables were evaluated against the known active set. We recorded the probability of a nonempty selected model, the true positive rate (TPR), the mean unconditional false discovery proportion (FDP), the selected-model size, and the independent-test concordance index. For a nonempty selected model, FDP was the number of selected null variables divided by the total number selected. An empty model contributed zero to the unconditional FDP average. The test concordance index, which measures independent-test risk-ranking ability, was computed after a post-selection Cox refit on the training observations, using Cox LASSO when at least two variables were selected and a univariate Cox refit when exactly one variable was selected, and evaluating the resulting risk scores on the independent test sample. All reported uncertainties are Monte Carlo standard errors across the 100 replicates.

\subsubsection{Primary one-iteration comparison}
The primary comparison is summarized at the largest setting, $(n,p)=(800,20{,}000)$, in Table~\ref{tab:surv-one-iteration-primary-performance} and Figure~\ref{fig:surv-one-iteration-tpr-fdp}. The results show a recovery--purity trade-off. In the correlated-proxy setting, both procedures selected a nonempty model in all replicates. S3VS had TPR $0.034$ and unconditional FDP $0.150$, whereas SIS--Cox--LASSO had higher TPR ($0.242$) but much higher unconditional FDP ($0.932$) and a much larger mean selected model ($98.0$ versus $1.0$ variables). The test C-index was slightly higher for S3VS ($0.585$ versus $0.572$). In the no-proxy setting, both procedures again selected nonempty models in all replicates. SIS--Cox--LASSO had higher TPR ($0.396$ versus $0.018$), but also much higher unconditional FDP ($0.937$ versus $0.560$) and a much larger mean selected model ($157.6$ versus $1.0$); its test C-index ($0.531$) was slightly higher than that of S3VS ($0.512$). Thus, one-iteration S3VS was substantially more conservative in model size and false discoveries, whereas SIS--Cox--LASSO recovered more active variables through much larger, less-pure models. The predictive difference was small and favored different procedures across the two structures.

The predictive comparison is therefore structure-dependent (Supplementary Table~S12
 and Supplementary Figure~S11. 
 Across the full grid, S3VS had higher test concordance than SIS--Cox--LASSO in every correlated-proxy configuration, with C-indices ranging from $0.558$ to $0.585$ for S3VS and from $0.523$ to $0.572$ for SIS--Cox--LASSO. In the no-proxy configurations, S3VS ranged from $0.502$ to $0.522$, whereas SIS--Cox--LASSO ranged from $0.497$ to $0.539$; S3VS had the higher C-index at $n=200$, while SIS--Cox--LASSO had the higher C-index at $n=500$ and $n=800$. The one-iteration results therefore show a modest predictive advantage for S3VS when correlated proxies are present, but no uniform advantage without them. The full iterative S3VS analysis is treated as secondary because it uses a larger iteration budget than the single-pass comparator.

\begin{table}[htbp]
\centering
\caption{One-iteration Cox-model simulation performance at $n = 800$ and $p = 20{,}000$. Entries are means with Monte Carlo standard errors in parentheses; FDP is the unconditional FDP with empty selections counted as zero.}
\label{tab:surv-one-iteration-primary-performance}
{\fontsize{9}{10}\selectfont
\setlength{\tabcolsep}{1pt}
\begin{tabular}{llrrrrrrr}
\toprule
Scenario & Method & n & p & \shortstack[c]{P(nonempty),\\MCSE} & \shortstack[c]{TPR,\\MCSE} & \shortstack[c]{FDP,\\MCSE} & \shortstack[c]{Mean size,\\MCSE} & \shortstack[c]{Test C-index,\\MCSE} \\
\midrule
Correlated proxy & S3VS (one iteration) & 800 & 20,000 & 1.000 (0.000) & 0.034 (0.001) & 0.150 (0.036) & 1.0 (0.0) & 0.585 (0.003) \\
Correlated proxy & SIS--Cox--LASSO (one iteration) & 800 & 20,000 & 1.000 (0.000) & 0.242 (0.008) & 0.932 (0.003) & 98.0 (2.5) & 0.572 (0.004) \\
No proxy & S3VS (one iteration) & 800 & 20,000 & 1.000 (0.000) & 0.018 (0.002) & 0.560 (0.050) & 1.0 (0.0) & 0.512 (0.002) \\
No proxy & SIS--Cox--LASSO (one iteration) & 800 & 20,000 & 1.000 (0.000) & 0.396 (0.009) & 0.937 (0.001) & 157.6 (0.8) & 0.531 (0.002) \\
\bottomrule
\end{tabular}
}
\end{table}

\begin{figure}[htbp]
\centering
\includegraphics[width=\textwidth]{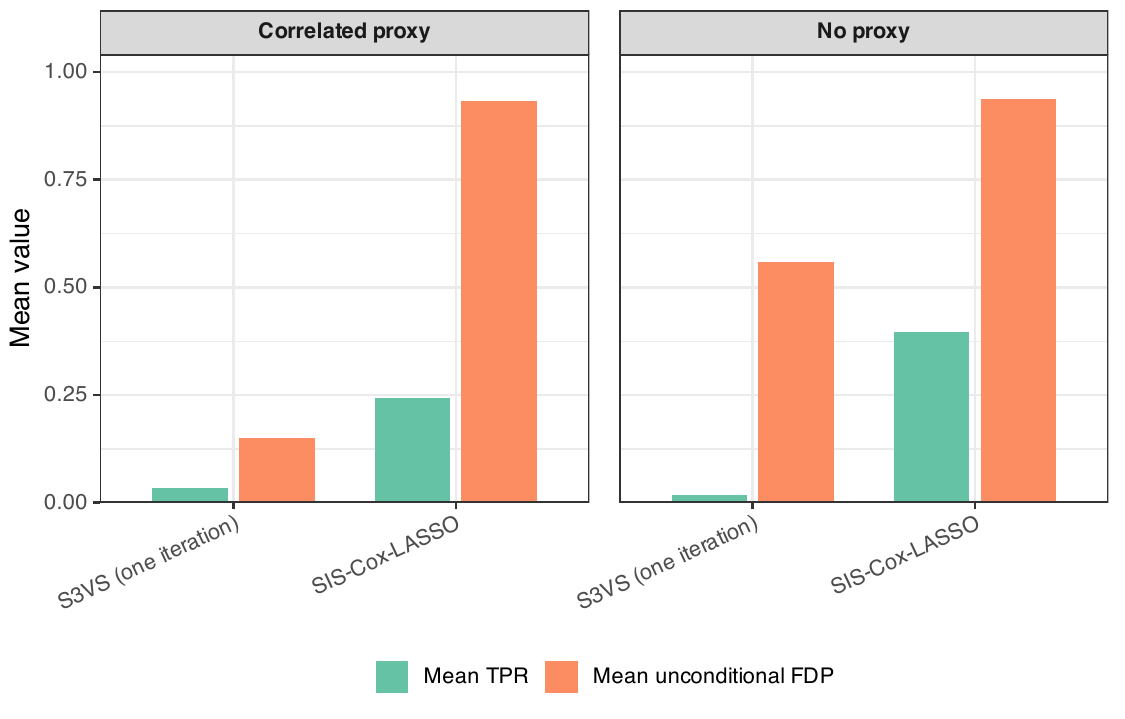}
\caption{Primary one-iteration comparison of S3VS and SIS--Cox--LASSO at $n = 800$ and $p = 20{,}000$. Bars show mean TPR and mean unconditional FDP in the correlated-proxy and no-proxy scenarios. Empty selections contribute zero to the unconditional FDP.}
\label{fig:surv-one-iteration-tpr-fdp}
\end{figure}

\section{Ovarian Cancer Data Analysis}
\label{sec:real-data-analysis}

\subsection{Data Sources and Preprocessing}
We evaluated S3VS using clinically annotated ovarian-cancer gene-expression data from the \texttt{curatedOvarianData} Bioconductor package \citep{ganzfried_et_al_2013}. The discovery analysis used \texttt{TCGA\_eset}, containing Affymetrix HT Human Genome U133A expression measurements and harmonized clinical annotations from The Cancer Genome Atlas ovarian-cancer study \citep{tcga_2011}. Overall survival was defined from \texttt{days\_to\_death} and \texttt{vital\_status}, with deceased patients coded as events. Patients with missing survival information, the documented expression outliers \texttt{TCGA.24.1927} and \texttt{TCGA.31.1955}, and probes mapped to multiple genes were excluded.

The clinical predictors were age at initial pathologic diagnosis, tumor stage, tumor grade, and debulking status. Within each outer training fold, gene-expression variables were standardized using the training-fold means and standard deviations, and variables with nonfinite or zero standard deviation were removed. The same transformation was applied to the corresponding held-out fold. Categorical variables were indicator coded using the training-fold levels, and the competing predictor sets were evaluated on their common complete cases.

External validation used \texttt{GSE9891\_eset}, an Affymetrix Human Genome U133 Plus 2.0 study of 285 serous and endometrioid ovarian tumors \citep{tothill_et_al_2008}. Patients with missing survival time or event status were excluded, and all genes selected from TCGA were available in GSE9891. For external validation, expression variables in the full TCGA cohort were standardized using the TCGA means and standard deviations, and the same TCGA-derived centering and scaling constants were then applied to the corresponding GSE9891 variables. GSE9891 outcomes were used only to define the evaluable validation set and to compute performance measures. They were not used for variable screening, tuning, or model fitting.

\subsection{S3VS Selection and Prediction Framework}
Using these preprocessed data, patients were assigned to five event-stratified outer folds. Within each outer training fold, we compared a one-iteration S3VS specification, a full-iteration S3VS specification, and a SIS--Cox--LASSO comparator. S3VS used top-$k$ leading-variable screening with $k=1$, top-$q$ leading sets with $q=5$, Cox LASSO within leading sets at $\lambda_{\min}$, conservative aggregation of selected variables, and conservative-from-the-beginning aggregation of nonselected variables. The one-iteration and full-iteration specifications used $(m_{\text{max}},\texttt{nskip})=(1,1)$ and $(100,3)$, respectively, whereas the comparator screened the 200 genes having the largest absolute marginal Cox Wald statistics.

After completing the five-fold TCGA analysis, each screening procedure was applied once to the full TCGA cohort to obtain the predictor set used for external validation. Using these full-TCGA selections, we fitted the post-selection Cox models in TCGA, with the penalty parameter for the penalized models chosen by five-fold cross-validation within TCGA. The resulting predictor sets, fitted coefficients, and tuning choices were then fixed before evaluation in GSE9891.

For each selected predictor set, we fitted an ordinary Cox model and Cox LASSO models at $\lambda_{\min}$ and $\lambda_{1\mathrm{se}}$, with clinical variables left unpenalized and gene coefficients penalized. The primary results in Table~\ref{tab:ovarian-internal-primary}, Table~\ref{tab:ovarian-external-primary}, and Figure~\ref{fig:ovarian-primary-performance} use the more conservative $\lambda_{1\mathrm{se}}$ fit. Complete internal results for all three fitting specifications are given in Table~\ref{tab:ovarian-internal-all}. Performance was assessed using Harrell's C-index, time-dependent AUC at 1, 3, and 5 years, and the integrated Brier score (IBS). Higher C-index and AUC values indicate better discrimination, whereas lower IBS indicates lower overall prediction error.

\begin{table}[htbp]
\centering
\small
\caption{Five-fold outer validation in TCGA for the primary penalized Cox specification. Entries are means with fold standard deviations in parentheses.}
\label{tab:ovarian-internal-primary}
{\fontsize{8.4}{9.4}\selectfont
\setlength{\tabcolsep}{1pt}
\begin{tabular}{lrrrrrrrr}
\toprule
Method & Predictors & \shortstack{Screened\\genes} & \shortstack{Final\\genes} & C-index & AUC 1 year & AUC 3 years & AUC 5 years & IBS
\\
\midrule
Clinical baseline & Clinical only & 0.0 & -- & 0.627 (0.026) & 0.697 (0.122) & 0.631 (0.057) & 0.590 (0.078) & 0.172 (0.008)
\\
S3VS one iteration & Clinical + S3VS one iteration genes & 3.2 & 0.0 & 0.624 (0.027) & 0.698 (0.126) & 0.624 (0.061) & 0.595 (0.085) & 0.174 (0.009)
\\
S3VS one iteration & S3VS one iteration genes only & 3.2 & 0.8 & 0.516 (0.070) & 0.514 (0.195) & 0.525 (0.082) & 0.561 (0.075) & 0.180 (0.009)
\\
S3VS full iteration & Clinical + S3VS full iteration genes & 28.6 & 17.6 & 0.640 (0.052) & 0.754 (0.098) & 0.630 (0.054) & 0.601 (0.119) & 0.182 (0.018)
\\
S3VS full iteration & S3VS full iteration genes only & 28.6 & 18.6 & 0.572 (0.067) & 0.672 (0.132) & 0.552 (0.105) & 0.524 (0.119) & 0.198 (0.021)
\\
SIS--Cox--LASSO & Clinical + SIS--Cox--LASSO genes & 200.0 & 9.4 & 0.617 (0.037) & 0.732 (0.095) & 0.612 (0.063) & 0.561 (0.043) & 0.179 (0.010)
\\
SIS--Cox--LASSO & SIS--Cox--LASSO genes only & 200.0 & 25.8 & 0.518 (0.017) & 0.569 (0.107) & 0.510 (0.033) & 0.464 (0.069) & 0.192 (0.011)
\\
\bottomrule
\end{tabular}
}
\end{table}

\begin{table}[htbp]
\centering
\small
\caption{Locked external validation in GSE9891 using models trained on the full TCGA cohort. The TCGA-selected predictor sets and fitted models were fixed before external evaluation, and the TCGA-derived centering and scaling constants were applied to the corresponding GSE9891 expression variables.}
\label{tab:ovarian-external-primary}
{\fontsize{8.4}{9.4}\selectfont
\setlength{\tabcolsep}{1pt}
\begin{tabular}{lrrrrrrrr}
\toprule
Method & Predictors & \shortstack{Screened\\genes} & \shortstack{Final\\genes} & C-index & AUC 1 year & AUC 3 years & AUC 5 years & IBS
\\
\midrule
Clinical baseline & Clinical only & 0 & -- & 0.668 & 0.671 & 0.713 & 0.659 & 0.165
\\
S3VS one iteration & S3VS one iteration genes only & 5 & 0 & 0.500 & 0.500 & 0.500 & 0.500 & 0.180
\\
S3VS one iteration & Clinical + S3VS one iteration genes & 5 & 0 & 0.666 & 0.697 & 0.705 & 0.641 & 0.167
\\
S3VS full iteration & S3VS full iteration genes only & 48 & 28 & 0.563 & 0.637 & 0.521 & 0.644 & 0.298
\\
S3VS full iteration & Clinical + S3VS full iteration genes & 48 & 26 & 0.645 & 0.730 & 0.631 & 0.644 & 0.267
\\
SIS--Cox--LASSO & SIS--Cox--LASSO genes only & 200 & 32 & 0.518 & 0.578 & 0.501 & 0.522 & 0.256
\\
SIS--Cox--LASSO & Clinical + SIS--Cox--LASSO genes & 200 & 14 & 0.672 & 0.736 & 0.706 & 0.679 & 0.176
\\
\bottomrule
\end{tabular}
}
\end{table}

\begin{figure}[t]
\centering
\includegraphics[width=0.96\textwidth]{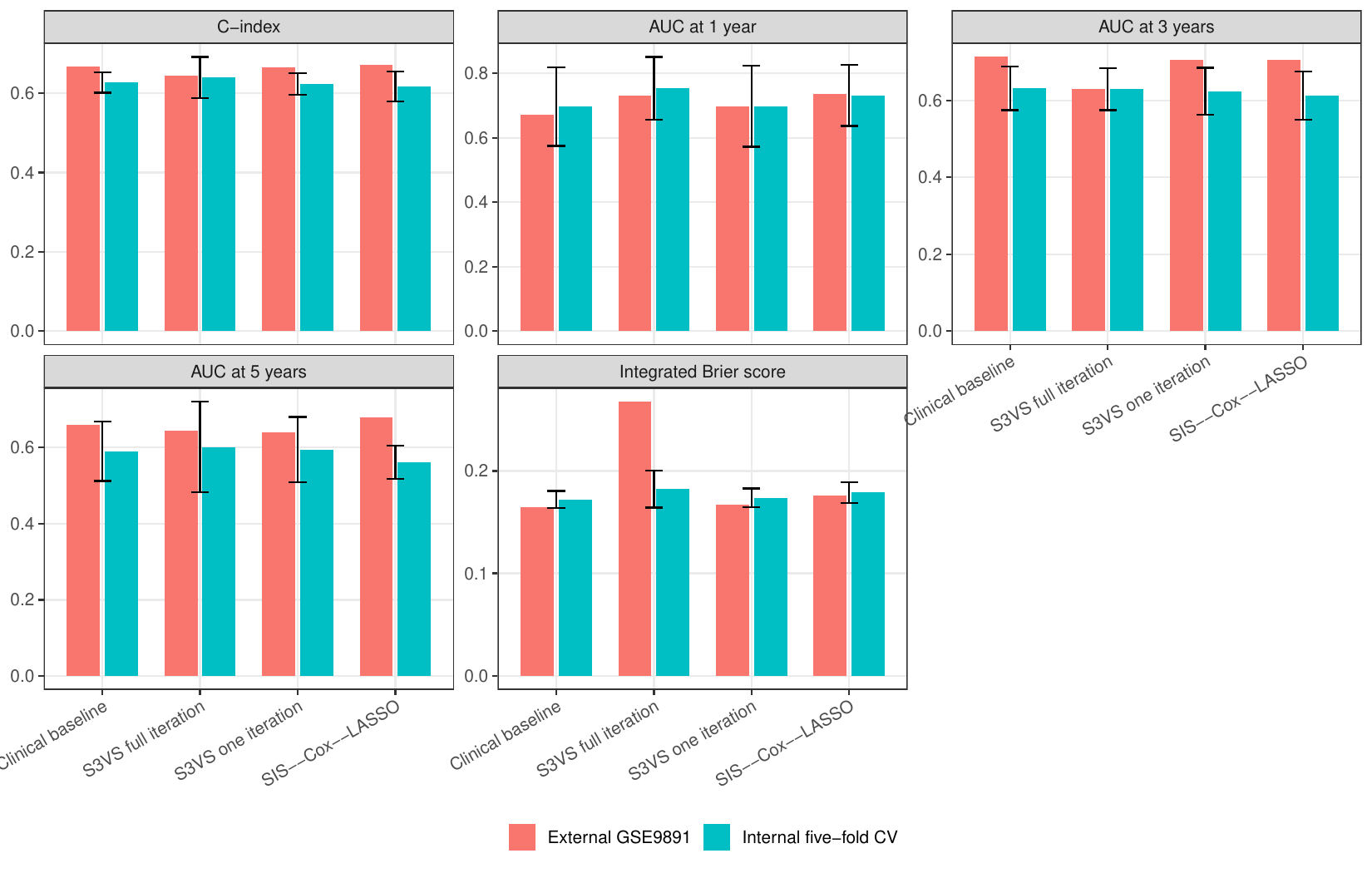}
\caption{C-index comparison for the primary $\lambda_{1\mathrm{se}}$ Cox specification. Error bars represent outer-fold standard deviations for internal validation. The external-validation cohort contributes a single estimate for each method.}
\label{fig:ovarian-primary-performance}
\end{figure}

\begin{table}
\centering
\small
\caption{Complete five-fold outer-validation results in TCGA. Entries are means with fold standard deviations in parentheses.}
\label{tab:ovarian-internal-all}
{\fontsize{7}{8}\selectfont
\setlength{\tabcolsep}{1pt}
\begin{tabular}{lrrrrrrrrr}
\toprule
Method & Predictors & Fitting & \shortstack{Screened\\genes} & \shortstack{Final\\genes} & C-index & AUC 1 year & AUC 3 years & AUC 5 years & IBS
\\
\midrule
Clinical baseline & Clinical only & Ordinary Cox & 0.0 & 0.0 & 0.627 (0.026) & 0.697 (0.122) & 0.631 (0.057) & 0.590 (0.078) & 0.172 (0.008)
\\
Clinical baseline & Clinical only & Penalized Cox: lambda.min & 0.0 & -- & 0.627 (0.026) & 0.697 (0.122) & 0.631 (0.057) & 0.590 (0.078) & 0.172 (0.008)
\\
Clinical baseline & Clinical only & Penalized Cox: lambda.1se & 0.0 & -- & 0.627 (0.026) & 0.697 (0.122) & 0.631 (0.057) & 0.590 (0.078) & 0.172 (0.008)
\\
S3VS one iteration & Clinical + S3VS one iteration genes & Ordinary Cox & 3.2 & 3.2 & 0.624 (0.031) & 0.700 (0.131) & 0.619 (0.037) & 0.625 (0.046) & 0.174 (0.008)
\\
S3VS one iteration & Clinical + S3VS one iteration genes & Penalized Cox: lambda.min & 3.2 & 3.2 & 0.630 (0.021) & 0.713 (0.112) & 0.620 (0.044) & 0.638 (0.044) & 0.173 (0.008)
\\
S3VS one iteration & Clinical + S3VS one iteration genes & Penalized Cox: lambda.1se & 3.2 & 0.0 & 0.624 (0.027) & 0.698 (0.126) & 0.624 (0.061) & 0.595 (0.085) & 0.174 (0.009)
\\
S3VS one iteration & S3VS one iteration genes only & Ordinary Cox & 3.2 & 3.2 & 0.540 (0.057) & 0.573 (0.189) & 0.515 (0.086) & 0.566 (0.115) & 0.185 (0.012)
\\
S3VS one iteration & S3VS one iteration genes only & Penalized Cox: lambda.min & 3.2 & 3.2 & 0.536 (0.060) & 0.565 (0.189) & 0.514 (0.086) & 0.559 (0.114) & 0.185 (0.012)
\\
S3VS one iteration & S3VS one iteration genes only & Penalized Cox: lambda.1se & 3.2 & 0.8 & 0.516 (0.070) & 0.514 (0.195) & 0.525 (0.082) & 0.561 (0.075) & 0.180 (0.009)
\\
S3VS full iteration & Clinical + S3VS full iteration genes & Ordinary Cox & 28.6 & 28.6 & 0.604 (0.072) & 0.703 (0.159) & 0.592 (0.086) & 0.540 (0.152) & 0.214 (0.034)
\\
S3VS full iteration & Clinical + S3VS full iteration genes & Penalized Cox: lambda.min & 28.6 & 23.8 & 0.615 (0.067) & 0.729 (0.134) & 0.598 (0.082) & 0.567 (0.119) & 0.200 (0.025)
\\
S3VS full iteration & Clinical + S3VS full iteration genes & Penalized Cox: lambda.1se & 28.6 & 17.6 & 0.640 (0.052) & 0.754 (0.098) & 0.630 (0.054) & 0.601 (0.119) & 0.182 (0.018)
\\
S3VS full iteration & S3VS full iteration genes only & Ordinary Cox & 28.6 & 28.6 & 0.571 (0.069) & 0.654 (0.150) & 0.554 (0.092) & 0.508 (0.155) & 0.227 (0.036)
\\
S3VS full iteration & S3VS full iteration genes only & Penalized Cox: lambda.min & 28.6 & 23.6 & 0.573 (0.069) & 0.656 (0.140) & 0.559 (0.097) & 0.515 (0.151) & 0.216 (0.031)
\\
S3VS full iteration & S3VS full iteration genes only & Penalized Cox: lambda.1se & 28.6 & 18.6 & 0.572 (0.067) & 0.672 (0.132) & 0.552 (0.105) & 0.524 (0.119) & 0.198 (0.021)
\\
SIS--Cox--LASSO & Clinical + SIS--Cox--LASSO genes & Ordinary Cox & 200.0 & 200.0 & 0.521 (0.063) & 0.526 (0.138) & 0.542 (0.102) & 0.515 (0.078) & 0.336 (0.031)
\\
SIS--Cox--LASSO & Clinical + SIS--Cox--LASSO genes & Penalized Cox: lambda.min & 200.0 & 48.8 & 0.612 (0.031) & 0.677 (0.104) & 0.628 (0.053) & 0.529 (0.109) & 0.187 (0.011)
\\
SIS--Cox--LASSO & Clinical + SIS--Cox--LASSO genes & Penalized Cox: lambda.1se & 200.0 & 9.4 & 0.617 (0.037) & 0.732 (0.095) & 0.612 (0.063) & 0.561 (0.043) & 0.179 (0.010)
\\
SIS--Cox--LASSO & SIS--Cox--LASSO genes only & Ordinary Cox & 200.0 & 200.0 & 0.494 (0.074) & 0.490 (0.168) & 0.505 (0.106) & 0.475 (0.106) & 0.349 (0.044)
\\
SIS--Cox--LASSO & SIS--Cox--LASSO genes only & Penalized Cox: lambda.min & 200.0 & 55.0 & 0.556 (0.050) & 0.579 (0.125) & 0.573 (0.105) & 0.486 (0.124) & 0.196 (0.017)
\\
SIS--Cox--LASSO & SIS--Cox--LASSO genes only & Penalized Cox: lambda.1se & 200.0 & 25.8 & 0.518 (0.017) & 0.569 (0.107) & 0.510 (0.033) & 0.464 (0.069) & 0.192 (0.011)
\\
\bottomrule
\end{tabular}
}
\end{table}

\subsection{Selection Stability}
Selection was not fully stable across the five outer training folds. Under full-iteration S3VS, \textit{SLAMF7} and \textit{PLA2G2D} were selected in four folds, \textit{PTPRCAP} and \textit{CCR2} were selected in three folds, and \textit{ARHGEF26}, \textit{CHGB}, \textit{DAP}, \textit{HOXB1}, \textit{LAX1}, \textit{MYH7B}, \textit{PYY}, \textit{SEC14L5}, and \textit{SRY} were selected in two folds. Under the one-iteration specification, \textit{PLA2G2D}, \textit{PTPRCAP}, and \textit{SLAMF7} were selected in three folds and \textit{CCR2} was selected in two folds. These results are summarized in Table~\ref{tab:ovarian-selection-stability}, and the full-iteration recurrence pattern is displayed in Figure~\ref{fig:ovarian-selection-stability}. Because no gene was selected in all five folds, the recurrent genes should be regarded as exploratory candidates rather than a definitive prognostic signature.

\begin{table}[htbp]
\centering
\small
\caption{Genes selected in at least two of the five outer TCGA training folds. The table describes selection stability, not a definitive biological signature.}
\label{tab:ovarian-selection-stability}
\begin{tabular}{llr}
\toprule
Method & Gene & Outer folds selected
\\
\midrule
S3VS full iteration & PLA2G2D & 4
\\
S3VS full iteration & SLAMF7 & 4
\\
S3VS full iteration & CCR2 & 3
\\
S3VS full iteration & PTPRCAP & 3
\\
S3VS full iteration & ARHGEF26 & 2
\\
S3VS full iteration & CHGB & 2
\\
S3VS full iteration & DAP & 2
\\
S3VS full iteration & HOXB1 & 2
\\
S3VS full iteration & LAX1 & 2
\\
S3VS full iteration & MYH7B & 2
\\
S3VS full iteration & PYY & 2
\\
S3VS full iteration & SEC14L5 & 2
\\
S3VS full iteration & SRY & 2
\\
S3VS one iteration & PLA2G2D & 3
\\
S3VS one iteration & PTPRCAP & 3
\\
S3VS one iteration & SLAMF7 & 3
\\
S3VS one iteration & CCR2 & 2
\\
\bottomrule
\end{tabular}
\end{table}

\begin{figure}[t]
\centering
\includegraphics[width=0.82\textwidth]{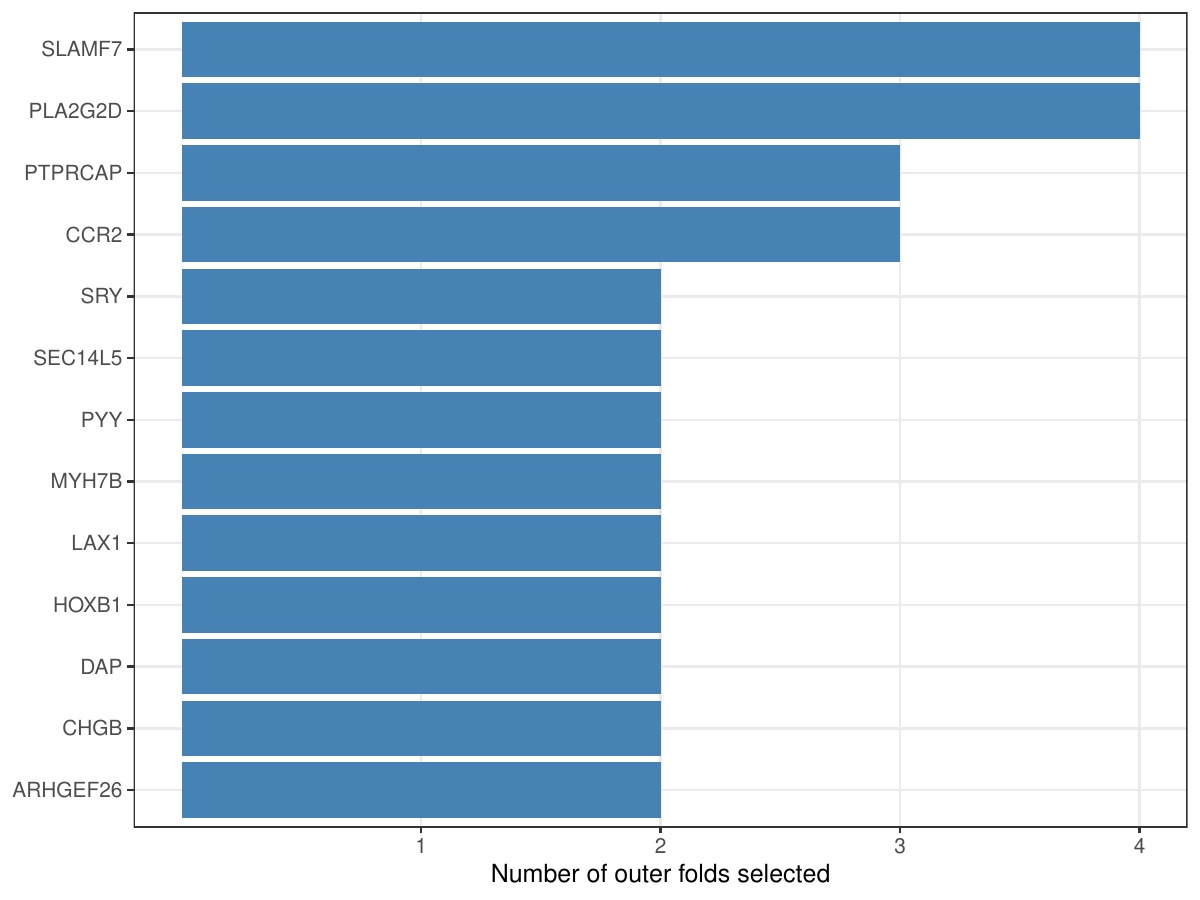}
\caption{Selection stability for genes selected by full-iteration S3VS in at least two of the five outer TCGA training folds.}
\label{fig:ovarian-selection-stability}
\end{figure}

\subsection{Internal Validation}
The primary five-fold outer-validation results are shown in Table~\ref{tab:ovarian-internal-primary}. The clinical baseline achieved a mean C-index of 0.627 and an IBS of 0.172. Full-iteration S3VS combined with clinical variables gave the highest internal C-index (0.640) and 1-year AUC (0.754), but its IBS was higher (0.182). The one-iteration combined model and the SIS--Cox--LASSO combined model did not improve C-index over the clinical baseline. Gene-only models were weaker. The full 1-, 3-, and 5-year AUC results are reported in Table~\ref{tab:ovarian-internal-primary}.

Figure~\ref{fig:ovarian-primary-performance} shows the same pattern. Relative to the clinical baseline, full-iteration S3VS plus clinical variables increased the mean C-index by 0.013 and the 1-year AUC by 0.057, while the 3-year AUC was nearly unchanged and the IBS increased from 0.172 to 0.182. Table~\ref{tab:ovarian-internal-all} also shows the effect of post-selection shrinkage. For the full-iteration clinical-plus-gene model, the mean C-index increased from 0.604 for ordinary Cox to 0.615 at $\lambda_{\min}$ and 0.640 at $\lambda_{1\mathrm{se}}$, while IBS decreased from 0.214 to 0.200 and 0.182. For SIS--Cox--LASSO plus clinical variables, the corresponding C-indices were 0.521, 0.612, and 0.617, and the IBS values were 0.336, 0.187, and 0.179. These results support post-selection shrinkage but do not show a consistent advantage of the expression variables over the clinical baseline.

\subsection{External Validation in GSE9891}
The external-validation results are reported in Table~\ref{tab:ovarian-external-primary}. The clinical baseline achieved a C-index of 0.668 and an IBS of 0.165. SIS--Cox--LASSO combined with clinical variables had the highest external C-index (0.672), but its IBS was higher (0.176), and its 3-year AUC was slightly lower than that of the clinical baseline. Full-iteration S3VS combined with clinical variables achieved a C-index of 0.645 and an IBS of 0.267. The one-iteration combined model achieved a C-index of 0.666 and an IBS of 0.167. Gene-only models were weaker, and the one-iteration gene-only model was at chance-level discrimination. Full AUC values are reported in Table~\ref{tab:ovarian-external-primary}.

Figure~\ref{fig:ovarian-primary-performance} shows that the internal ranking did not carry over fully to the external cohort. Full-iteration S3VS plus clinical variables had the highest internal C-index, whereas SIS--Cox--LASSO plus clinical variables had the highest external C-index by a small margin. The external SIS--Cox--LASSO model improved AUC at 1 and 5 years relative to the clinical baseline, but its 3-year AUC was slightly lower, and its IBS was higher. Full-iteration S3VS plus clinical variables also had a higher 1-year AUC than the clinical baseline but a substantially higher IBS. Thus, gains in discrimination at particular prediction times did not consistently translate into lower overall prediction error or better performance across cohorts.

\subsection{Discussion}
The ovarian-cancer analysis highlights two main points. First, clinical adjustment and post-selection shrinkage were important for prediction. Second, an internal gain in discrimination did not necessarily lead to lower overall prediction error or better external performance. The selected expression variables also varied across the outer folds. The recurrent genes in Table~\ref{tab:ovarian-selection-stability} and Figure~\ref{fig:ovarian-selection-stability} should therefore be viewed as exploratory candidates for further study rather than as a validated gene panel.

Several limitations should be kept in mind. The discovery cohort is small relative to the number and correlation of the expression variables, only five outer folds were used, and the external cohort contains tumor subtypes that do not exactly match those in TCGA. External validation was also limited to one cohort measured on a different microarray platform. The analysis should therefore be viewed as an illustration of the S3VS selection and prediction procedure rather than as evidence for a clinically ready ovarian-cancer risk model.

\section{Conclusion}
\label{sec:conclusion}

This paper develops S3VS as a general variable-selection framework for settings in which $p$ greatly exceeds $n$. Unlike purely marginal screening, S3VS combines two complementary sources of information. Predictor--outcome association identifies leading variables, while predictor--predictor association is used to form smaller local sets in which a model-specific selector is applied. Iteration then allows additional candidate variables to be considered after earlier selections and exclusions. This common architecture supports implementations for linear, generalized linear, and survival models, including accelerated failure time and Cox proportional hazards models, while allowing the screening scores, selection procedures, and outcome updates to be tailored to each model.

This framework also places the present work in relation to earlier structured screen-and-select methodology. The GWASinlps formulation of \citet{sanyal_et_al_2019} introduced the leading-variable and leading-set perspective for incorporating predictor dependence into high-dimensional screening. S3VS develops that idea into a unified iterative methodology with a broader design space for leading-variable and leading-set construction, aggregation, and within-set selection. The resulting methodology combines this general algorithmic framework with a focused asymptotic analysis, extensive finite-sample evaluation, and an openly available software implementation.

Within this broader framework, the theoretical analysis focuses on one specified one-step linear configuration using top-$d_L$ marginal-score leading variables, top-$m_n$ absolute-correlation leading sets, a within-set selector satisfying a conditional setwise sure-screening property, and active-preserving aggregation. This configuration isolates the role of proxy structure. Under the stated conditions, an active predictor with a weak marginal signal can be recovered through a more strongly associated proxy when the active predictor lies in the proxy's leading set and the within-set selector retains it. When the proxy signal is stronger than the active predictor's own marginal signal, the derived failure-probability bound can decrease faster than the corresponding SIS bound, and the post-screening candidate pool can be smaller. The independent-predictor counterexample shows that, without useful proxy structure, the one-step S3VS screen can reduce to SIS and has no strict rate advantage. The theory therefore identifies both the mechanism through which structured screening can help and the conditions under which no strict advantage should be expected. Because the analysis concerns this specified one-step linear configuration, analogous guarantees for later iterations, data-dependent updating, generalized linear models, and survival models remain outside its scope.

The simulations show the same dependence on structure and selector choice. In the linear study, later S3VS iterations substantially increased recovery relative to the first iteration, but the resulting trade-off differed across selectors. Full S3VS--LASSO recovered many active variables but selected large models with high FDP. MCP and SCAD produced much smaller models with low FDP but also low TPR. Full S3VS--NLP provided an intermediate selection trade-off and had the lowest or nearly lowest independent-test MSE in the primary comparison. In the logistic study, S3VS had higher TPR, smaller selected models, lower FDP, and lower test log-loss than SIS--logistic--LASSO in the correlated-proxy setting. Without useful proxies, SIS--logistic--LASSO had higher TPR, while S3VS remained much sparser and had lower test log-loss. These results show that the benefit of the structured search depends on the predictor structure rather than being uniform across settings. The Cox simulation was interpreted more cautiously. To make the primary comparison fair, S3VS was restricted to one iteration so that its iteration budget matched the single-pass SIS--Cox--LASSO procedure. Across both predictor scenarios, one-iteration S3VS was more conservative, selecting substantially smaller models with lower unconditional FDP but also lower TPR. Its predictive performance was modestly better in the correlated-proxy setting but did not show an advantage in the no-proxy setting. The full-iteration Cox results are therefore best viewed as a secondary analysis rather than as a matched method-to-method comparison.

The ovarian-cancer analysis also argues for cautious interpretation. Full-iteration S3VS combined with clinical variables had the highest internal C-index and 1-year AUC, but it did not improve the integrated Brier score over the clinical baseline. In external validation, SIS--Cox--LASSO combined with clinical variables had the highest C-index by a small margin, while the clinical baseline had the lowest IBS. Gene-only models were weaker, and no gene was selected in all five outer folds. The selected expression variables may therefore improve discrimination at some prediction times without producing a stable or consistently better prognostic model across cohorts.

Several limitations should be noted. The theory applies only to a one-step linear S3VS operator under explicit proxy, correlation, selector, and aggregation conditions. Cross-validated LASSO does not by itself guarantee the assumed conditional sure-screening property. S3VS also does not provide formal false-discovery-rate control, so FDP is used only as an empirical simulation measure. The simulation studies are finite, and the primary Cox analysis deliberately uses one S3VS iteration to match a single-pass comparator. The ovarian discovery cohort is small relative to the number of expression variables, only five outer folds were used, and external validation was limited to one cohort measured on a different microarray platform with a somewhat different subtype composition. These limitations should be kept in mind when interpreting the numerical differences among methods.

Future work could study theory for the full iterative procedure, data-driven choices of screening thresholds, faster construction of leading sets for extremely large $p$, and grouped or pathway-based leading sets when such information is available. Post-selection inference and uncertainty also remain important. In genomic applications, repeated external validation, training-based cross-platform harmonization, calibration assessment, and more formal stability analysis would provide stronger evidence about performance across cohorts.

Overall, S3VS should not be viewed as a replacement for marginal screening or penalized regression in every setting. It provides an additional way to use local predictor dependence before model-specific selection. Its main advantage is expected when a predictor with a clear marginal signal is correlated with an active predictor that is harder to detect directly. The one-step theory formalizes this setting, and the numerical studies show both the possible gains and the accompanying recovery--FDP trade-offs.

\begin{appendices}

\section{Technical Lemmas and Proofs}
\label{app:asymptotic-proofs}

The marginal-screening statement in Assumption~\ref{ass:marginal} is deliberately used as an explicit premise. The results below therefore verify only the correlation-based coverage, local-retention, aggregation, and comparison steps. They do not reinterpret a first-iteration marginal-screening premise as a guarantee for later data-dependent iterations.

\subsection{Auxiliary Lemmas}

\subsubsection{Uniform sample-correlation control}

\begin{lemma}[Uniform correlation concentration]
\label{lem:corr}
Under the iid-row model in \eqref{eq:model} and Assumption~\ref{ass:subg}, there exist constants $C,c,t_0>0$, depending only on the uniform sub-Gaussian bound, with $t_0<1$, such that, for all sufficiently large $n$ and $0<t<t_0$, 

\begin{equation} \mathbb{P}\left\{\max_{1\le j,k\le p}|\widehat\rho_{jk}-\rho_{jk}|>t\right\}\leq Cp^2\exp(-cnt^2). \label{eq:corr-conc} 
\end{equation} 

Consequently, Assumption~\ref{ass:gap} implies 

\begin{equation} \mathbb{P}\left\{\max_{1\le j,k\le p}|\widehat\rho_{jk}-\rho_{jk}|>a_n\right\}\longrightarrow0. \label{eq:corr-uniform-convergence} 
\end{equation}
\end{lemma}

\begin{proof}
Write $X_{ij}$ for the $j$th coordinate in row $i$ of $\bfX$, and define $\overline X_j=n^{-1}\sum_{i=1}^nX_{ij}$ and $\widehat\sigma_{jk}=n^{-1}\sum_{i=1}^n(X_{ij}-\overline X_j)(X_{ik}-\overline X_k)$. Since the coordinates are uniformly sub-Gaussian, $X_jX_k-\mathbb{E}(X_jX_k)$ is uniformly sub-exponential. Bernstein's inequality therefore gives, for $0<u<1$, 

\begin{equation} \mathbb{P}\left\{\left|n^{-1}\sum_{i=1}^nX_{ij}X_{ik}-\mathbb{E}(X_jX_k)\right|>u\right\}\leq2\exp(-cnu^2). \label{eq:product-bernstein} 
\end{equation} 

Sub-Gaussian concentration of the sample means gives $\mathbb{P}\{|\overline X_j|>u\}\leq2\exp(-cnu^2)$. Taking $u=\sqrt{t/2}$ and using $0<t<1$, a union bound gives 

\begin{equation} \mathbb{P}\left\{\max_{j,k}|\overline X_j\overline X_k|>t/2\right\}\leq2p\exp(-cnt)\leq2p^2\exp(-cnt^2). \label{eq:mean-product-bound} 
\end{equation} 

Because $\widehat\sigma_{jk}=n^{-1}\sum_iX_{ij}X_{ik}-\overline X_j\overline X_k$ and the population means are zero, another union bound yields, after adjusting constants, 

\begin{equation} \mathbb{P}\left\{\max_{j,k}|\widehat\sigma_{jk}-\sigma_{jk}|>t\right\}\leq Cp^2\exp(-cnt^2),\qquad \sigma_{jk}=\operatorname{Cov}(X_j,X_k). \label{eq:covariance-concentration} 
\end{equation} 

In particular, since $\sigma_{jj}=1$, with probability at least $1-Cp\exp(-cn)$ all sample variances lie in $[1/2,3/2]$. On this event, the map $(u,v,w)\mapsto u(vw)^{-1/2}$ is Lipschitz on the relevant compact domain, and therefore $\max_{j,k}|\widehat\rho_{jk}-\rho_{jk}|\leq C\max_{j,k}|\widehat\sigma_{jk}-\sigma_{jk}|$. This proves \eqref{eq:corr-conc}. Since $a_n\downarrow0$ and $na_n^2/\log p\longrightarrow\infty$ by Assumption~\ref{ass:gap}, $p^2\exp(-cna_n^2)=\exp\{2\log p-cna_n^2\}\longrightarrow0$, which proves \eqref{eq:corr-uniform-convergence}.
\end{proof}

\subsubsection{Leading-set inclusion}

\begin{lemma}[A proxy leading set contains its associated active predictor]
\label{lem:set-inclusion}
Suppose Assumption~\ref{ass:gap} holds and define 

\begin{equation} \mathcal{E}_\rho(a_n)=\left\{\max_{1\le j,k\le p}|\widehat\rho_{jk}-\rho_{jk}|\leq a_n\right\}. \label{eq:correlation-event} 
\end{equation} 

Fix $j\in\mathcal{M}_*$. If $q(j)\in\mathcal{L}^{(1)}$ and $\mathcal{E}_\rho(a_n)$ occurs, then the leading set constructed around $q(j)$ in \eqref{eq:one-step-leading-set} contains $j$.
\end{lemma}

\begin{proof}
If $q(j)=j$, the conclusion follows from the inclusion of every leader in its own leading set. Suppose $q(j)\neq j$ and write $\ell=q(j)$. By Assumption~\ref{ass:gap}, at most $m_n-1$ indices $k\neq j$ satisfy $|\rho_{k\ell}|\geq|\rho_{j\ell}|-2a_n$. For every other $k$, on $\mathcal{E}_\rho(a_n)$, 

\begin{equation} |\widehat\rho_{k\ell}|\leq|\rho_{k\ell}|+a_n<|\rho_{j\ell}|-a_n\leq|\widehat\rho_{j\ell}|. \label{eq:sample-gap-comparison} 
\end{equation} 

Thus at most $m_n-1$ indices other than $j$ have sample absolute correlation with $X_\ell$ at least as large as that of $X_j$. Hence $j$ has rank at most $m_n$ in the absolute-correlation ranking around $\ell$ and belongs to $\mathcal{S}_r^{(1)}$ for the unique leader position $r$ with $\ell_r=\ell$.
\end{proof}

\subsection{Proof of Theorem~\ref{thm:s3vs}}
\label{app:proof-thm-s3vs}

\begin{proof}
Let $\mathcal{E}_L=\{\mathcal{P}_*\subseteq\mathcal{L}^{(1)}\}$. Define $\pi(j)$ to be the unique position for which $\ell_{\pi(j)}=q(j)$ when $q(j)\in\mathcal{L}^{(1)}$, and set $\pi(j)=1$ otherwise. Define 

\begin{equation} \mathcal{E}_G=\{j\in\mathcal{S}_{\pi(j)}^{(1)}\text{ for every }j\in\mathcal{M}_*\},\qquad \mathcal{E}_V=\{j\in\mathcal{S}_{\pi(j),\mathrm{sel}}^{(1)}\text{ for every }j\in\mathcal{M}_*\}. \label{eq:designated-events} 
\end{equation} 

On $\mathcal{E}_L\cap\mathcal{E}_G\cap\mathcal{E}_V$, every active predictor is selected in at least one leading set, so the active-preserving aggregation property in \eqref{eq:active-preserving} implies 

\begin{equation} \mathcal{E}_L\cap\mathcal{E}_G\cap\mathcal{E}_V\subseteq\{\mathcal{M}_*\subseteq\mathcal{S}_{\mathrm{sel}}^{(1)}\}. \label{eq:event-inclusion} 
\end{equation} 

Therefore, 

\begin{equation} \mathbb{P}\{\mathcal{M}_*\not\subseteq\mathcal{S}_{\mathrm{sel}}^{(1)}\}\leq\mathbb{P}(\mathcal{E}_L^c)+\mathbb{P}(\mathcal{E}_G^c\cap\mathcal{E}_L)+\mathbb{P}(\mathcal{E}_V^c\cap\mathcal{E}_L\cap\mathcal{E}_G). \label{eq:union-main} 
\end{equation}

\textbf{Step 1: proxy leading-variable coverage.} By Assumption~\ref{ass:proxy}, $\mathcal{P}_*$ is a deterministic target set with $|\mathcal{P}_*|\leq d_L$ and $\min_{\ell\in\mathcal{P}_*}|\operatorname{Cov}(X_\ell,Y)|\geq c_Ln^{-\kappa_L}$. Applying Assumption~\ref{ass:marginal} with $(\mathcal{T}_n,d_n,\kappa_A,\theta_A)=(\mathcal{P}_*,d_L,\kappa_L,\theta_L)$ gives, after renaming constants, 

\begin{equation} \mathbb{P}(\mathcal{E}_L^c)\leq C_1\exp\left\{-C_2\frac{n^{1-2\kappa_L}}{\log n}\right\}. \label{eq:EL-bound} 
\end{equation}

\textbf{Step 2: leading-set coverage.} On $\mathcal{E}_L$, every proxy $q(j)$ is a selected leader. Lemma~\ref{lem:set-inclusion} gives $\mathcal{E}_L\cap\mathcal{E}_\rho(a_n)\subseteq\mathcal{E}_G$, and Lemma~\ref{lem:corr} gives, for all sufficiently large $n$, 

\begin{equation} \mathbb{P}(\mathcal{E}_G^c\cap\mathcal{E}_L)\leq\mathbb{P}\{\mathcal{E}_\rho(a_n)^c\}\leq Cp^2\exp(-cna_n^2). \label{eq:EG-bound} 
\end{equation}

\textbf{Step 3: within-set retention.} On $\mathcal{E}_L\cap\mathcal{E}_G$, each active predictor belongs to its designated proxy leading set. If $\mathcal{E}_V$ fails, at least one generated leading set fails to retain an active member. Conditional on $\mathfrak{S}^{(1)}$, Assumption~\ref{ass:selector} and a union bound over at most $d_L$ generated sets give 

\begin{equation} \mathbb{P}(\mathcal{E}_V^c\cap\mathcal{E}_L\cap\mathcal{E}_G\mid\mathfrak{S}^{(1)})\leq C_Vd_L\exp(-c_Vn^\zeta). \label{eq:EV-conditional-bound} 
\end{equation} Taking expectations yields 

\begin{equation} \mathbb{P}(\mathcal{E}_V^c\cap\mathcal{E}_L\cap\mathcal{E}_G)\leq C_Vd_L\exp(-c_Vn^\zeta). \label{eq:EV-bound} 
\end{equation}

\textbf{Step 4: combination.} Substituting \eqref{eq:EL-bound}, \eqref{eq:EG-bound}, and \eqref{eq:EV-bound} into \eqref{eq:union-main} proves \eqref{eq:s3vs-bound} after renaming the correlation constants. The first exponential term tends to zero because $\xi<1-2\kappa_L$, and the two additional conditions in the theorem give \eqref{eq:s3vs-consistency}. The final assertion is immediate from \eqref{eq:s3vs-bound} when both remainder terms are of smaller order than its first term.
\end{proof}

\subsection{Proof of Corollary~\ref{cor:prob}}
\label{app:proof-cor-prob}

\begin{proof}
Set $a_L=1-2\kappa_L$ and $a_M=1-2\kappa_M$. Since $\kappa_M>\kappa_L$, $a_L>a_M$, and the ratio of the two terms in \eqref{eq:probability-bound-improvement} is 

\begin{equation} \frac{\exp\{-C_Ln^{a_L}/\log n\}}{\exp\{-C_Mn^{a_M}/\log n\}}=\exp\left\{-\frac{C_Ln^{a_L}-C_Mn^{a_M}}{\log n}\right\}\longrightarrow0. \label{eq:upper-bound-ratio} 
\end{equation} 

This proves the displayed comparison. The qualification about the complete S3VS bound follows from the assumed negligibility of its correlation and selector remainders. The argument compares only the stated upper-bound expressions.
\end{proof}

\subsection{Proof of Corollary~\ref{cor:dim}}
\label{app:proof-cor-dim}

\begin{proof}
By the union bound for finite sets and the definition in \eqref{eq:s3vs-candidate-union}, 

\begin{equation} D_{\mathrm{S3VS}}=\left|\bigcup_{j=1}^{d_L}\mathcal{S}_j^{(1)}\right|\leq\sum_{j=1}^{d_L}|\mathcal{S}_j^{(1)}|\leq d_Lm_n=O(n^{\theta_L+\alpha}), \label{eq:dim-union-bound} 
\end{equation} 

which proves \eqref{eq:s3vs-candidate-bound}. Therefore $D_{\mathrm{S3VS}}/d_{\mathrm{SIS}}=O(n^{\theta_L+\alpha-\theta_M})\longrightarrow0$ whenever $\theta_L+\alpha<\theta_M$. A value of $\theta_L$ satisfying both $2\kappa_L+\tau<\theta_L<1$ and $\theta_L+\alpha<\theta_M$ exists whenever $2\kappa_L+\tau+\alpha<\theta_M$, proving \eqref{eq:dim-improve-existence}. With common positive slack, $\theta_L=2\kappa_L+\tau+\epsilon$ and $\theta_M=2\kappa_M+\tau+\epsilon$, so $\theta_L+\alpha<\theta_M$ is equivalent to $\alpha<2(\kappa_M-\kappa_L)$, proving \eqref{eq:alpha-cond}. Finally, $m_n/d_{\mathrm{SIS}}=O(n^{\alpha-\theta_M})\longrightarrow0$ whenever $\alpha<\theta_M$.
\end{proof}

\subsection{Proof of Corollary~\ref{cor:est}}
\label{app:proof-cor-est}

\begin{proof}
The bounds in \eqref{eq:conditional-estimation-rates} are assumptions for the same post-screening estimator, so their candidate-size dependence is represented by the logarithmic factors. If $q_n=n^{a+o(1)}$ and $d_{\mathrm{SIS}}=n^{b+o(1)}$, then $\log q_n=(a+o(1))\log n$ and $\log d_{\mathrm{SIS}}=(b+o(1))\log n$, which proves \eqref{eq:log-constant-improvement}. Taking square roots also gives the constant factor $\sqrt{a/b}$ stated in the corollary. If $\log q_n=o(\log d_{\mathrm{SIS}})$, then the S3VS logarithmic factor is of smaller order as stated. No estimator-specific rate is derived from Theorem~\ref{thm:s3vs}.
\end{proof}

\subsection{Proof of Proposition~\ref{prop:no-free}}
\label{app:proof-prop-no-free}

\begin{proof}
Let $\bfx\sim N(\bfzero,\bfI_p)$ and $Y=\bfx^T\bfbeta+\varepsilon$, where $\varepsilon\sim N(0,\sigma^2)$ is independent of $\bfx$. Let $s=|\mathcal{M}_*|$ be fixed, set $\beta_j=c_jn^{-\kappa}$ for $j\in\mathcal{M}_*$ with $0<c_-\leq|c_j|\leq c_+<\infty$, set $\beta_k=0$ for $k\notin\mathcal{M}_*$, and assume $0\leq\kappa<1/2$. Then $\bfSigma=\bfI_p$, $\tau=0$, and 

\begin{equation} \operatorname{Var}(Y)=\sigma^2+\|\bfbeta\|_2^2=\sigma^2+n^{-2\kappa}\sum_{j\in\mathcal{M}_*}c_j^2=O(1). \label{eq:no-proxy-variance} 
\end{equation} 

For $j\in\mathcal{M}_*$ and $k\notin\mathcal{M}_*$, respectively, 

\begin{equation} \operatorname{Cov}(X_j,Y)=\beta_j=c_jn^{-\kappa},\qquad \operatorname{Cov}(X_k,Y)=0. \label{eq:no-proxy-scores} 
\end{equation} 

Thus the active marginal exponent is $\kappa_M=\kappa$, while any proxy satisfying the positive signal condition \eqref{eq:proxy-strength} must be active and has exponent at least $\kappa$. Indeed, if an active proxy had $\kappa_L<\kappa$, then $|c_j|n^{-(\kappa-\kappa_L)}\geq c_L$ would have to hold for all large $n$, which is impossible. Hence 

\begin{equation} \kappa_L\geq\kappa_M. \label{eq:no-proxy-exponent-improvement} 
\end{equation} 

Moreover, $\operatorname{Corr}(X_j,X_k)=0$ for $j\neq k$. With the valid choice $m_n=1$, each leading set is exactly its leader, so the leading-variable screening stage has the same candidate set as SIS and cannot recover an active predictor omitted by that stage through a proxy leading set. This design therefore supplies no structural proxy advantage, and a uniform strict improvement theorem requires additional assumptions such as Assumptions~\ref{ass:proxy} and \ref{ass:gap}.
\end{proof}

\end{appendices}

\bibliographystyle{jasa}
\bibliography{mybib}  

\clearpage

\begin{center}
{\bfseries \huge \noindent Supplementary Materials\\[10pt] \par 
\Large ``Structured Screen-and-Select for Ultra-High-Dimensional Variable Selection'' by Nilotpal Sanyal and Padmore N. Prempeh}
\end{center}

\begingroup
\setlength{\parskip}{0pt}
\addtocontents{toc}{\protect\setcounter{tocdepth}{2}}
\setcounter{section}{0}
\setcounter{subsection}{0}
\tableofcontents
\endgroup

\section{Additional Linear-Model Simulation Results}
\label{app:lm-supplement}

This section provides the design keys, calibration and diagnostic summaries, runtime results, and scenario-specific performance results for the linear-model simulation. The main text reports the pooled primary comparison and summarizes the runtime results. The tables and figures below provide the full supporting results, including the prespecified scenario-specific summaries.

\subsection{Methods, scenarios, and calibration}
Table~\ref{tab:LM-method-key} records the three analysis approaches for each selector: full S3VS, the first S3VS iteration extracted from the same full run, and one-pass SIS. Table~\ref{tab:LM-scenario-key} records the 13 design scenarios, including the proxy constructions and the covariance and non-Gaussian structures. The realized training SNRs in Table~\ref{tab:LM-snr-calibration} agree with the target values $0.25$, $1$, and $4$ for the weak, moderate, and strong regimes, respectively. Figure~\ref{fig:LM-supp-snr} shows the same agreement.

\begin{table}[htbp]
\centering
\small
\caption{Named methods used in the linear-model simulation.}
\label{tab:LM-method-key}
\begin{adjustbox}{max width=\textwidth,keepaspectratio}
\begin{tabular}{lrrr}
\toprule
Method & Approach & Selector & Description \\
\midrule
S3VS-full-LASSO & S3VS-full & LASSO & Full S3VS with LASSO using one leading variable, leading sets of size five, liberal aggregation, conservative-from-the-beginning removal, $m_{\max}=100$, and $n_{\mathrm{skip}}=3$. \\
S3VS-full-SCAD & S3VS-full & SCAD & Full S3VS with SCAD using one leading variable, leading sets of size five, liberal aggregation, conservative-from-the-beginning removal, $m_{\max}=100$, and $n_{\mathrm{skip}}=3$. \\
S3VS-full-MCP & S3VS-full & MCP & Full S3VS with MCP using one leading variable, leading sets of size five, liberal aggregation, conservative-from-the-beginning removal, $m_{\max}=100$, and $n_{\mathrm{skip}}=3$. \\
S3VS-full-NLP & S3VS-full & NLP & Full S3VS with NLP using one leading variable, leading sets of size five, liberal aggregation, conservative-from-the-beginning removal, $m_{\max}=100$, and $n_{\mathrm{skip}}=3$. \\
S3VS-one-step-LASSO & S3VS-one-step & LASSO & First S3VS iteration with LASSO derived from the full run \\
S3VS-one-step-SCAD & S3VS-one-step & SCAD & First S3VS iteration with SCAD derived from the full run \\
S3VS-one-step-MCP & S3VS-one-step & MCP & First S3VS iteration with MCP derived from the full run \\
S3VS-one-step-NLP & S3VS-one-step & NLP & First S3VS iteration with NLP derived from the full run \\
SIS-LASSO & SIS & LASSO & SIS to d=min(200,n-1), followed by LASSO \\
SIS-SCAD & SIS & SCAD & One-pass SIS retaining $d=\min(200,n-1)$ predictors, followed by SCAD \\
SIS-MCP & SIS & MCP & One-pass SIS retaining $d=\min(200,n-1)$ predictors, followed by MCP \\
SIS-NLP & SIS & NLP & One-pass SIS retaining $d=\min(200,n-1)$ predictors, followed by NLP \\
\bottomrule
\end{tabular}
\end{adjustbox}
\end{table}

\begin{table}[htbp]
\centering
\small
\caption{Named proxy, no-proxy, covariance, and non-Gaussian design scenarios.}
\label{tab:LM-scenario-key}
\begin{tabular}{lrr}
\toprule
Scenario & Design structure & Description \\
\midrule
baseline-block & unequal-block & Unequal positive-correlation blocks \\
proxy-cancellation & proxy & Active cancellation with a null proxy \\
shared-proxy & proxy & Several active variables sharing one proxy \\
null-proxy & proxy & One active variable with several null proxies \\
no-proxy & independent & Independent predictors: no useful proxy \\
misleading-null & proxy & Correlated null leader and null neighborhood \\
ar1 & ar1 & AR(1) dependence \\
factor & factor & Latent-factor dependence \\
sparse-graph & sparse-graph & Sparse pairwise graph \\
overlapping-modules & overlap & Overlapping factor modules \\
negative-dependence & negative & Signed factor loadings and negative correlations \\
heavy-tailed & heavy-tailed & Heavy-tailed non-Gaussian predictors \\
contaminated & contaminated & Gaussian predictors with row/feature contamination \\
\bottomrule
\end{tabular}
\end{table}

\begin{table}[htbp]
\centering
\small
\caption{Calibration of the weak, moderate, and strong linear-model signal regimes. SNR is the variance of the training linear predictor divided by the noise variance.}
\label{tab:LM-snr-calibration}
\begin{tabular}{lrrrr}
\toprule
Signal & Replicates & Target SNR & Empirical training SNR & Response variance \\
\midrule
moderate & 11700 & 1.000 (0.000) & 1.000 (0.000) & 1.999 (0.001) \\
strong & 11700 & 4.000 (0.000) & 4.000 (0.000) & 4.999 (0.002) \\
weak & 11700 & 0.250 (0.000) & 0.250 (0.000) & 1.253 (0.001) \\
\bottomrule
\end{tabular}
\end{table}

\begin{figure}[htbp]
\centering
\includegraphics[width=0.75\textwidth]{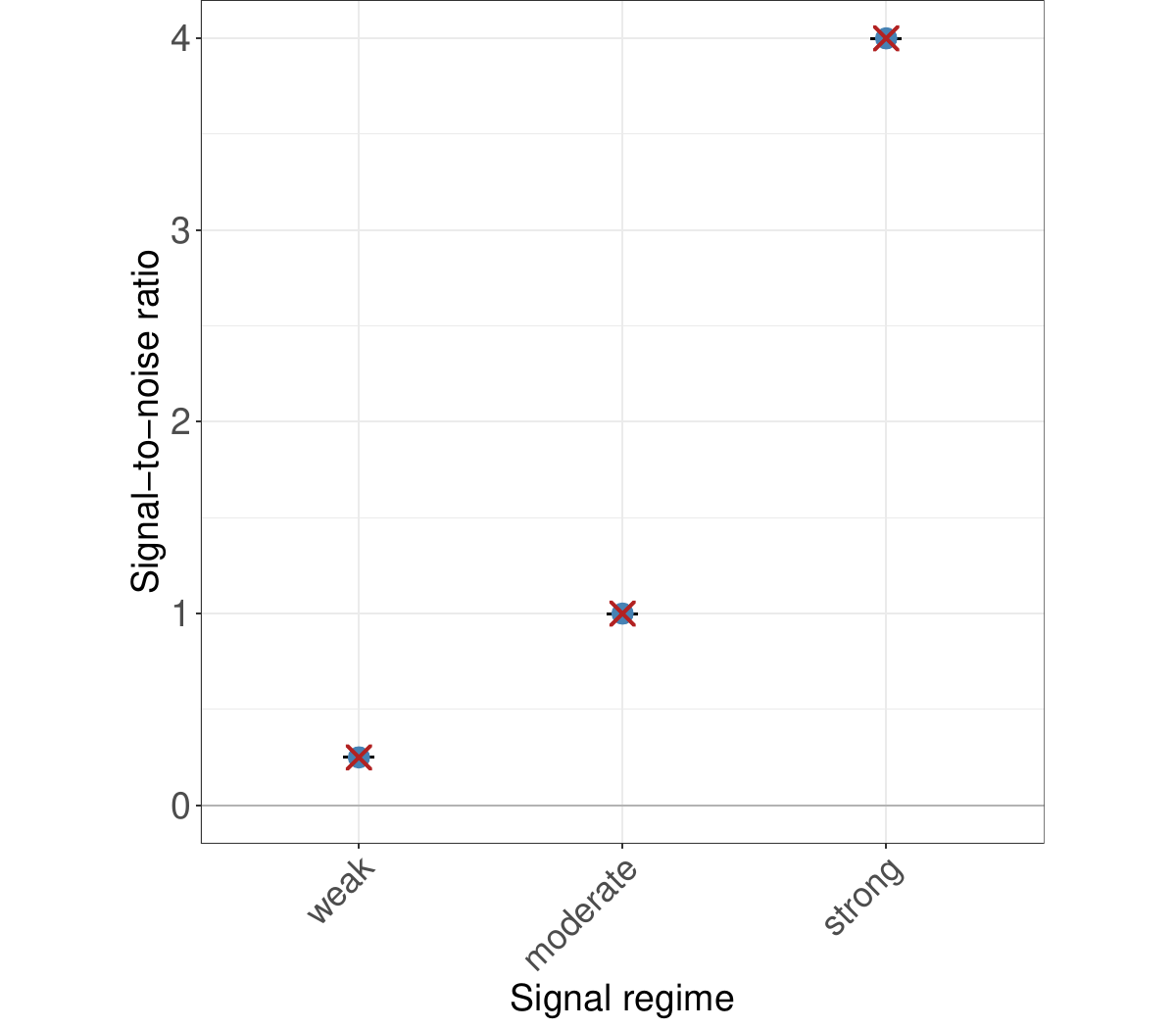}
\caption{Calibration of the weak, moderate, and strong linear-model signal regimes. Blue points are empirical training SNRs and red crosses are target SNRs.}
\label{fig:LM-supp-snr}
\end{figure}

\subsection{Runtime}
The runtime results for the separately timed full S3VS and one-pass SIS analyses are reported in Supplementary Table~\ref{tab:LM-runtime}. Runtime generally increases with $p$ at fixed $n$, whereas its dependence on $n$ is not monotone because the cross-validation and optimization workloads vary across replicates. Full S3VS--SCAD and full S3VS--MCP are the least expensive S3VS variants, whereas full S3VS--NLP and SIS--NLP are more costly in several settings. First-iteration S3VS results are not assigned separate runtimes because they are extracted from the corresponding full S3VS runs.

{
\small
\setlength{\LTcapwidth}{\textwidth}
\begin{longtable}{llrrrr}
\caption{Runtime for separately timed full S3VS and one-pass SIS analyses. First-iteration S3VS rows are derived from the full run and are not assigned a separate runtime.}
\label{tab:LM-runtime}\\
\toprule
Approach & Selector & $n$ & $p$ & Replicates & Runtime (minutes), MCSE \\
\midrule
\endfirsthead
\multicolumn{6}{c}%
{{\tablename\ \thetable{} -- continued from previous page}}\\
\toprule
Approach & Selector & $n$ & $p$ & Replicates & Runtime (minutes), MCSE \\
\midrule
\endhead
\midrule
\multicolumn{6}{r}{{Continued on next page}}\\
\endfoot
\bottomrule
\endlastfoot
S3VS-full & LASSO & 200 & 10000 & 3900 & 0.039 (0.000) \\
S3VS-full & LASSO & 200 & 15000 & 3900 & 0.049 (0.000) \\
S3VS-full & LASSO & 200 & 20000 & 3900 & 0.060 (0.000) \\
S3VS-full & LASSO & 500 & 10000 & 3900 & 0.076 (0.000) \\
S3VS-full & LASSO & 500 & 15000 & 3900 & 0.107 (0.000) \\
S3VS-full & LASSO & 500 & 20000 & 3900 & 0.133 (0.000) \\
S3VS-full & LASSO & 800 & 10000 & 3900 & 0.112 (0.000) \\
S3VS-full & LASSO & 800 & 15000 & 3900 & 0.154 (0.000) \\
S3VS-full & LASSO & 800 & 20000 & 3900 & 0.199 (0.000) \\

S3VS-full & MCP & 200 & 10000 & 3900 & 0.003 (0.000) \\
S3VS-full & MCP & 200 & 15000 & 3900 & 0.004 (0.000) \\
S3VS-full & MCP & 200 & 20000 & 3900 & 0.005 (0.000) \\
S3VS-full & MCP & 500 & 10000 & 3900 & 0.008 (0.000) \\
S3VS-full & MCP & 500 & 15000 & 3900 & 0.011 (0.000) \\
S3VS-full & MCP & 500 & 20000 & 3900 & 0.015 (0.000) \\
S3VS-full & MCP & 800 & 10000 & 3900 & 0.020 (0.000) \\
S3VS-full & MCP & 800 & 15000 & 3900 & 0.028 (0.000) \\
S3VS-full & MCP & 800 & 20000 & 3900 & 0.037 (0.001) \\

S3VS-full & NLP & 200 & 10000 & 3900 & 0.035 (0.000) \\
S3VS-full & NLP & 200 & 15000 & 3900 & 0.056 (0.000) \\
S3VS-full & NLP & 200 & 20000 & 3900 & 0.081 (0.000) \\
S3VS-full & NLP & 500 & 10000 & 3900 & 0.024 (0.000) \\
S3VS-full & NLP & 500 & 15000 & 3900 & 0.046 (0.000) \\
S3VS-full & NLP & 500 & 20000 & 3900 & 0.081 (0.001) \\
S3VS-full & NLP & 800 & 10000 & 3900 & 0.037 (0.000) \\
S3VS-full & NLP & 800 & 15000 & 3900 & 0.060 (0.000) \\
S3VS-full & NLP & 800 & 20000 & 3900 & 0.087 (0.001) \\

S3VS-full & SCAD & 200 & 10000 & 3900 & 0.003 (0.000) \\
S3VS-full & SCAD & 200 & 15000 & 3900 & 0.004 (0.000) \\
S3VS-full & SCAD & 200 & 20000 & 3900 & 0.005 (0.000) \\
S3VS-full & SCAD & 500 & 10000 & 3900 & 0.009 (0.000) \\
S3VS-full & SCAD & 500 & 15000 & 3900 & 0.012 (0.000) \\
S3VS-full & SCAD & 500 & 20000 & 3900 & 0.015 (0.000) \\
S3VS-full & SCAD & 800 & 10000 & 3900 & 0.021 (0.000) \\
S3VS-full & SCAD & 800 & 15000 & 3900 & 0.029 (0.000) \\
S3VS-full & SCAD & 800 & 20000 & 3900 & 0.038 (0.001) \\

SIS & LASSO & 200 & 10000 & 3900 & 0.005 (0.000) \\
SIS & LASSO & 200 & 15000 & 3900 & 0.005 (0.000) \\
SIS & LASSO & 200 & 20000 & 3900 & 0.005 (0.000) \\
SIS & LASSO & 500 & 10000 & 3900 & 0.004 (0.000) \\
SIS & LASSO & 500 & 15000 & 3900 & 0.004 (0.000) \\
SIS & LASSO & 500 & 20000 & 3900 & 0.004 (0.000) \\
SIS & LASSO & 800 & 10000 & 3900 & 0.004 (0.000) \\
SIS & LASSO & 800 & 15000 & 3900 & 0.004 (0.000) \\
SIS & LASSO & 800 & 20000 & 3900 & 0.005 (0.000) \\

SIS & MCP & 200 & 10000 & 3900 & 0.036 (0.000) \\
SIS & MCP & 200 & 15000 & 3900 & 0.035 (0.000) \\
SIS & MCP & 200 & 20000 & 3900 & 0.035 (0.000) \\
SIS & MCP & 500 & 10000 & 3900 & 0.041 (0.000) \\
SIS & MCP & 500 & 15000 & 3900 & 0.042 (0.000) \\
SIS & MCP & 500 & 20000 & 3900 & 0.042 (0.000) \\
SIS & MCP & 800 & 10000 & 3900 & 0.040 (0.000) \\
SIS & MCP & 800 & 15000 & 3900 & 0.041 (0.000) \\
SIS & MCP & 800 & 20000 & 3900 & 0.042 (0.000) \\

SIS & NLP & 200 & 10000 & 3900 & 0.621 (0.014) \\
SIS & NLP & 200 & 15000 & 3900 & 0.658 (0.014) \\
SIS & NLP & 200 & 20000 & 3900 & 0.691 (0.015) \\
SIS & NLP & 500 & 10000 & 3900 & 0.060 (0.001) \\
SIS & NLP & 500 & 15000 & 3900 & 0.088 (0.001) \\
SIS & NLP & 500 & 20000 & 3900 & 0.103 (0.001) \\
SIS & NLP & 800 & 10000 & 3900 & 0.033 (0.001) \\
SIS & NLP & 800 & 15000 & 3900 & 0.043 (0.001) \\
SIS & NLP & 800 & 20000 & 3900 & 0.056 (0.001) \\

SIS & SCAD & 200 & 10000 & 3900 & 0.037 (0.000) \\
SIS & SCAD & 200 & 15000 & 3900 & 0.036 (0.000) \\
SIS & SCAD & 200 & 20000 & 3900 & 0.036 (0.000) \\
SIS & SCAD & 500 & 10000 & 3900 & 0.035 (0.000) \\
SIS & SCAD & 500 & 15000 & 3900 & 0.036 (0.000) \\
SIS & SCAD & 500 & 20000 & 3900 & 0.037 (0.000) \\
SIS & SCAD & 800 & 10000 & 3900 & 0.038 (0.000) \\
SIS & SCAD & 800 & 15000 & 3900 & 0.039 (0.000) \\
SIS & SCAD & 800 & 20000 & 3900 & 0.040 (0.000) \\

\end{longtable}
\normalsize
}

\subsection{Proxy and covariance-structure results}
The prespecified proxy and no-proxy results are given in Table~\ref{tab:LM-proxy-scenarios}. In these settings, full S3VS--LASSO attains TPRs around $0.76$ but has mean FDP around $0.81$ and selects about 101 variables. Full S3VS--NLP retains similar TPR, roughly $0.74$--$0.77$, with mean FDP around $0.28$ and mean model sizes of about 26--28. SIS--NLP is somewhat sparser and has lower mean FDP, about $0.20$--$0.22$, but its TPR is lower, about $0.69$--$0.73$. SIS--LASSO, SIS--SCAD, and SIS--MCP recover many active variables only by selecting very large models, with mean FDPs close to $0.88$--$0.90$. The one-step S3VS procedures remain near-zero-recovery baselines in these scenarios, which reinforces the role of later iterations.

\begin{table}[htbp]
\centering
\small
\caption{Performance in the prespecified proxy and no-proxy scenarios at $n$ = 800, $p$ = 20,000, and the moderate signal regime.}
\label{tab:LM-proxy-scenarios}
\begin{adjustbox}{max width=\textwidth,max height=0.82\textheight,keepaspectratio}
\begin{tabular}{lrrrrrrrrr}
\toprule
Scenario & Method & \shortstack[c]{Independent\\datasets} & \shortstack[c]{P(nonempty),\\MCSE} & \shortstack[c]{TPR,\\MCSE} & \shortstack[c]{Sure screening,\\MCSE} & \shortstack[c]{FDP | nonempty,\\MCSE} & \shortstack[c]{Mean FDP,\\MCSE} & \shortstack[c]{Mean size,\\MCSE} & \shortstack[c]{Independent-test\\MSE, MCSE} \\
\midrule
misleading-null & S3VS-full-LASSO & 100 & 1.000 (0.000) & 0.762 (0.008) & 0.000 (0.000) & 0.811 (0.002) & 0.811 (0.002) & 100.940 (0.108) & 1.495 (0.013) \\
misleading-null & S3VS-full-MCP & 100 & 0.620 (0.049) & 0.044 (0.004) & 0.000 (0.000) & 0.000 (0.000) & 0.000 (0.000) & 1.100 (0.111) & 1.932 (0.016) \\
misleading-null & S3VS-full-NLP & 100 & 1.000 (0.000) & 0.772 (0.008) & 0.000 (0.000) & 0.286 (0.011) & 0.286 (0.011) & 27.670 (0.503) & 1.236 (0.010) \\
misleading-null & S3VS-full-SCAD & 100 & 0.660 (0.048) & 0.047 (0.005) & 0.000 (0.000) & 0.000 (0.000) & 0.000 (0.000) & 1.170 (0.115) & 1.930 (0.016) \\
misleading-null & S3VS-one-step-LASSO & 100 & 1.000 (0.000) & 0.040 (0.000) & 0.000 (0.000) & 0.000 (0.000) & 0.000 (0.000) & 1.000 (0.000) & 1.979 (0.013) \\
misleading-null & S3VS-one-step-MCP & 100 & 0.600 (0.049) & 0.024 (0.002) & 0.000 (0.000) & 0.000 (0.000) & 0.000 (0.000) & 0.600 (0.049) & 1.979 (0.013) \\
misleading-null & S3VS-one-step-NLP & 100 & 1.000 (0.000) & 0.040 (0.000) & 0.000 (0.000) & 0.000 (0.000) & 0.000 (0.000) & 1.000 (0.000) & 1.979 (0.013) \\
misleading-null & S3VS-one-step-SCAD & 100 & 0.640 (0.048) & 0.026 (0.002) & 0.000 (0.000) & 0.000 (0.000) & 0.000 (0.000) & 0.640 (0.048) & 1.979 (0.013) \\
misleading-null & SIS-LASSO & 100 & 1.000 (0.000) & 0.788 (0.007) & 0.000 (0.000) & 0.892 (0.001) & 0.892 (0.001) & 182.450 (0.378) & 1.590 (0.011) \\
misleading-null & SIS-MCP & 100 & 1.000 (0.000) & 0.784 (0.007) & 0.000 (0.000) & 0.881 (0.001) & 0.881 (0.001) & 166.050 (0.909) & 1.602 (0.012) \\
misleading-null & SIS-NLP & 100 & 1.000 (0.000) & 0.728 (0.008) & 0.000 (0.000) & 0.223 (0.010) & 0.223 (0.010) & 23.880 (0.416) & 1.244 (0.010) \\
misleading-null & SIS-SCAD & 100 & 1.000 (0.000) & 0.786 (0.007) & 0.000 (0.000) & 0.887 (0.001) & 0.887 (0.001) & 174.000 (0.731) & 1.595 (0.012) \\
no-proxy & S3VS-full-LASSO & 100 & 1.000 (0.000) & 0.763 (0.008) & 0.000 (0.000) & 0.811 (0.002) & 0.811 (0.002) & 101.020 (0.112) & 1.540 (0.013) \\
no-proxy & S3VS-full-MCP & 100 & 0.620 (0.049) & 0.038 (0.004) & 0.000 (0.000) & 0.000 (0.000) & 0.000 (0.000) & 0.940 (0.107) & 2.000 (0.012) \\
no-proxy & S3VS-full-NLP & 100 & 1.000 (0.000) & 0.764 (0.008) & 0.010 (0.010) & 0.286 (0.011) & 0.286 (0.011) & 27.210 (0.427) & 1.269 (0.010) \\
no-proxy & S3VS-full-SCAD & 100 & 0.630 (0.049) & 0.038 (0.004) & 0.000 (0.000) & 0.000 (0.000) & 0.000 (0.000) & 0.940 (0.105) & 2.001 (0.012) \\
no-proxy & S3VS-one-step-LASSO & 100 & 1.000 (0.000) & 0.040 (0.000) & 0.000 (0.000) & 0.000 (0.000) & 0.000 (0.000) & 1.000 (0.000) & 2.028 (0.010) \\
no-proxy & S3VS-one-step-MCP & 100 & 0.590 (0.049) & 0.024 (0.002) & 0.000 (0.000) & 0.000 (0.000) & 0.000 (0.000) & 0.590 (0.049) & 2.028 (0.010) \\
no-proxy & S3VS-one-step-NLP & 100 & 1.000 (0.000) & 0.040 (0.000) & 0.000 (0.000) & 0.000 (0.000) & 0.000 (0.000) & 1.000 (0.000) & 2.028 (0.010) \\
no-proxy & S3VS-one-step-SCAD & 100 & 0.590 (0.049) & 0.024 (0.002) & 0.000 (0.000) & 0.000 (0.000) & 0.000 (0.000) & 0.590 (0.049) & 2.028 (0.010) \\
no-proxy & SIS-LASSO & 100 & 1.000 (0.000) & 0.790 (0.007) & 0.000 (0.000) & 0.892 (0.001) & 0.892 (0.001) & 183.730 (0.395) & 1.637 (0.011) \\
no-proxy & SIS-MCP & 100 & 1.000 (0.000) & 0.787 (0.007) & 0.000 (0.000) & 0.882 (0.001) & 0.882 (0.001) & 167.490 (1.054) & 1.648 (0.011) \\
no-proxy & SIS-NLP & 100 & 1.000 (0.000) & 0.718 (0.008) & 0.000 (0.000) & 0.222 (0.012) & 0.222 (0.012) & 23.530 (0.401) & 1.281 (0.010) \\
no-proxy & SIS-SCAD & 100 & 1.000 (0.000) & 0.788 (0.007) & 0.000 (0.000) & 0.887 (0.001) & 0.887 (0.001) & 174.880 (0.870) & 1.642 (0.011) \\
null-proxy & S3VS-full-LASSO & 100 & 1.000 (0.000) & 0.764 (0.007) & 0.000 (0.000) & 0.811 (0.002) & 0.811 (0.002) & 101.170 (0.111) & 1.526 (0.013) \\
null-proxy & S3VS-full-MCP & 100 & 0.640 (0.048) & 0.045 (0.004) & 0.000 (0.000) & 0.008 (0.008) & 0.005 (0.005) & 1.130 (0.111) & 1.962 (0.013) \\
null-proxy & S3VS-full-NLP & 100 & 1.000 (0.000) & 0.767 (0.008) & 0.000 (0.000) & 0.278 (0.010) & 0.278 (0.010) & 26.960 (0.409) & 1.245 (0.009) \\
null-proxy & S3VS-full-SCAD & 100 & 0.690 (0.046) & 0.048 (0.004) & 0.000 (0.000) & 0.007 (0.007) & 0.005 (0.005) & 1.220 (0.112) & 1.959 (0.013) \\
null-proxy & S3VS-one-step-LASSO & 100 & 1.000 (0.000) & 0.040 (0.000) & 0.000 (0.000) & 0.000 (0.000) & 0.000 (0.000) & 1.000 (0.000) & 2.007 (0.011) \\
null-proxy & S3VS-one-step-MCP & 100 & 0.630 (0.049) & 0.025 (0.002) & 0.000 (0.000) & 0.000 (0.000) & 0.000 (0.000) & 0.630 (0.049) & 2.007 (0.011) \\
null-proxy & S3VS-one-step-NLP & 100 & 1.000 (0.000) & 0.040 (0.000) & 0.000 (0.000) & 0.000 (0.000) & 0.000 (0.000) & 1.000 (0.000) & 2.007 (0.011) \\
null-proxy & S3VS-one-step-SCAD & 100 & 0.650 (0.048) & 0.026 (0.002) & 0.000 (0.000) & 0.000 (0.000) & 0.000 (0.000) & 0.650 (0.048) & 2.007 (0.011) \\
null-proxy & SIS-LASSO & 100 & 1.000 (0.000) & 0.795 (0.007) & 0.000 (0.000) & 0.890 (0.001) & 0.890 (0.001) & 180.270 (0.427) & 1.608 (0.012) \\
null-proxy & SIS-MCP & 100 & 1.000 (0.000) & 0.792 (0.007) & 0.000 (0.000) & 0.878 (0.001) & 0.878 (0.001) & 163.310 (1.012) & 1.618 (0.012) \\
null-proxy & SIS-NLP & 100 & 1.000 (0.000) & 0.718 (0.008) & 0.000 (0.000) & 0.221 (0.009) & 0.221 (0.009) & 23.330 (0.377) & 1.262 (0.010) \\
null-proxy & SIS-SCAD & 100 & 1.000 (0.000) & 0.793 (0.007) & 0.000 (0.000) & 0.884 (0.001) & 0.884 (0.001) & 171.610 (0.794) & 1.613 (0.012) \\
proxy-cancellation & S3VS-full-LASSO & 100 & 1.000 (0.000) & 0.759 (0.007) & 0.000 (0.000) & 0.812 (0.002) & 0.812 (0.002) & 100.960 (0.113) & 1.508 (0.012) \\
proxy-cancellation & S3VS-full-MCP & 100 & 0.710 (0.046) & 0.054 (0.005) & 0.000 (0.000) & 0.000 (0.000) & 0.000 (0.000) & 1.350 (0.122) & 1.943 (0.014) \\
proxy-cancellation & S3VS-full-NLP & 100 & 1.000 (0.000) & 0.758 (0.008) & 0.000 (0.000) & 0.299 (0.011) & 0.299 (0.011) & 27.610 (0.488) & 1.253 (0.009) \\
proxy-cancellation & S3VS-full-SCAD & 100 & 0.690 (0.046) & 0.055 (0.005) & 0.000 (0.000) & 0.000 (0.000) & 0.000 (0.000) & 1.370 (0.131) & 1.940 (0.014) \\
proxy-cancellation & S3VS-one-step-LASSO & 100 & 1.000 (0.000) & 0.040 (0.000) & 0.000 (0.000) & 0.000 (0.000) & 0.000 (0.000) & 1.000 (0.000) & 2.007 (0.010) \\
proxy-cancellation & S3VS-one-step-MCP & 100 & 0.680 (0.047) & 0.027 (0.002) & 0.000 (0.000) & 0.000 (0.000) & 0.000 (0.000) & 0.680 (0.047) & 2.007 (0.010) \\
proxy-cancellation & S3VS-one-step-NLP & 100 & 1.000 (0.000) & 0.040 (0.000) & 0.000 (0.000) & 0.000 (0.000) & 0.000 (0.000) & 1.000 (0.000) & 2.007 (0.010) \\
proxy-cancellation & S3VS-one-step-SCAD & 100 & 0.660 (0.048) & 0.026 (0.002) & 0.000 (0.000) & 0.000 (0.000) & 0.000 (0.000) & 0.660 (0.048) & 2.007 (0.010) \\
proxy-cancellation & SIS-LASSO & 100 & 1.000 (0.000) & 0.760 (0.006) & 0.000 (0.000) & 0.895 (0.001) & 0.895 (0.001) & 181.810 (0.432) & 1.641 (0.010) \\
proxy-cancellation & SIS-MCP & 100 & 1.000 (0.000) & 0.758 (0.006) & 0.000 (0.000) & 0.886 (0.001) & 0.886 (0.001) & 166.310 (0.966) & 1.655 (0.010) \\
proxy-cancellation & SIS-NLP & 100 & 1.000 (0.000) & 0.689 (0.007) & 0.000 (0.000) & 0.219 (0.013) & 0.219 (0.013) & 22.630 (0.414) & 1.271 (0.009) \\
proxy-cancellation & SIS-SCAD & 100 & 1.000 (0.000) & 0.759 (0.006) & 0.000 (0.000) & 0.890 (0.001) & 0.890 (0.001) & 173.270 (0.835) & 1.646 (0.010) \\
shared-proxy & S3VS-full-LASSO & 100 & 1.000 (0.000) & 0.759 (0.008) & 0.000 (0.000) & 0.812 (0.002) & 0.812 (0.002) & 101.090 (0.109) & 1.488 (0.011) \\
shared-proxy & S3VS-full-MCP & 100 & 0.650 (0.048) & 0.053 (0.006) & 0.000 (0.000) & 0.000 (0.000) & 0.000 (0.000) & 1.320 (0.141) & 1.944 (0.015) \\
shared-proxy & S3VS-full-NLP & 100 & 1.000 (0.000) & 0.741 (0.008) & 0.000 (0.000) & 0.285 (0.010) & 0.285 (0.010) & 26.330 (0.408) & 1.247 (0.009) \\
shared-proxy & S3VS-full-SCAD & 100 & 0.640 (0.048) & 0.052 (0.006) & 0.000 (0.000) & 0.003 (0.003) & 0.002 (0.002) & 1.320 (0.142) & 1.941 (0.016) \\
shared-proxy & S3VS-one-step-LASSO & 100 & 1.000 (0.000) & 0.040 (0.000) & 0.000 (0.000) & 0.000 (0.000) & 0.000 (0.000) & 1.000 (0.000) & 2.004 (0.011) \\
shared-proxy & S3VS-one-step-MCP & 100 & 0.640 (0.048) & 0.026 (0.002) & 0.000 (0.000) & 0.000 (0.000) & 0.000 (0.000) & 0.640 (0.048) & 2.004 (0.011) \\
shared-proxy & S3VS-one-step-NLP & 100 & 1.000 (0.000) & 0.040 (0.000) & 0.000 (0.000) & 0.000 (0.000) & 0.000 (0.000) & 1.000 (0.000) & 2.004 (0.011) \\
shared-proxy & S3VS-one-step-SCAD & 100 & 0.610 (0.049) & 0.024 (0.002) & 0.000 (0.000) & 0.000 (0.000) & 0.000 (0.000) & 0.610 (0.049) & 2.004 (0.011) \\
shared-proxy & SIS-LASSO & 100 & 1.000 (0.000) & 0.802 (0.006) & 0.000 (0.000) & 0.888 (0.001) & 0.888 (0.001) & 179.610 (0.470) & 1.596 (0.010) \\
shared-proxy & SIS-MCP & 100 & 1.000 (0.000) & 0.792 (0.006) & 0.000 (0.000) & 0.879 (0.001) & 0.879 (0.001) & 164.050 (0.934) & 1.607 (0.010) \\
shared-proxy & SIS-NLP & 100 & 1.000 (0.000) & 0.706 (0.007) & 0.000 (0.000) & 0.201 (0.010) & 0.201 (0.010) & 22.420 (0.337) & 1.244 (0.009) \\
shared-proxy & SIS-SCAD & 100 & 1.000 (0.000) & 0.796 (0.006) & 0.000 (0.000) & 0.884 (0.001) & 0.884 (0.001) & 171.570 (0.833) & 1.598 (0.010) \\
\bottomrule
\end{tabular}
\end{adjustbox}
\end{table}

Table~\ref{tab:LM-covariance-performance} and Figure~\ref{fig:LM-supp-covariance} show that recovery varies substantially across dependence and tail structures. Full S3VS--NLP has relatively low TPR in the unequal-block and factor-related settings and much higher TPR in the contaminated and heavy-tailed settings. The relative performance of NLP, LASSO, MCP, and SCAD therefore depends on the design structure. Full S3VS--LASSO again selects about 100 variables and has high FDP, whereas full S3VS--MCP and full S3VS--SCAD keep FDP low at the cost of lower TPR. Table~\ref{tab:LM-method-effects} and Figure~\ref{fig:LM-supp-method-effects} give the corresponding pooled approach-by-selector summaries.

{
\setlength{\LTcapwidth}{\textwidth}
\setlength{\tabcolsep}{3.5pt}
\footnotesize
\begin{longtable}{llrrrrr}
\caption{Performance across named covariance and non-Gaussian design structures at $n = 800$, $p = 20{,}000$, and the moderate signal regime.}
\label{tab:LM-covariance-performance}\\

\toprule
Scenario & Method &
\shortstack[c]{Replicates} &
\shortstack[c]{TPR,\\MCSE} &
\shortstack[c]{Mean FDP,\\MCSE} &
\shortstack[c]{Mean size,\\MCSE} &
\shortstack[c]{Independent-test\\MSE, MCSE} \\
\midrule
\endfirsthead

\multicolumn{7}{c}%
{{\tablename\ \thetable{} -- continued from previous page}}\\
\toprule
Scenario & Method &
\shortstack[c]{Independent\\datasets} &
\shortstack[c]{TPR,\\MCSE} &
\shortstack[c]{Mean FDP,\\MCSE} &
\shortstack[c]{Mean size,\\MCSE} &
\shortstack[c]{Independent-test\\MSE, MCSE} \\
\midrule
\endhead

\midrule
\multicolumn{7}{r}{{Continued on next page}}\\
\endfoot

\bottomrule
\endlastfoot

ar1 & S3VS-full-LASSO & 100 & 0.671 (0.017) & 0.834 (0.004) & 100.940 (0.117) & 1.524 (0.011) \\
ar1 & S3VS-full-MCP & 100 & 0.175 (0.008) & 0.000 (0.000) & 4.370 (0.204) & 1.278 (0.012) \\
ar1 & S3VS-full-NLP & 100 & 0.357 (0.014) & 0.359 (0.015) & 14.000 (0.433) & 1.218 (0.008) \\
ar1 & S3VS-full-SCAD & 100 & 0.235 (0.010) & 0.000 (0.000) & 5.880 (0.242) & 1.252 (0.012) \\
ar1 & S3VS-one-step-LASSO & 100 & 0.040 (0.000) & 0.000 (0.000) & 1.000 (0.000) & 2.011 (0.009) \\
ar1 & S3VS-one-step-MCP & 100 & 0.040 (0.000) & 0.000 (0.000) & 1.000 (0.000) & 2.011 (0.009) \\
ar1 & S3VS-one-step-NLP & 100 & 0.040 (0.000) & 0.000 (0.000) & 1.000 (0.000) & 2.011 (0.009) \\
ar1 & S3VS-one-step-SCAD & 100 & 0.040 (0.000) & 0.000 (0.000) & 1.000 (0.000) & 2.011 (0.009) \\
ar1 & SIS-LASSO & 100 & 0.502 (0.009) & 0.896 (0.002) & 121.210 (0.930) & 1.464 (0.017) \\
ar1 & SIS-MCP & 100 & 0.342 (0.010) & 0.890 (0.003) & 78.330 (1.430) & 1.480 (0.018) \\
ar1 & SIS-NLP & 100 & 0.246 (0.009) & 0.140 (0.015) & 7.360 (0.289) & 1.233 (0.012) \\
ar1 & SIS-SCAD & 100 & 0.356 (0.012) & 0.904 (0.003) & 91.670 (1.203) & 1.474 (0.018) \\

baseline-block & S3VS-full-LASSO & 100 & 0.129 (0.006) & 0.968 (0.001) & 101.050 (0.118) & 1.944 (0.043) \\
baseline-block & S3VS-full-MCP & 100 & 0.051 (0.002) & 0.025 (0.013) & 1.320 (0.062) & 1.827 (0.035) \\
baseline-block & S3VS-full-NLP & 100 & 0.091 (0.005) & 0.515 (0.028) & 6.070 (0.330) & 1.426 (0.032) \\
baseline-block & S3VS-full-SCAD & 100 & 0.056 (0.003) & 0.035 (0.014) & 1.490 (0.085) & 1.806 (0.036) \\
baseline-block & S3VS-one-step-LASSO & 100 & 0.038 (0.001) & 0.040 (0.020) & 1.000 (0.000) & 2.019 (0.011) \\
baseline-block & S3VS-one-step-MCP & 100 & 0.038 (0.001) & 0.010 (0.010) & 0.970 (0.017) & 2.019 (0.011) \\
baseline-block & S3VS-one-step-NLP & 100 & 0.038 (0.001) & 0.020 (0.014) & 0.980 (0.014) & 2.019 (0.011) \\
baseline-block & S3VS-one-step-SCAD & 100 & 0.038 (0.001) & 0.010 (0.010) & 0.970 (0.017) & 2.019 (0.011) \\
baseline-block & SIS-LASSO & 100 & 0.588 (0.027) & 0.728 (0.025) & 88.920 (5.040) & 1.382 (0.031) \\
baseline-block & SIS-MCP & 100 & 0.400 (0.031) & 0.724 (0.027) & 61.970 (4.651) & 1.458 (0.032) \\
baseline-block & SIS-NLP & 100 & 0.387 (0.033) & 0.153 (0.017) & 11.320 (0.924) & 1.213 (0.021) \\
baseline-block & SIS-SCAD & 100 & 0.473 (0.030) & 0.755 (0.023) & 74.750 (4.693) & 1.423 (0.032) \\

contaminated & S3VS-full-LASSO & 100 & 0.747 (0.008) & 0.815 (0.002) & 100.840 (0.106) & 1.530 (0.011) \\
contaminated & S3VS-full-MCP & 100 & 0.039 (0.004) & 0.000 (0.000) & 0.970 (0.111) & 1.965 (0.014) \\
contaminated & S3VS-full-NLP & 100 & 0.755 (0.009) & 0.277 (0.010) & 26.440 (0.385) & 1.248 (0.009) \\
contaminated & S3VS-full-SCAD & 100 & 0.036 (0.004) & 0.000 (0.000) & 0.890 (0.105) & 1.968 (0.014) \\
contaminated & S3VS-one-step-LASSO & 100 & 0.040 (0.000) & 0.000 (0.000) & 1.000 (0.000) & 2.008 (0.011) \\
contaminated & S3VS-one-step-MCP & 100 & 0.021 (0.002) & 0.000 (0.000) & 0.520 (0.050) & 2.008 (0.011) \\
contaminated & S3VS-one-step-NLP & 100 & 0.040 (0.000) & 0.000 (0.000) & 1.000 (0.000) & 2.008 (0.011) \\
contaminated & S3VS-one-step-SCAD & 100 & 0.021 (0.002) & 0.000 (0.000) & 0.520 (0.050) & 2.008 (0.011) \\
contaminated & SIS-LASSO & 100 & 0.792 (0.007) & 0.886 (0.001) & 173.880 (1.109) & 1.597 (0.011) \\
contaminated & SIS-MCP & 100 & 0.787 (0.007) & 0.867 (0.003) & 153.810 (2.525) & 1.606 (0.011) \\
contaminated & SIS-NLP & 100 & 0.711 (0.008) & 0.206 (0.011) & 22.790 (0.396) & 1.257 (0.009) \\
contaminated & SIS-SCAD & 100 & 0.790 (0.007) & 0.876 (0.002) & 162.070 (2.121) & 1.596 (0.011) \\

factor & S3VS-full-LASSO & 100 & 0.585 (0.012) & 0.855 (0.003) & 100.790 (0.108) & 1.499 (0.014) \\
factor & S3VS-full-MCP & 100 & 0.101 (0.005) & 0.127 (0.025) & 2.910 (0.127) & 1.724 (0.017) \\
factor & S3VS-full-NLP & 100 & 0.479 (0.013) & 0.282 (0.013) & 16.720 (0.407) & 1.327 (0.011) \\
factor & S3VS-full-SCAD & 100 & 0.099 (0.005) & 0.267 (0.028) & 3.630 (0.169) & 1.712 (0.016) \\
factor & S3VS-one-step-LASSO & 100 & 0.037 (0.001) & 0.070 (0.026) & 1.000 (0.000) & 1.991 (0.011) \\
factor & S3VS-one-step-MCP & 100 & 0.037 (0.001) & 0.070 (0.026) & 0.990 (0.010) & 1.991 (0.011) \\
factor & S3VS-one-step-NLP & 100 & 0.037 (0.001) & 0.070 (0.026) & 1.000 (0.000) & 1.991 (0.011) \\
factor & S3VS-one-step-SCAD & 100 & 0.037 (0.001) & 0.070 (0.026) & 0.990 (0.010) & 1.991 (0.011) \\
factor & SIS-LASSO & 100 & 0.216 (0.010) & 0.690 (0.018) & 21.100 (1.081) & 1.632 (0.015) \\
factor & SIS-MCP & 100 & 0.191 (0.009) & 0.350 (0.025) & 8.490 (0.498) & 1.637 (0.016) \\
factor & SIS-NLP & 100 & 0.154 (0.007) & 0.141 (0.020) & 4.430 (0.178) & 1.655 (0.016) \\
factor & SIS-SCAD & 100 & 0.216 (0.010) & 0.684 (0.018) & 20.290 (0.992) & 1.631 (0.015) \\

heavy-tailed & S3VS-full-LASSO & 100 & 0.755 (0.008) & 0.813 (0.002) & 101.050 (0.109) & 1.474 (0.014) \\
heavy-tailed & S3VS-full-MCP & 100 & 0.038 (0.004) & 0.000 (0.000) & 0.940 (0.093) & 1.952 (0.013) \\
heavy-tailed & S3VS-full-NLP & 100 & 0.765 (0.008) & 0.287 (0.009) & 27.210 (0.421) & 1.229 (0.009) \\
heavy-tailed & S3VS-full-SCAD & 100 & 0.035 (0.004) & 0.000 (0.000) & 0.870 (0.088) & 1.957 (0.012) \\
heavy-tailed & S3VS-one-step-LASSO & 100 & 0.040 (0.000) & 0.000 (0.000) & 1.000 (0.000) & 1.984 (0.011) \\
heavy-tailed & S3VS-one-step-MCP & 100 & 0.024 (0.002) & 0.000 (0.000) & 0.600 (0.049) & 1.984 (0.011) \\
heavy-tailed & S3VS-one-step-NLP & 100 & 0.040 (0.000) & 0.000 (0.000) & 1.000 (0.000) & 1.984 (0.011) \\
heavy-tailed & S3VS-one-step-SCAD & 100 & 0.023 (0.002) & 0.000 (0.000) & 0.580 (0.050) & 1.984 (0.011) \\
heavy-tailed & SIS-LASSO & 100 & 0.791 (0.007) & 0.890 (0.001) & 180.470 (0.709) & 1.587 (0.011) \\
heavy-tailed & SIS-MCP & 100 & 0.788 (0.007) & 0.876 (0.003) & 161.800 (1.828) & 1.597 (0.011) \\
heavy-tailed & SIS-NLP & 100 & 0.716 (0.008) & 0.215 (0.012) & 23.280 (0.434) & 1.238 (0.009) \\
heavy-tailed & SIS-SCAD & 100 & 0.790 (0.007) & 0.883 (0.002) & 170.230 (1.342) & 1.590 (0.011) \\

negative-dependence & S3VS-full-LASSO & 100 & 0.570 (0.011) & 0.859 (0.003) & 101.010 (0.119) & 1.526 (0.013) \\
negative-dependence & S3VS-full-MCP & 100 & 0.086 (0.004) & 0.137 (0.022) & 2.570 (0.101) & 1.783 (0.016) \\
negative-dependence & S3VS-full-NLP & 100 & 0.471 (0.013) & 0.284 (0.013) & 16.490 (0.365) & 1.350 (0.012) \\
negative-dependence & S3VS-full-SCAD & 100 & 0.086 (0.004) & 0.273 (0.028) & 3.300 (0.142) & 1.766 (0.016) \\
negative-dependence & S3VS-one-step-LASSO & 100 & 0.036 (0.001) & 0.100 (0.030) & 1.000 (0.000) & 2.007 (0.012) \\
negative-dependence & S3VS-one-step-MCP & 100 & 0.036 (0.001) & 0.100 (0.030) & 1.000 (0.000) & 2.007 (0.012) \\
negative-dependence & S3VS-one-step-NLP & 100 & 0.036 (0.001) & 0.100 (0.030) & 1.000 (0.000) & 2.007 (0.012) \\
negative-dependence & S3VS-one-step-SCAD & 100 & 0.036 (0.001) & 0.100 (0.030) & 1.000 (0.000) & 2.007 (0.012) \\
negative-dependence & SIS-LASSO & 100 & 0.226 (0.009) & 0.721 (0.017) & 24.880 (1.211) & 1.643 (0.014) \\
negative-dependence & SIS-MCP & 100 & 0.198 (0.009) & 0.396 (0.027) & 9.560 (0.481) & 1.647 (0.015) \\
negative-dependence & SIS-NLP & 100 & 0.154 (0.007) & 0.177 (0.020) & 4.650 (0.166) & 1.671 (0.014) \\
negative-dependence & SIS-SCAD & 100 & 0.226 (0.009) & 0.718 (0.016) & 23.530 (1.091) & 1.643 (0.013) \\

overlapping-modules & S3VS-full-LASSO & 100 & 0.543 (0.013) & 0.866 (0.003) & 101.190 (0.122) & 1.527 (0.014) \\
overlapping-modules & S3VS-full-MCP & 100 & 0.088 (0.004) & 0.180 (0.028) & 2.730 (0.101) & 1.717 (0.016) \\
overlapping-modules & S3VS-full-NLP & 100 & 0.404 (0.014) & 0.291 (0.014) & 14.170 (0.406) & 1.374 (0.013) \\
overlapping-modules & S3VS-full-SCAD & 100 & 0.089 (0.005) & 0.341 (0.031) & 3.720 (0.150) & 1.699 (0.016) \\
overlapping-modules & S3VS-one-step-LASSO & 100 & 0.035 (0.001) & 0.130 (0.034) & 1.000 (0.000) & 1.999 (0.014) \\
overlapping-modules & S3VS-one-step-MCP & 100 & 0.034 (0.001) & 0.130 (0.034) & 0.990 (0.010) & 1.999 (0.014) \\
overlapping-modules & S3VS-one-step-NLP & 100 & 0.035 (0.001) & 0.130 (0.034) & 1.000 (0.000) & 1.999 (0.014) \\
overlapping-modules & S3VS-one-step-SCAD & 100 & 0.034 (0.001) & 0.130 (0.034) & 0.990 (0.010) & 1.999 (0.014) \\
overlapping-modules & SIS-LASSO & 100 & 0.204 (0.009) & 0.735 (0.017) & 23.770 (1.084) & 1.622 (0.014) \\
overlapping-modules & SIS-MCP & 100 & 0.175 (0.008) & 0.415 (0.027) & 8.800 (0.426) & 1.620 (0.015) \\
overlapping-modules & SIS-NLP & 100 & 0.134 (0.007) & 0.213 (0.023) & 4.220 (0.143) & 1.653 (0.016) \\
overlapping-modules & SIS-SCAD & 100 & 0.204 (0.009) & 0.720 (0.018) & 22.400 (0.965) & 1.620 (0.014) \\

sparse-graph & S3VS-full-LASSO & 100 & 0.668 (0.011) & 0.835 (0.003) & 100.910 (0.108) & 1.547 (0.011) \\
sparse-graph & S3VS-full-MCP & 100 & 0.214 (0.005) & 0.000 (0.000) & 5.350 (0.131) & 1.436 (0.013) \\
sparse-graph & S3VS-full-NLP & 100 & 0.414 (0.009) & 0.390 (0.013) & 17.390 (0.422) & 1.313 (0.010) \\
sparse-graph & S3VS-full-SCAD & 100 & 0.326 (0.008) & 0.002 (0.001) & 8.170 (0.199) & 1.366 (0.013) \\
sparse-graph & S3VS-one-step-LASSO & 100 & 0.040 (0.000) & 0.000 (0.000) & 1.000 (0.000) & 2.001 (0.013) \\
sparse-graph & S3VS-one-step-MCP & 100 & 0.040 (0.000) & 0.000 (0.000) & 1.000 (0.000) & 2.001 (0.013) \\
sparse-graph & S3VS-one-step-NLP & 100 & 0.040 (0.000) & 0.000 (0.000) & 1.000 (0.000) & 2.001 (0.013) \\
sparse-graph & S3VS-one-step-SCAD & 100 & 0.040 (0.000) & 0.000 (0.000) & 1.000 (0.000) & 2.001 (0.013) \\
sparse-graph & SIS-LASSO & 100 & 0.605 (0.009) & 0.912 (0.001) & 171.640 (0.545) & 1.618 (0.014) \\
sparse-graph & SIS-MCP & 100 & 0.564 (0.009) & 0.907 (0.002) & 152.260 (1.087) & 1.638 (0.015) \\
sparse-graph & SIS-NLP & 100 & 0.361 (0.007) & 0.219 (0.016) & 11.950 (0.288) & 1.301 (0.011) \\
sparse-graph & SIS-SCAD & 100 & 0.579 (0.009) & 0.910 (0.001) & 161.340 (0.938) & 1.623 (0.014) \\
\end{longtable}
\normalsize
}

\begin{table}[htbp]
\centering
\small
\caption{Descriptive approach-by-selector summaries pooled over named scenarios at $n=800$, $p=20{,}000$, and the moderate signal regime.}
\label{tab:LM-method-effects}
\begin{adjustbox}{max width=\textwidth,keepaspectratio}
\begin{tabular}{lrrrrrrr}
\toprule
Approach & Selector & \shortstack[c]{Independent\\datasets} & \shortstack[c]{TPR,\\MCSE} & \shortstack[c]{Mean FDP,\\MCSE} & \shortstack[c]{P(nonempty),\\MCSE} & \shortstack[c]{Mean size,\\MCSE} & \shortstack[c]{Independent-test\\MSE, MCSE} \\
\midrule
S3VS-full & LASSO & 1300 & 0.652 (0.005) & 0.839 (0.001) & 1.000 (0.000) & 100.997 (0.031) & 1.548 (0.006) \\
S3VS-full & MCP & 1300 & 0.079 (0.002) & 0.036 (0.004) & 0.798 (0.011) & 2.077 (0.051) & 1.805 (0.008) \\
S3VS-full & NLP & 1300 & 0.580 (0.007) & 0.317 (0.004) & 0.998 (0.001) & 21.098 (0.225) & 1.287 (0.004) \\
S3VS-full & SCAD & 1300 & 0.092 (0.003) & 0.071 (0.005) & 0.798 (0.011) & 2.613 (0.072) & 1.792 (0.008) \\
S3VS-one-step & LASSO & 1300 & 0.039 (0.000) & 0.026 (0.004) & 1.000 (0.000) & 1.000 (0.000) & 2.004 (0.003) \\
S3VS-one-step & MCP & 1300 & 0.030 (0.000) & 0.024 (0.004) & 0.785 (0.011) & 0.785 (0.011) & 2.004 (0.003) \\
S3VS-one-step & NLP & 1300 & 0.039 (0.000) & 0.025 (0.004) & 0.998 (0.001) & 0.998 (0.001) & 2.004 (0.003) \\
S3VS-one-step & SCAD & 1300 & 0.030 (0.000) & 0.024 (0.004) & 0.785 (0.011) & 0.785 (0.011) & 2.004 (0.003) \\
SIS & LASSO & 1300 & 0.605 (0.007) & 0.840 (0.004) & 1.000 (0.000) & 131.826 (1.864) & 1.586 (0.005) \\
SIS & MCP & 1300 & 0.566 (0.008) & 0.756 (0.007) & 1.000 (0.000) & 112.479 (1.879) & 1.601 (0.005) \\
SIS & NLP & 1300 & 0.494 (0.008) & 0.196 (0.004) & 0.998 (0.001) & 15.830 (0.254) & 1.348 (0.006) \\
SIS & SCAD & 1300 & 0.581 (0.007) & 0.837 (0.004) & 1.000 (0.000) & 122.432 (1.810) & 1.592 (0.005) \\
\bottomrule
\end{tabular}
\end{adjustbox}
\end{table}

\begin{figure}[htbp]
\centering
\includegraphics[width=\textwidth]{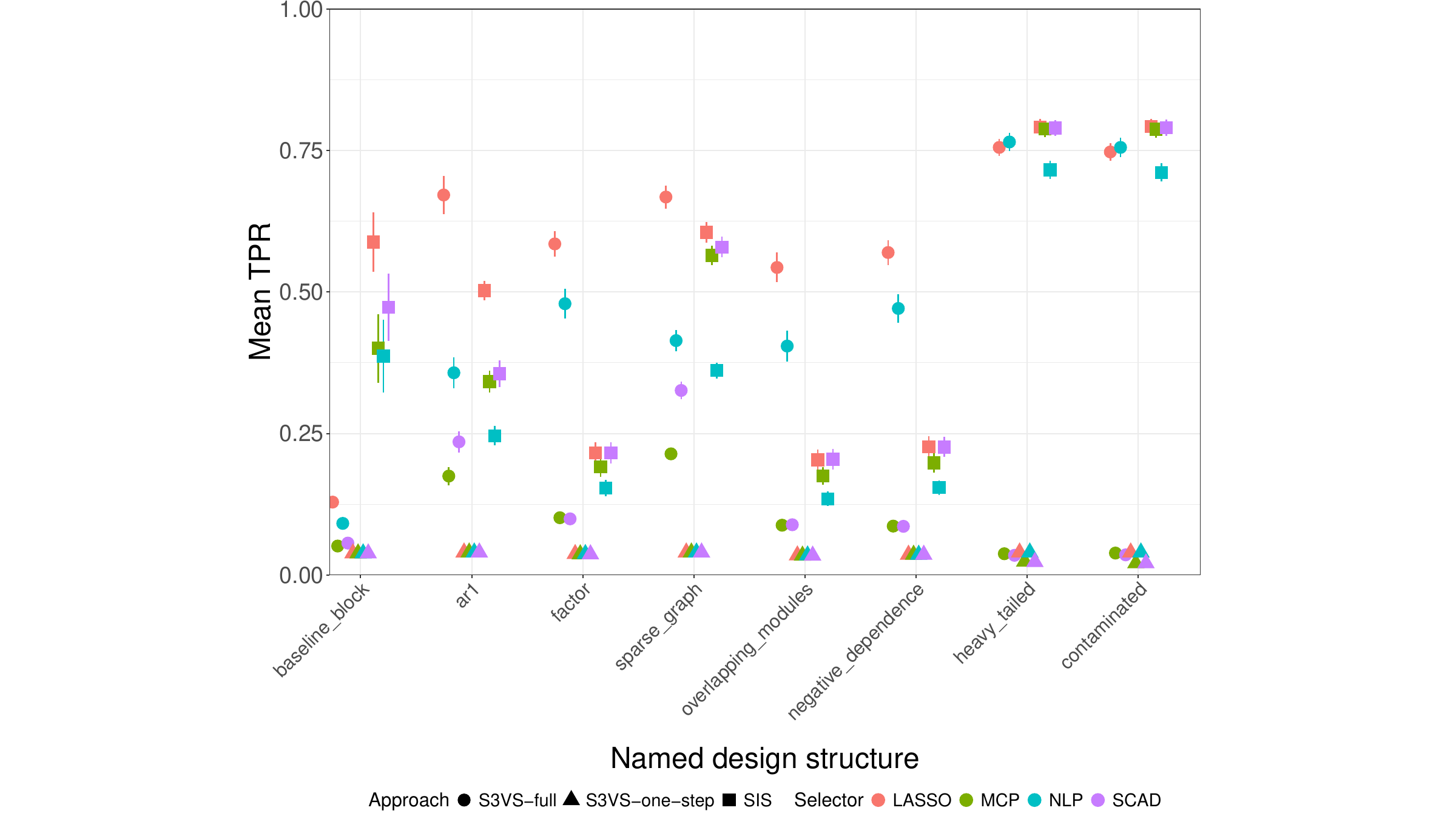}
\caption{TPR across the named covariance and non-Gaussian design structures at $n=800$, $p=20{,}000$, and the moderate signal level. Points are means and bars are 95\% Monte Carlo intervals.}
\label{fig:LM-supp-covariance}
\end{figure}

\begin{figure}[htbp]
\centering
\includegraphics[width=\textwidth]{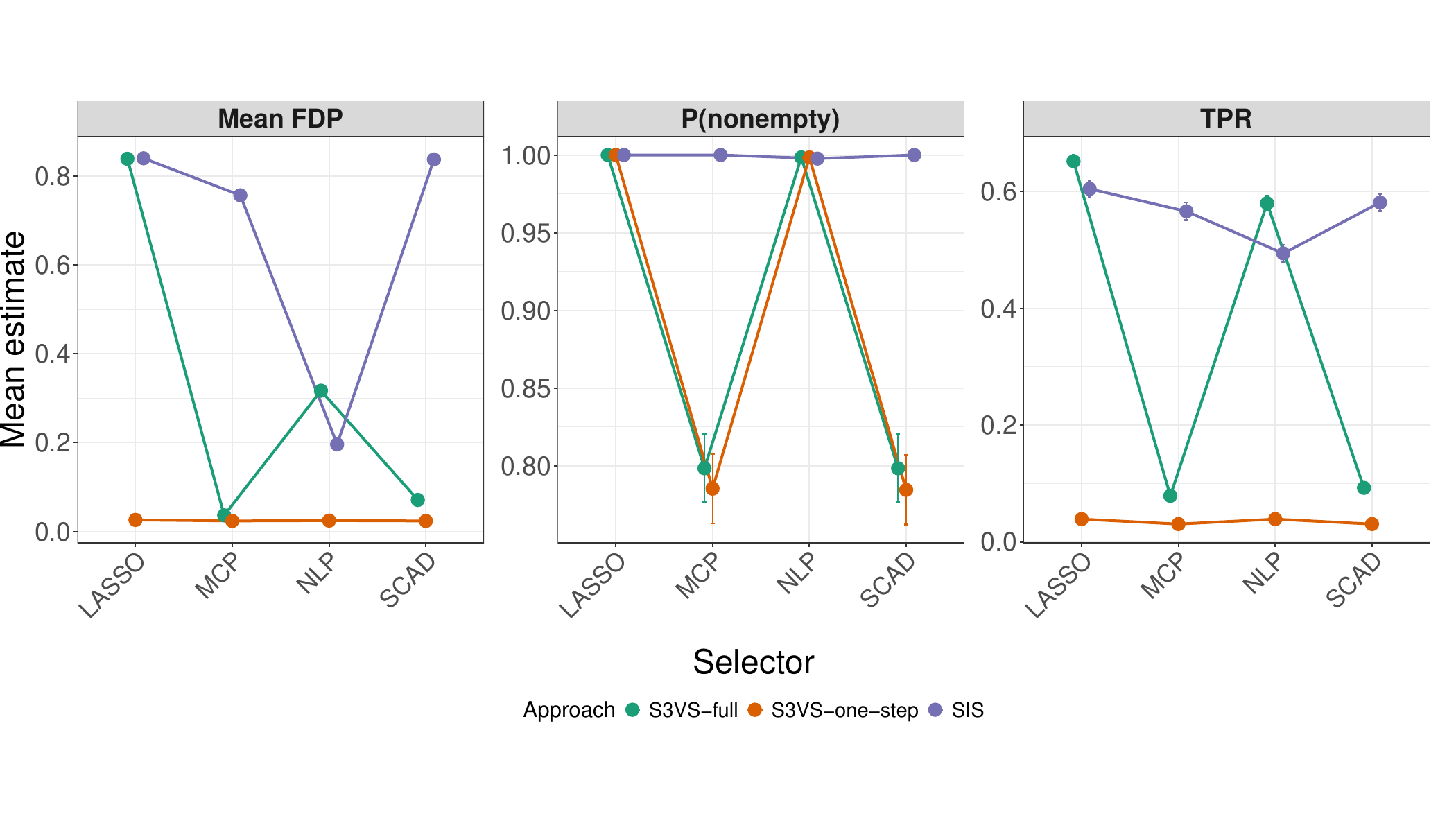}
\caption{Descriptive approach-by-selector summaries pooled over the named scenarios at the primary $n$, $p$, and moderate-signal setting.}
\label{fig:LM-supp-method-effects}
\end{figure}

\subsection{Error control, coefficient error, and diagnostics}
Figure~\ref{fig:LM-supp-fdp-heatmap} reports unconditional FDP over signal regimes and scenarios, with empty selections counted as FDP zero. Figure~\ref{fig:LM-supp-fdp-conditional} removes that structural zero by displaying FDP only among nonempty selections. The two displays should be read together: the low unconditional FDP of a method that often returns an empty or nearly empty set is not equivalent to a low FDP conditional on making a selection. In particular, the conditional distributions expose the high false-discovery burden of the large LASSO and SIS models, while the small one-step models remain low-FDP but low-TPR procedures.

\begin{figure}[htbp]
\centering
\includegraphics[width=\textwidth]{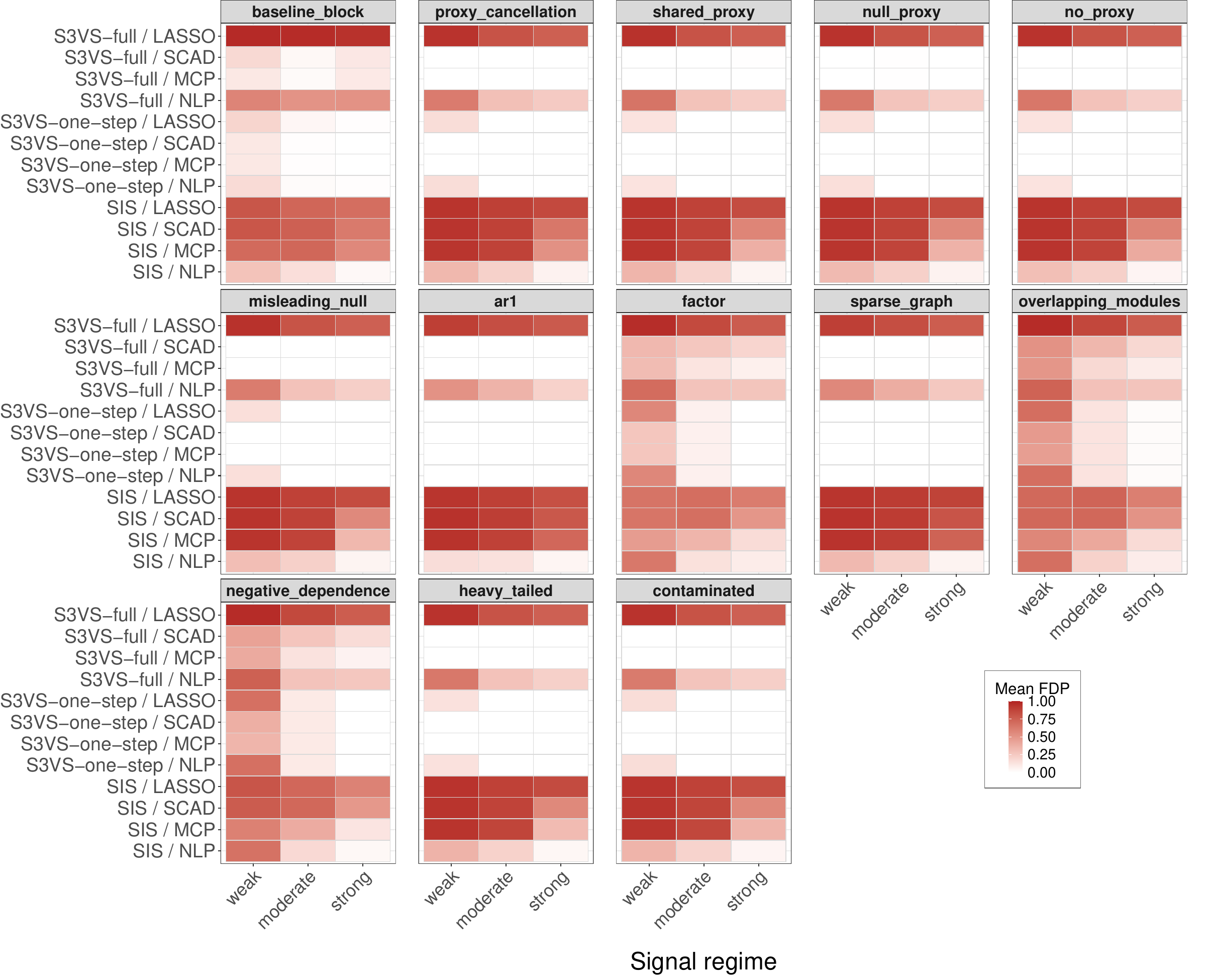}
\caption{Named-method unconditional FDP heat map at $n=800$ and $p=20{,}000$. Empty selections contribute zero to the displayed FDP.}
\label{fig:LM-supp-fdp-heatmap}
\end{figure}

\begin{figure}[htbp]
\centering
\includegraphics[width=\textwidth]{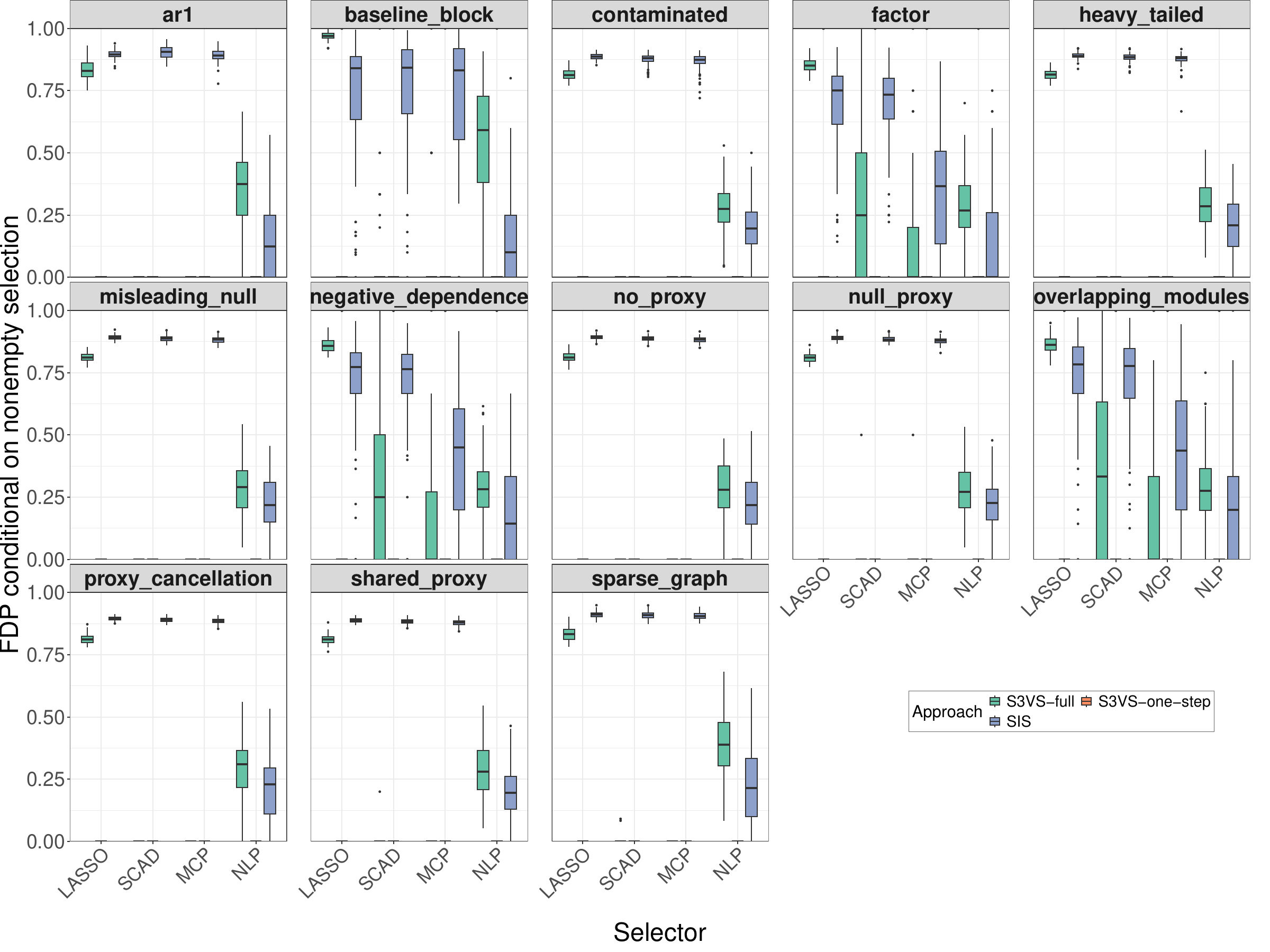}
\caption{Conditional FDP distributions among nonempty selections at $n=800$, $p=20{,}000$, and the moderate signal level.}
\label{fig:LM-supp-fdp-conditional}
\end{figure}

Figure~\ref{fig:LM-supp-coefficient-error} gives replicate-level coefficient errors after the common post-selection LASSO refit. Full S3VS often lowers coefficient error relative to its first-iteration result by adding signal variables, although the size of the improvement depends on the design structure. No selector has the smallest coefficient error in every setting. Table~\ref{tab:LM-fit-diagnostics} summarizes fitting and post-selection diagnostics. Before interpreting that table, the meanings of selection success, history success, and postfit success should be stated explicitly.

\begin{table}[htbp]
\centering
\small
\caption{Fit, history, post-selection, and iteration diagnostics across all linear-model datasets. ``Independent datasets'' denotes the number of unique simulated datasets analyzed for each method. ``Selection success, MCSE'' is the proportion of datasets for which the reported selection procedure completed successfully, with its Monte Carlo standard error (MCSE). ``History success, MCSE'' is the proportion of full iterative S3VS fits that returned a valid iteration history, with its MCSE. ``Postfit success, MCSE'' is the proportion of datasets for which the common post-selection LASSO refit and independent-test MSE calculation completed successfully, with its MCSE. ``Mean iterations, MCSE'' is the average number of iterations recorded in the full S3VS history, with its MCSE. For the one-step S3VS and one-pass SIS procedures, history-success and mean-iteration entries are not applicable and are reported as NA.}
\label{tab:LM-fit-diagnostics}
\begin{adjustbox}{max width=\textwidth,keepaspectratio}
\begin{tabular}{lrrrrr}
\toprule
Method & \shortstack[c]{Independent\\datasets} & \shortstack[c]{Selection success,\\MCSE} & \shortstack[c]{History success,\\MCSE} & \shortstack[c]{Postfit success,\\MCSE} & \shortstack[c]{Mean iterations,\\MCSE} \\
\midrule
S3VS-full-LASSO & 35100 & 1.000 (0.000) & 1.000 (0.000) & 1.000 (0.000) & 101.024 (0.006) \\
S3VS-full-MCP & 35100 & 1.000 (0.000) & 1.000 (0.000) & 0.787 (0.002) & 6.032 (0.026) \\
S3VS-full-NLP & 35100 & 1.000 (0.000) & 1.000 (0.000) & 0.996 (0.000) & 42.586 (0.183) \\
S3VS-full-SCAD & 35100 & 1.000 (0.000) & 1.000 (0.000) & 0.823 (0.002) & 6.407 (0.028) \\
S3VS-one-step-LASSO & 35100 & 1.000 (0.000) & NA & 0.000 (0.000) & NA \\
S3VS-one-step-MCP & 35100 & 1.000 (0.000) & NA & 0.374 (0.003) & NA \\
S3VS-one-step-NLP & 35100 & 1.000 (0.000) & NA & 0.000 (0.000) & NA \\
S3VS-one-step-SCAD & 35100 & 1.000 (0.000) & NA & 0.376 (0.003) & NA \\
SIS-LASSO & 35100 & 1.000 (0.000) & NA & 0.998 (0.000) & NA \\
SIS-MCP & 35100 & 1.000 (0.000) & NA & 0.982 (0.001) & NA \\
SIS-NLP & 35100 & 1.000 (0.000) & NA & 0.962 (0.001) & NA \\
SIS-SCAD & 35100 & 1.000 (0.000) & NA & 0.997 (0.000) & NA \\
\bottomrule
\end{tabular}
\end{adjustbox}
\end{table}

\begin{figure}[htbp]
\centering
\includegraphics[width=\textwidth]{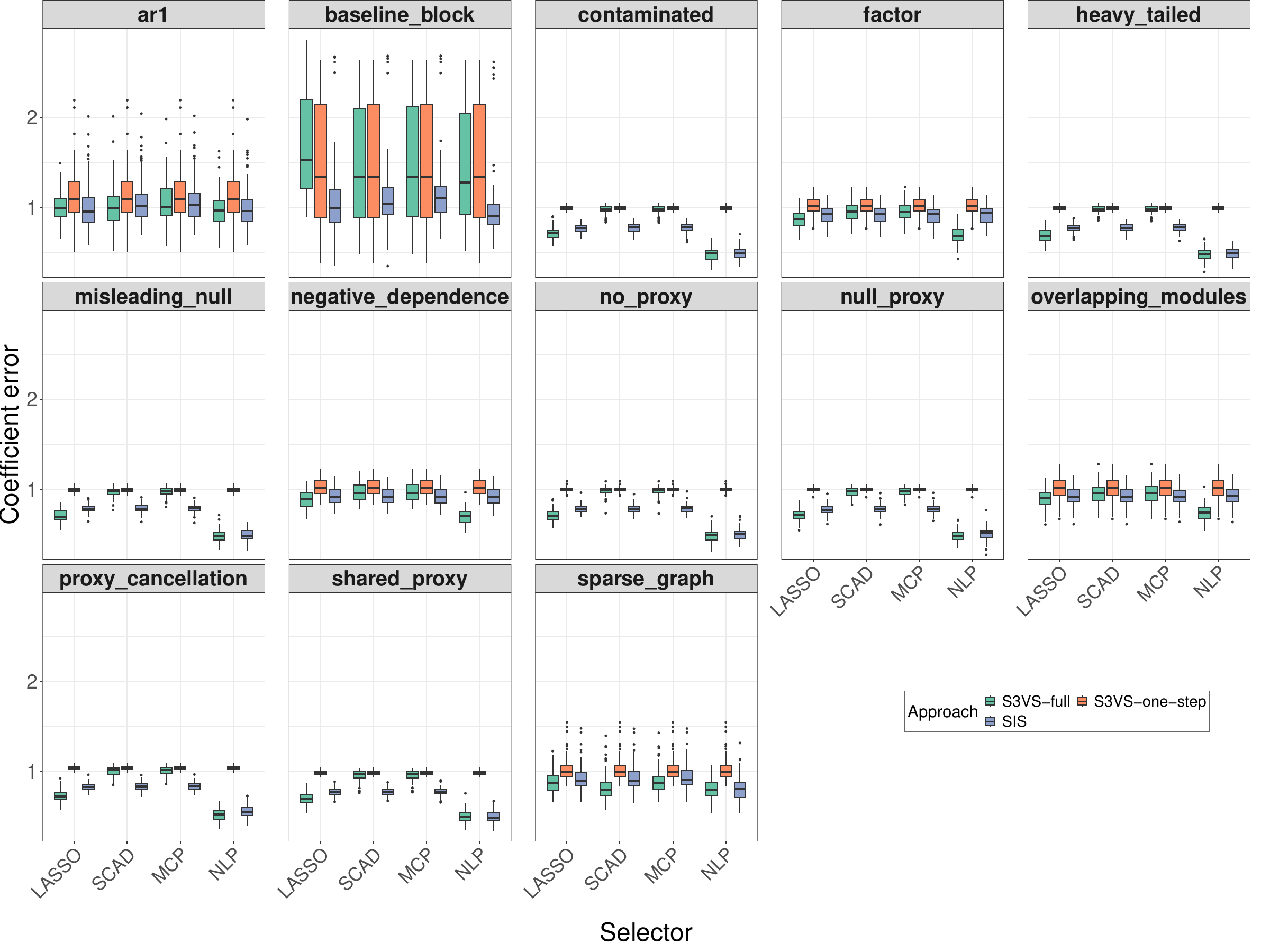}
\caption{Replicate-level coefficient error after the common post-selection LASSO refit across the named design scenarios.}
\label{fig:LM-supp-coefficient-error}
\end{figure}

\subsection{Sample-size and dimension sensitivity}
Finally, Figure~\ref{fig:LM-supp-scalability} pools over the named scenarios at the moderate signal level and displays TPR across the full $(n,p)$ grid. Recovery generally increases with $n$ for full S3VS and SIS. At a fixed $n$, increasing $p$ from 10,000 to 20,000 has a smaller and usually mildly adverse effect. The first-iteration procedures remain near the bottom of the plot across the grid. The complete condition-level summaries, including Monte Carlo standard errors, are supplied in the machine-readable result files accompanying the manuscript.

\begin{figure}[htbp]
\centering
\includegraphics[width=\textwidth]{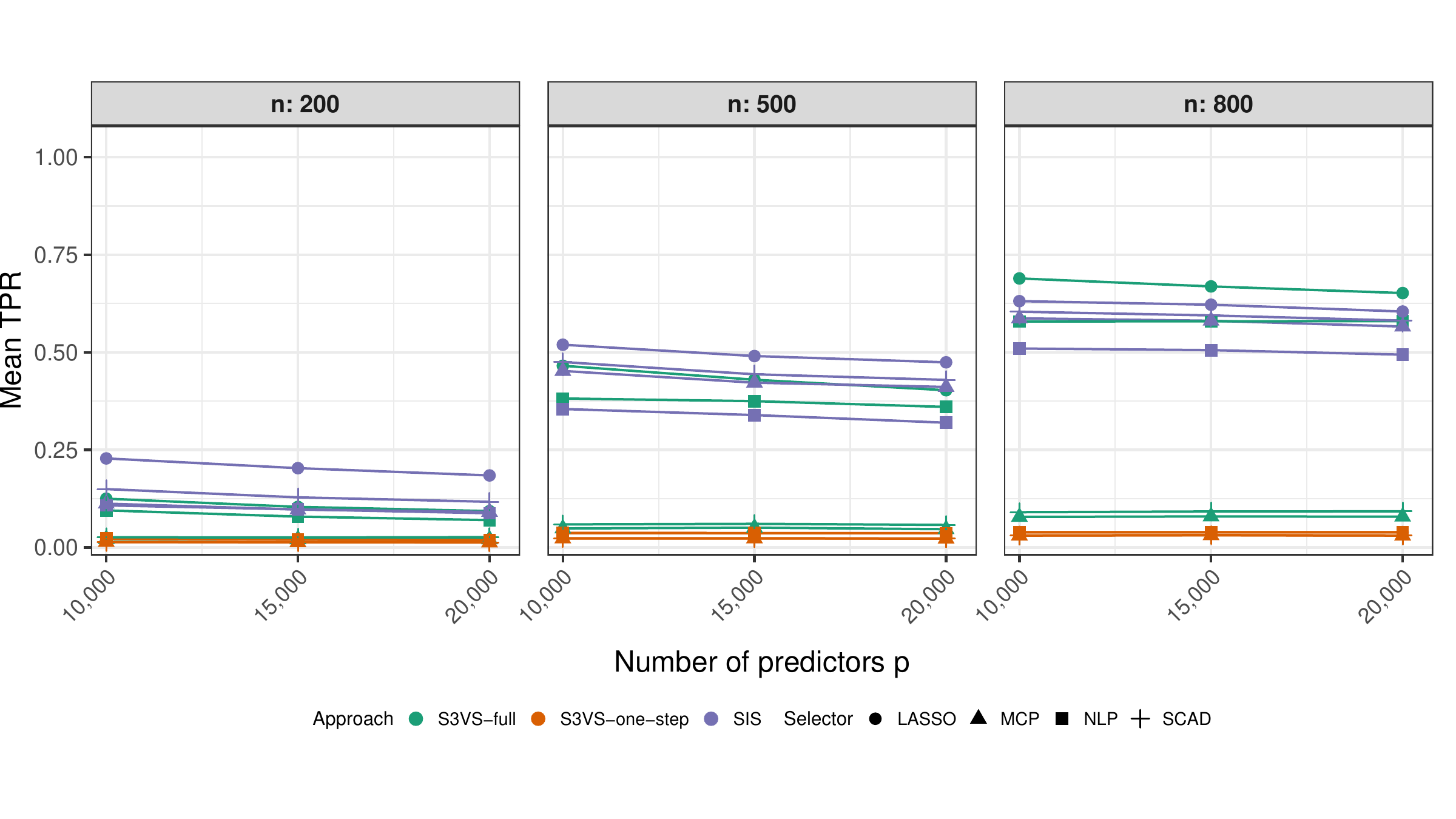}
\caption{TPR across the full sample-size and dimension grid. Means are pooled over named scenarios at the moderate signal level and bars show 95\% Monte Carlo intervals.}
\label{fig:LM-supp-scalability}
\end{figure}

\section{Additional Binary-Logistic Simulation Results} \label{app:glm-supplement}

The main text reports the largest $(n,p)$ setting to keep the primary comparison compact.  This section gives the complete results for all nine sample-size--dimension combinations.  Table~\ref{tab:GLM-full-grid} reports selection success, probability of a nonempty selection, TPR, unconditional FDP, selected-model size, and independent-test log-loss for every scenario and method.

\begin{table}[htbp]
\centering
\small
\caption{Full binary-logistic simulation results over the n--p grid. Entries are means with Monte Carlo standard errors in parentheses.}
\label{tab:GLM-full-grid}
\begin{adjustbox}{max width=\textwidth,max height=0.92\textheight,keepaspectratio}
\begin{tabular}{llrrrrrrrr}
\toprule
Scenario & Method & n & p & \shortstack[c]{Selection success,\\MCSE} & \shortstack[c]{P(nonempty),\\MCSE} & \shortstack[c]{TPR,\\MCSE} & \shortstack[c]{FDP,\\MCSE} & \shortstack[c]{Mean size,\\MCSE} & \shortstack[c]{Test log-loss,\\MCSE} \\
\midrule
Correlated proxy & S3VS & 200 & 10,000 & 1.000 (0.000) & 1.000 (0.000) & 0.363 (0.032) & 0.820 (0.012) & 44.8 (3.0) & 0.735 (0.010) \\
Correlated proxy & SIS--logistic--LASSO & 200 & 10,000 & 1.000 (0.000) & 1.000 (0.000) & 0.075 (0.006) & 0.979 (0.002) & 97.2 (1.5) & 1.217 (0.026) \\
Correlated proxy & S3VS & 200 & 15,000 & 1.000 (0.000) & 1.000 (0.000) & 0.392 (0.033) & 0.815 (0.011) & 46.7 (2.9) & 0.749 (0.011) \\
Correlated proxy & SIS--logistic--LASSO & 200 & 15,000 & 1.000 (0.000) & 1.000 (0.000) & 0.086 (0.006) & 0.976 (0.002) & 96.7 (1.6) & 1.245 (0.024) \\
Correlated proxy & S3VS & 200 & 20,000 & 1.000 (0.000) & 1.000 (0.000) & 0.383 (0.035) & 0.837 (0.011) & 48.6 (3.0) & 0.756 (0.010) \\
Correlated proxy & SIS--logistic--LASSO & 200 & 20,000 & 1.000 (0.000) & 1.000 (0.000) & 0.065 (0.005) & 0.982 (0.002) & 99.5 (1.5) & 1.248 (0.027) \\
Correlated proxy & S3VS & 500 & 10,000 & 1.000 (0.000) & 1.000 (0.000) & 0.700 (0.036) & 0.772 (0.008) & 74.6 (3.4) & 0.681 (0.005) \\
Correlated proxy & SIS--logistic--LASSO & 500 & 10,000 & 1.000 (0.000) & 1.000 (0.000) & 0.189 (0.009) & 0.957 (0.002) & 119.0 (2.3) & 0.910 (0.011) \\
Correlated proxy & S3VS & 500 & 15,000 & 1.000 (0.000) & 1.000 (0.000) & 0.676 (0.037) & 0.786 (0.009) & 74.7 (3.4) & 0.686 (0.005) \\
Correlated proxy & SIS--logistic--LASSO & 500 & 15,000 & 1.000 (0.000) & 1.000 (0.000) & 0.169 (0.009) & 0.962 (0.002) & 120.4 (2.4) & 0.936 (0.011) \\
Correlated proxy & S3VS & 500 & 20,000 & 1.000 (0.000) & 1.000 (0.000) & 0.646 (0.037) & 0.790 (0.011) & 72.7 (3.3) & 0.694 (0.006) \\
Correlated proxy & SIS--logistic--LASSO & 500 & 20,000 & 1.000 (0.000) & 1.000 (0.000) & 0.166 (0.008) & 0.963 (0.002) & 123.5 (2.6) & 0.943 (0.011) \\
Correlated proxy & S3VS & 800 & 10,000 & 1.000 (0.000) & 1.000 (0.000) & 0.787 (0.028) & 0.741 (0.007) & 79.5 (3.1) & 0.666 (0.004) \\
Correlated proxy & SIS--logistic--LASSO & 800 & 10,000 & 1.000 (0.000) & 1.000 (0.000) & 0.335 (0.016) & 0.931 (0.003) & 127.8 (2.1) & 0.844 (0.006) \\
Correlated proxy & S3VS & 800 & 15,000 & 1.000 (0.000) & 1.000 (0.000) & 0.770 (0.030) & 0.747 (0.006) & 78.8 (3.2) & 0.670 (0.005) \\
Correlated proxy & SIS--logistic--LASSO & 800 & 15,000 & 1.000 (0.000) & 1.000 (0.000) & 0.307 (0.015) & 0.936 (0.003) & 129.6 (2.4) & 0.863 (0.007) \\
Correlated proxy & S3VS & 800 & 20,000 & 1.000 (0.000) & 1.000 (0.000) & 0.800 (0.029) & 0.764 (0.007) & 84.5 (2.8) & 0.664 (0.004) \\
Correlated proxy & SIS--logistic--LASSO & 800 & 20,000 & 1.000 (0.000) & 1.000 (0.000) & 0.316 (0.014) & 0.933 (0.003) & 125.5 (2.3) & 0.862 (0.006) \\
No proxy & S3VS & 200 & 10,000 & 1.000 (0.000) & 1.000 (0.000) & 0.025 (0.003) & 0.966 (0.005) & 20.6 (0.8) & 0.846 (0.008) \\
No proxy & SIS--logistic--LASSO & 200 & 10,000 & 1.000 (0.000) & 1.000 (0.000) & 0.085 (0.005) & 0.982 (0.001) & 116.0 (0.6) & 1.485 (0.022) \\
No proxy & S3VS & 200 & 15,000 & 1.000 (0.000) & 1.000 (0.000) & 0.019 (0.003) & 0.977 (0.003) & 22.6 (1.0) & 0.847 (0.008) \\
No proxy & SIS--logistic--LASSO & 200 & 15,000 & 1.000 (0.000) & 1.000 (0.000) & 0.063 (0.003) & 0.987 (0.001) & 117.8 (0.6) & 1.434 (0.021) \\
No proxy & S3VS & 200 & 20,000 & 1.000 (0.000) & 1.000 (0.000) & 0.022 (0.003) & 0.980 (0.003) & 30.1 (1.2) & 0.889 (0.010) \\
No proxy & SIS--logistic--LASSO & 200 & 20,000 & 1.000 (0.000) & 1.000 (0.000) & 0.051 (0.004) & 0.989 (0.001) & 119.1 (0.7) & 1.458 (0.020) \\
No proxy & S3VS & 500 & 10,000 & 1.000 (0.000) & 1.000 (0.000) & 0.091 (0.006) & 0.891 (0.007) & 21.0 (0.9) & 0.754 (0.004) \\
No proxy & SIS--logistic--LASSO & 500 & 10,000 & 1.000 (0.000) & 1.000 (0.000) & 0.283 (0.006) & 0.958 (0.001) & 168.3 (0.5) & 1.003 (0.009) \\
No proxy & S3VS & 500 & 15,000 & 1.000 (0.000) & 1.000 (0.000) & 0.088 (0.006) & 0.903 (0.007) & 23.9 (1.0) & 0.768 (0.004) \\
No proxy & SIS--logistic--LASSO & 500 & 15,000 & 1.000 (0.000) & 1.000 (0.000) & 0.244 (0.009) & 0.964 (0.001) & 168.7 (0.4) & 1.043 (0.011) \\
No proxy & S3VS & 500 & 20,000 & 1.000 (0.000) & 1.000 (0.000) & 0.072 (0.005) & 0.925 (0.007) & 26.5 (1.1) & 0.785 (0.004) \\
No proxy & SIS--logistic--LASSO & 500 & 20,000 & 1.000 (0.000) & 1.000 (0.000) & 0.211 (0.008) & 0.969 (0.001) & 168.8 (0.4) & 1.060 (0.010) \\
No proxy & S3VS & 800 & 10,000 & 1.000 (0.000) & 1.000 (0.000) & 0.200 (0.008) & 0.767 (0.010) & 23.1 (0.9) & 0.718 (0.002) \\
No proxy & SIS--logistic--LASSO & 800 & 10,000 & 1.000 (0.000) & 1.000 (0.000) & 0.485 (0.008) & 0.935 (0.001) & 185.9 (0.3) & 0.889 (0.005) \\
No proxy & S3VS & 800 & 15,000 & 1.000 (0.000) & 1.000 (0.000) & 0.182 (0.007) & 0.815 (0.009) & 27.6 (1.1) & 0.736 (0.003) \\
No proxy & SIS--logistic--LASSO & 800 & 15,000 & 1.000 (0.000) & 1.000 (0.000) & 0.431 (0.009) & 0.942 (0.001) & 185.5 (0.3) & 0.916 (0.006) \\
No proxy & S3VS & 800 & 20,000 & 1.000 (0.000) & 1.000 (0.000) & 0.155 (0.008) & 0.858 (0.007) & 28.5 (1.2) & 0.747 (0.003) \\
No proxy & SIS--logistic--LASSO & 800 & 20,000 & 1.000 (0.000) & 1.000 (0.000) & 0.377 (0.008) & 0.949 (0.001) & 186.0 (0.3) & 0.937 (0.006) \\
\bottomrule
\end{tabular}
\end{adjustbox}
\end{table}

Across the correlated-proxy grid, increasing $n$ generally improved TPR for S3VS: its values ranged from approximately 0.36--0.39 at $n=200$ to 0.65--0.80 at $n=800$, whereas SIS--logistic--LASSO ranged from about 0.07--0.09 to 0.31--0.34. S3VS also selected fewer variables and had lower unconditional FDP throughout this scenario. The corresponding independent-test log-loss values were consistently lower for S3VS in the correlated-proxy setting.

In the no-proxy scenario, S3VS had lower TPR across the grid, approximately 0.02--0.20, while SIS--logistic--LASSO ranged from about 0.05--0.49. S3VS nevertheless selected much smaller models and generally had lower test log-loss. FDP remained high for both methods, so the results should be read as a comparison of the trade-off between recovery and model size rather than as evidence that either procedure provides uniformly good variable selection in the absence of useful proxy structure. Selection success and nonempty-selection probability were 1.000 in the displayed grid, so the primary comparisons were not driven by frequent algorithmic failure or empty models.

Figure~\ref{fig:GLM-full-grid} summarizes the same grid graphically. Its three rows show mean TPR, mean unconditional FDP, and independent-test log-loss, while its two columns separate the correlated-proxy and no-proxy scenarios. The figure shows the main qualification clearly. The TPR advantage of S3VS is concentrated in the correlated-proxy setting. Without proxies, its main advantages are smaller selected models and lower test log-loss.

\begin{figure}[htbp]   
\centering   
\includegraphics[width=\textwidth]{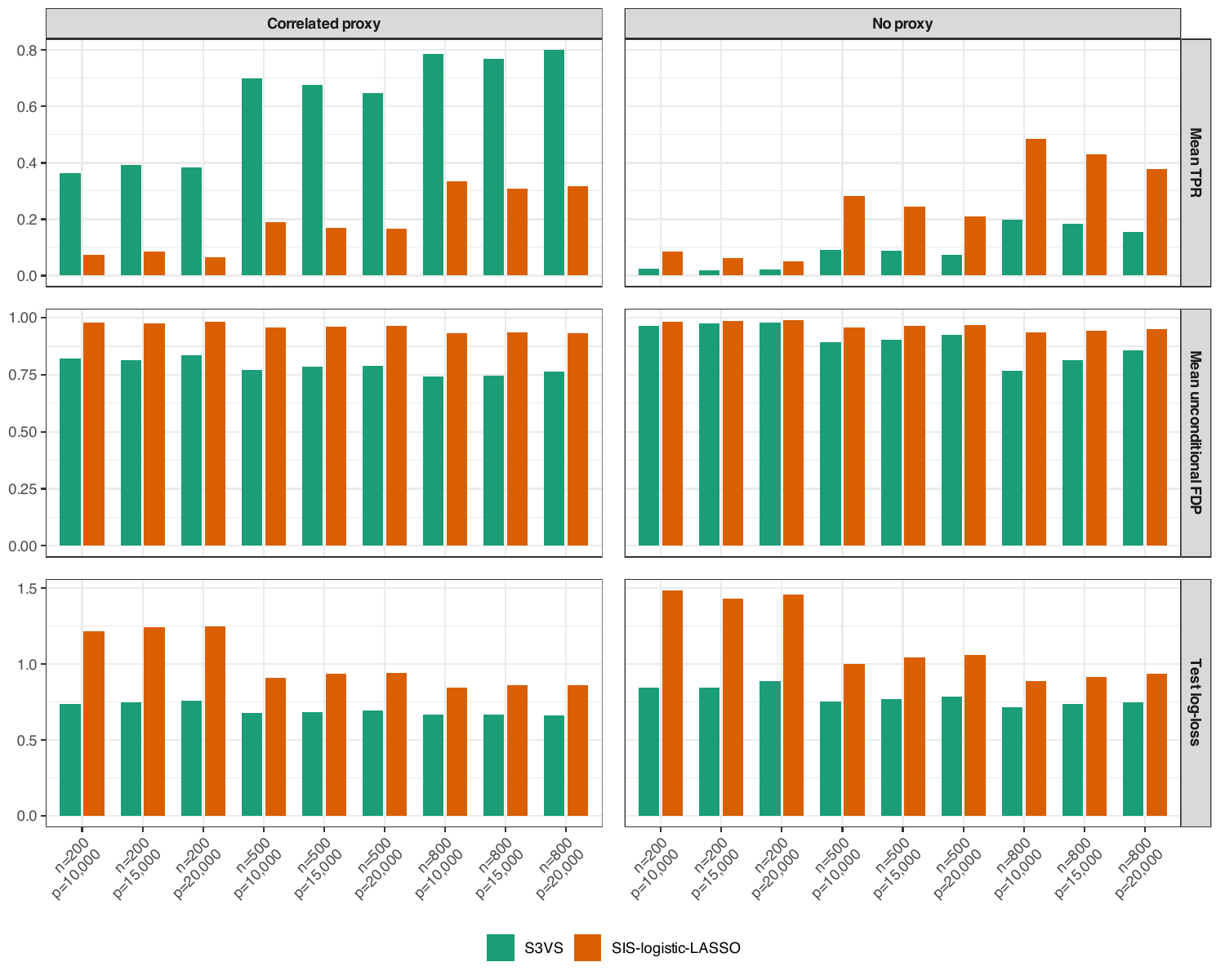}   
\caption{Full binary-logistic simulation grid. Columns correspond to the correlated-proxy and no-proxy scenarios. Rows show mean TPR, mean unconditional FDP, and independent-test log-loss over the 100 replicates for each $(n,p)$ combination.}   
\label{fig:GLM-full-grid} 
\end{figure}

\section{Additional Cox-Model Simulation Results}
\label{sec:supplement-cox}

This section provides the secondary survival-simulation results associated with Section~\ref{sec:simulation-cox}.  The full iterative S3VS analysis is reported separately because its iteration budget is not directly matched to the single-pass SIS--Cox--LASSO comparator.  The complete $(n,p)$ grids are also given for both the full iterative analysis and the primary one-iteration comparison.

For clarity, in all Cox simulations the SIS--Cox--LASSO comparator ranked candidate predictors by the absolute marginal Cox Wald statistic from separate univariate Cox proportional-hazards fits using both the observed survival times and event indicators. For candidate predictor $j$, the statistic was $|z_j|$, where $z_j=\widehat{\beta}_{j,\mathrm{marg}}/\widehat{\operatorname{se}}(\widehat{\beta}_{j,\mathrm{marg}})$. The comparator retained up to $\min\{200,n_{\mathrm{train}}-1,p\}$ highest-ranked predictors and then fitted Cox LASSO to the screened set using ten-fold cross-validation and the one-standard-error tuning value.

\subsection{Full iterative S3VS analysis}
The full S3VS analysis used the same data-generating design and marginal-Cox-Wald SIS--Cox--LASSO comparator as the primary analysis, but allowed S3VS to run with $m_{\text{max}}=100$ and \texttt{nskip=3}. The SIS--Cox--LASSO comparator remained a single-pass procedure. At $(n,p)=(800,20{,}000)$, the resulting TPR, unconditional FDP, selected-model size, and test concordance summaries are reported in Supplementary Table~\ref{tab:surv-primary-performance}; the corresponding TPR and FDP plot is shown in Supplementary Figure~\ref{fig:surv-full-iterative-primary}. In the correlated-proxy setting, full-iteration S3VS had TPR $0.066$, unconditional FDP $0.821$, mean selected-model size $10.7$, and test C-index $0.582$, compared with $0.242$, $0.932$, $98.0$, and $0.572$ for SIS--Cox--LASSO. In the no-proxy setting, the corresponding S3VS values were $0.110$, $0.817$, $15.3$, and $0.523$, compared with $0.396$, $0.937$, $157.6$, and $0.531$ for SIS--Cox--LASSO. These results document the behavior of the implemented iterative procedure but are not used as the primary method-to-method comparison.

\begin{table}[htbp]
\centering
\small
\caption{Secondary full-iteration Cox-model simulation performance at $(n,p)=(800,20{,}000)$. Entries are means with Monte Carlo standard errors in parentheses; FDP is the unconditional FDP with empty selections counted as zero.\\[-5pt]}
\label{tab:surv-primary-performance}
\begin{adjustbox}{max width=\textwidth,keepaspectratio}
\begin{tabular}{llrrrrrrr}
\toprule
Scenario & Method & n & p & \shortstack[c]{P(nonempty),\\MCSE} & \shortstack[c]{TPR,\\MCSE} & \shortstack[c]{FDP,\\MCSE} & \shortstack[c]{Mean size,\\MCSE} & \shortstack[c]{Test C-index,\\MCSE} \\
\midrule
Correlated proxy & S3VS & 800 & 20,000 & 1.000 (0.000) & 0.066 (0.003) & 0.821 (0.013) & 10.7 (0.3) & 0.582 (0.005) \\
Correlated proxy & SIS--Cox--LASSO & 800 & 20,000 & 1.000 (0.000) & 0.242 (0.008) & 0.932 (0.003) & 98.0 (2.5) & 0.572 (0.004) \\
No proxy & S3VS & 800 & 20,000 & 1.000 (0.000) & 0.110 (0.006) & 0.817 (0.011) & 15.3 (0.3) & 0.523 (0.002) \\
No proxy & SIS--Cox--LASSO & 800 & 20,000 & 1.000 (0.000) & 0.396 (0.009) & 0.937 (0.001) & 157.6 (0.8) & 0.531 (0.002) \\
\bottomrule
\end{tabular}
\end{adjustbox}
\end{table}

\begin{figure}[htbp]
\centering
\includegraphics[width=\textwidth]{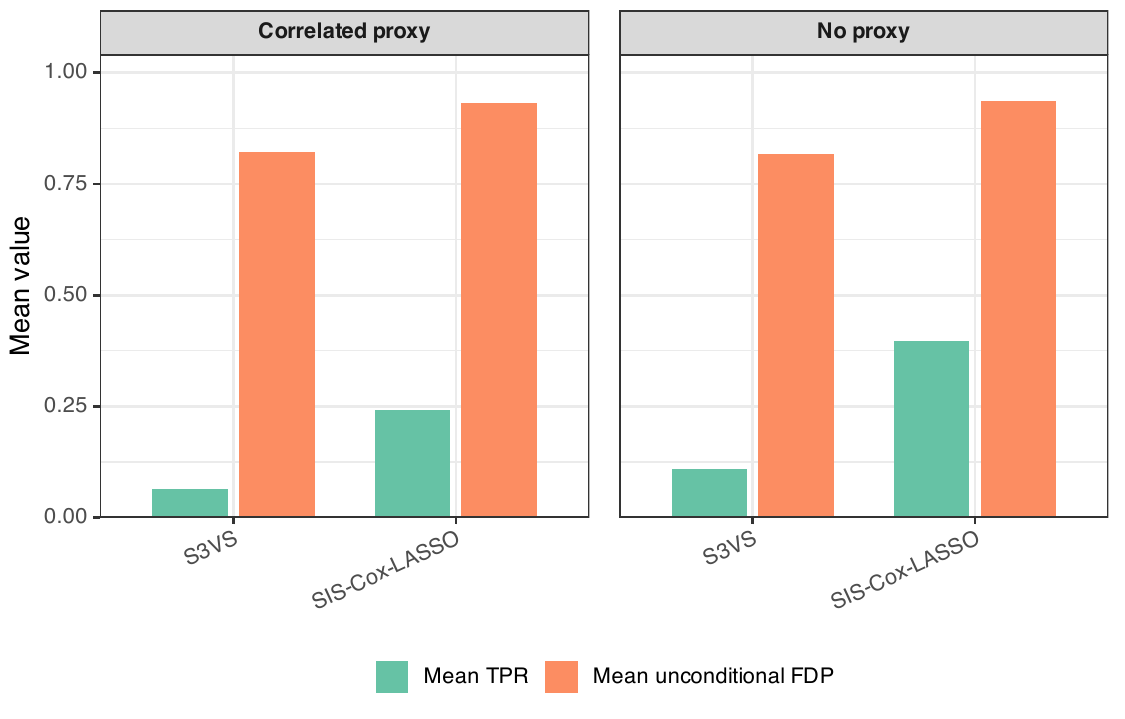}
\caption{Secondary full-iteration Cox-model comparison at $(n,p)=(800,20{,}000)$. Bars show mean TPR and mean unconditional FDP for the full iterative S3VS procedure and the single-pass SIS--Cox--LASSO procedure in the correlated-proxy and no-proxy scenarios.}
\label{fig:surv-full-iterative-primary}
\end{figure}

\subsection{Full grid for the full iterative analysis}
The full-grid results for the iterative analysis are given in Supplementary Table~\ref{tab:surv-full-grid} and Supplementary Figure~\ref{fig:surv-full-grid}. The table additionally reports training event and censoring proportions, and selection and postfit success rates. Across the grid, full-iteration S3VS consistently selected much smaller models than SIS--Cox--LASSO. Its TPR was lower in every correlated-proxy cell and in all but one no-proxy cell, where the two procedures were nearly identical at $(n,p)=(200,10{,}000)$. The FDP comparison was not uniform over the grid: S3VS had lower FDP than SIS--Cox--LASSO at the largest setting but often higher FDP at smaller sample sizes. The figure displays the corresponding patterns for TPR, unconditional FDP, selected-model size, and independent-test concordance.

\begin{table}[htbp]
\centering
\small
\caption{Complete full-iteration Cox-model simulation results over the $(n,p)$ grid. Event and censoring rates are calculated in the training samples; other entries are means with Monte Carlo standard errors in parentheses. FDP is the unconditional FDP with empty selections counted as zero.}
\label{tab:surv-full-grid}
\begin{adjustbox}{max width=\textwidth,max height=0.92\textheight,keepaspectratio}
\begin{tabular}{llrrrrrrrrrr}
\toprule
Scenario & Method & n & p & \shortstack[c]{Event rate,\\MCSE} & \shortstack[c]{Censoring rate,\\MCSE} & \shortstack[c]{Selection success,\\MCSE} & \shortstack[c]{P(nonempty),\\MCSE} & \shortstack[c]{TPR,\\MCSE} & \shortstack[c]{FDP,\\MCSE} & \shortstack[c]{Mean size,\\MCSE} & \shortstack[c]{Test C-index,\\MCSE} \\
\midrule
Correlated proxy & S3VS & 200 & 10,000 & 0.698 (0.003) & 0.302 (0.003) & 1.000 (0.000) & 1.000 (0.000) & 0.021 (0.002) & 0.901 (0.011) & 5.7 (0.1) & 0.546 (0.004) \\
Correlated proxy & SIS--Cox--LASSO & 200 & 10,000 & 0.698 (0.003) & 0.302 (0.003) & 1.000 (0.000) & 0.410 (0.049) & 0.035 (0.005) & 0.334 (0.045) & 16.5 (2.6) & 0.523 (0.004) \\
Correlated proxy & S3VS & 200 & 15,000 & 0.702 (0.003) & 0.298 (0.003) & 1.000 (0.000) & 1.000 (0.000) & 0.017 (0.002) & 0.919 (0.013) & 5.7 (0.2) & 0.554 (0.005) \\
Correlated proxy & SIS--Cox--LASSO & 200 & 15,000 & 0.702 (0.003) & 0.298 (0.003) & 1.000 (0.000) & 0.490 (0.050) & 0.044 (0.006) & 0.433 (0.046) & 19.6 (2.7) & 0.529 (0.004) \\
Correlated proxy & S3VS & 200 & 20,000 & 0.702 (0.003) & 0.298 (0.003) & 1.000 (0.000) & 1.000 (0.000) & 0.012 (0.002) & 0.943 (0.009) & 6.1 (0.2) & 0.544 (0.005) \\
Correlated proxy & SIS--Cox--LASSO & 200 & 20,000 & 0.702 (0.003) & 0.298 (0.003) & 1.000 (0.000) & 0.590 (0.049) & 0.046 (0.006) & 0.532 (0.047) & 29.1 (3.0) & 0.531 (0.005) \\
Correlated proxy & S3VS & 500 & 10,000 & 0.702 (0.002) & 0.298 (0.002) & 1.000 (0.000) & 1.000 (0.000) & 0.044 (0.003) & 0.852 (0.011) & 8.4 (0.3) & 0.567 (0.005) \\
Correlated proxy & SIS--Cox--LASSO & 500 & 10,000 & 0.702 (0.002) & 0.298 (0.002) & 1.000 (0.000) & 0.920 (0.027) & 0.159 (0.009) & 0.804 (0.030) & 68.8 (3.8) & 0.553 (0.005) \\
Correlated proxy & S3VS & 500 & 15,000 & 0.702 (0.002) & 0.298 (0.002) & 1.000 (0.000) & 1.000 (0.000) & 0.042 (0.002) & 0.872 (0.010) & 9.3 (0.2) & 0.567 (0.005) \\
Correlated proxy & SIS--Cox--LASSO & 500 & 15,000 & 0.702 (0.002) & 0.298 (0.002) & 1.000 (0.000) & 0.990 (0.010) & 0.159 (0.008) & 0.896 (0.019) & 83.7 (3.3) & 0.553 (0.004) \\
Correlated proxy & S3VS & 500 & 20,000 & 0.699 (0.002) & 0.301 (0.002) & 1.000 (0.000) & 1.000 (0.000) & 0.037 (0.002) & 0.887 (0.009) & 9.6 (0.3) & 0.564 (0.005) \\
Correlated proxy & SIS--Cox--LASSO & 500 & 20,000 & 0.699 (0.002) & 0.301 (0.002) & 1.000 (0.000) & 0.990 (0.010) & 0.150 (0.008) & 0.923 (0.014) & 89.0 (3.3) & 0.554 (0.004) \\
Correlated proxy & S3VS & 800 & 10,000 & 0.700 (0.002) & 0.300 (0.002) & 1.000 (0.000) & 1.000 (0.000) & 0.064 (0.003) & 0.774 (0.017) & 8.8 (0.3) & 0.569 (0.006) \\
Correlated proxy & SIS--Cox--LASSO & 800 & 10,000 & 0.700 (0.002) & 0.300 (0.002) & 1.000 (0.000) & 1.000 (0.000) & 0.234 (0.010) & 0.877 (0.021) & 94.3 (3.6) & 0.567 (0.005) \\
Correlated proxy & S3VS & 800 & 15,000 & 0.700 (0.002) & 0.300 (0.002) & 1.000 (0.000) & 1.000 (0.000) & 0.062 (0.003) & 0.795 (0.018) & 9.8 (0.4) & 0.570 (0.005) \\
Correlated proxy & SIS--Cox--LASSO & 800 & 15,000 & 0.700 (0.002) & 0.300 (0.002) & 1.000 (0.000) & 1.000 (0.000) & 0.238 (0.009) & 0.931 (0.004) & 101.3 (2.9) & 0.567 (0.004) \\
Correlated proxy & S3VS & 800 & 20,000 & 0.698 (0.001) & 0.302 (0.001) & 1.000 (0.000) & 1.000 (0.000) & 0.066 (0.003) & 0.821 (0.013) & 10.7 (0.3) & 0.582 (0.005) \\
Correlated proxy & SIS--Cox--LASSO & 800 & 20,000 & 0.698 (0.001) & 0.302 (0.001) & 1.000 (0.000) & 1.000 (0.000) & 0.242 (0.008) & 0.932 (0.003) & 98.0 (2.5) & 0.572 (0.004) \\
No proxy & S3VS & 200 & 10,000 & 0.702 (0.003) & 0.298 (0.003) & 1.000 (0.000) & 1.000 (0.000) & 0.011 (0.002) & 0.954 (0.009) & 5.9 (0.2) & 0.502 (0.003) \\
No proxy & SIS--Cox--LASSO & 200 & 10,000 & 0.702 (0.003) & 0.298 (0.003) & 1.000 (0.000) & 0.120 (0.033) & 0.010 (0.003) & 0.116 (0.032) & 8.4 (2.4) & 0.501 (0.001) \\
No proxy & S3VS & 200 & 15,000 & 0.701 (0.003) & 0.299 (0.003) & 1.000 (0.000) & 1.000 (0.000) & 0.007 (0.002) & 0.971 (0.007) & 6.4 (0.2) & 0.505 (0.003) \\
No proxy & SIS--Cox--LASSO & 200 & 15,000 & 0.701 (0.003) & 0.299 (0.003) & 1.000 (0.000) & 0.150 (0.036) & 0.008 (0.003) & 0.147 (0.035) & 10.6 (2.6) & 0.501 (0.001) \\
No proxy & S3VS & 200 & 20,000 & 0.699 (0.003) & 0.301 (0.003) & 1.000 (0.000) & 1.000 (0.000) & 0.003 (0.001) & 0.989 (0.005) & 6.4 (0.1) & 0.503 (0.003) \\
No proxy & SIS--Cox--LASSO & 200 & 20,000 & 0.699 (0.003) & 0.301 (0.003) & 1.000 (0.000) & 0.330 (0.047) & 0.007 (0.002) & 0.328 (0.047) & 23.6 (3.6) & 0.497 (0.002) \\
No proxy & S3VS & 500 & 10,000 & 0.700 (0.002) & 0.300 (0.002) & 1.000 (0.000) & 1.000 (0.000) & 0.063 (0.004) & 0.844 (0.011) & 10.2 (0.2) & 0.513 (0.002) \\
No proxy & SIS--Cox--LASSO & 500 & 10,000 & 0.700 (0.002) & 0.300 (0.002) & 1.000 (0.000) & 0.760 (0.043) & 0.198 (0.013) & 0.714 (0.040) & 88.8 (5.5) & 0.513 (0.002) \\
No proxy & S3VS & 500 & 15,000 & 0.701 (0.002) & 0.299 (0.002) & 1.000 (0.000) & 1.000 (0.000) & 0.048 (0.004) & 0.888 (0.010) & 11.0 (0.2) & 0.510 (0.002) \\
No proxy & SIS--Cox--LASSO & 500 & 15,000 & 0.701 (0.002) & 0.299 (0.002) & 1.000 (0.000) & 0.840 (0.037) & 0.168 (0.010) & 0.801 (0.035) & 101.8 (5.1) & 0.511 (0.002) \\
No proxy & S3VS & 500 & 20,000 & 0.701 (0.002) & 0.299 (0.002) & 1.000 (0.000) & 1.000 (0.000) & 0.041 (0.004) & 0.906 (0.010) & 11.3 (0.2) & 0.512 (0.002) \\
No proxy & SIS--Cox--LASSO & 500 & 20,000 & 0.701 (0.002) & 0.299 (0.002) & 1.000 (0.000) & 0.940 (0.024) & 0.181 (0.009) & 0.905 (0.023) & 120.8 (3.4) & 0.514 (0.002) \\
No proxy & S3VS & 800 & 10,000 & 0.700 (0.002) & 0.300 (0.002) & 1.000 (0.000) & 1.000 (0.000) & 0.153 (0.006) & 0.703 (0.013) & 13.6 (0.3) & 0.525 (0.002) \\
No proxy & SIS--Cox--LASSO & 800 & 10,000 & 0.700 (0.002) & 0.300 (0.002) & 1.000 (0.000) & 0.980 (0.014) & 0.468 (0.011) & 0.900 (0.013) & 145.8 (2.7) & 0.539 (0.002) \\
No proxy & S3VS & 800 & 15,000 & 0.699 (0.002) & 0.301 (0.002) & 1.000 (0.000) & 1.000 (0.000) & 0.122 (0.007) & 0.781 (0.014) & 14.5 (0.3) & 0.524 (0.002) \\
No proxy & SIS--Cox--LASSO & 800 & 15,000 & 0.699 (0.002) & 0.301 (0.002) & 1.000 (0.000) & 1.000 (0.000) & 0.414 (0.010) & 0.932 (0.002) & 154.2 (1.7) & 0.533 (0.002) \\
No proxy & S3VS & 800 & 20,000 & 0.699 (0.002) & 0.301 (0.002) & 1.000 (0.000) & 1.000 (0.000) & 0.110 (0.006) & 0.817 (0.011) & 15.3 (0.3) & 0.523 (0.002) \\
No proxy & SIS--Cox--LASSO & 800 & 20,000 & 0.699 (0.002) & 0.301 (0.002) & 1.000 (0.000) & 1.000 (0.000) & 0.396 (0.009) & 0.937 (0.001) & 157.6 (0.8) & 0.531 (0.002) \\
\bottomrule
\end{tabular}
\end{adjustbox}
\end{table}

\begin{figure}[htbp]
\centering
\includegraphics[width=\textwidth]{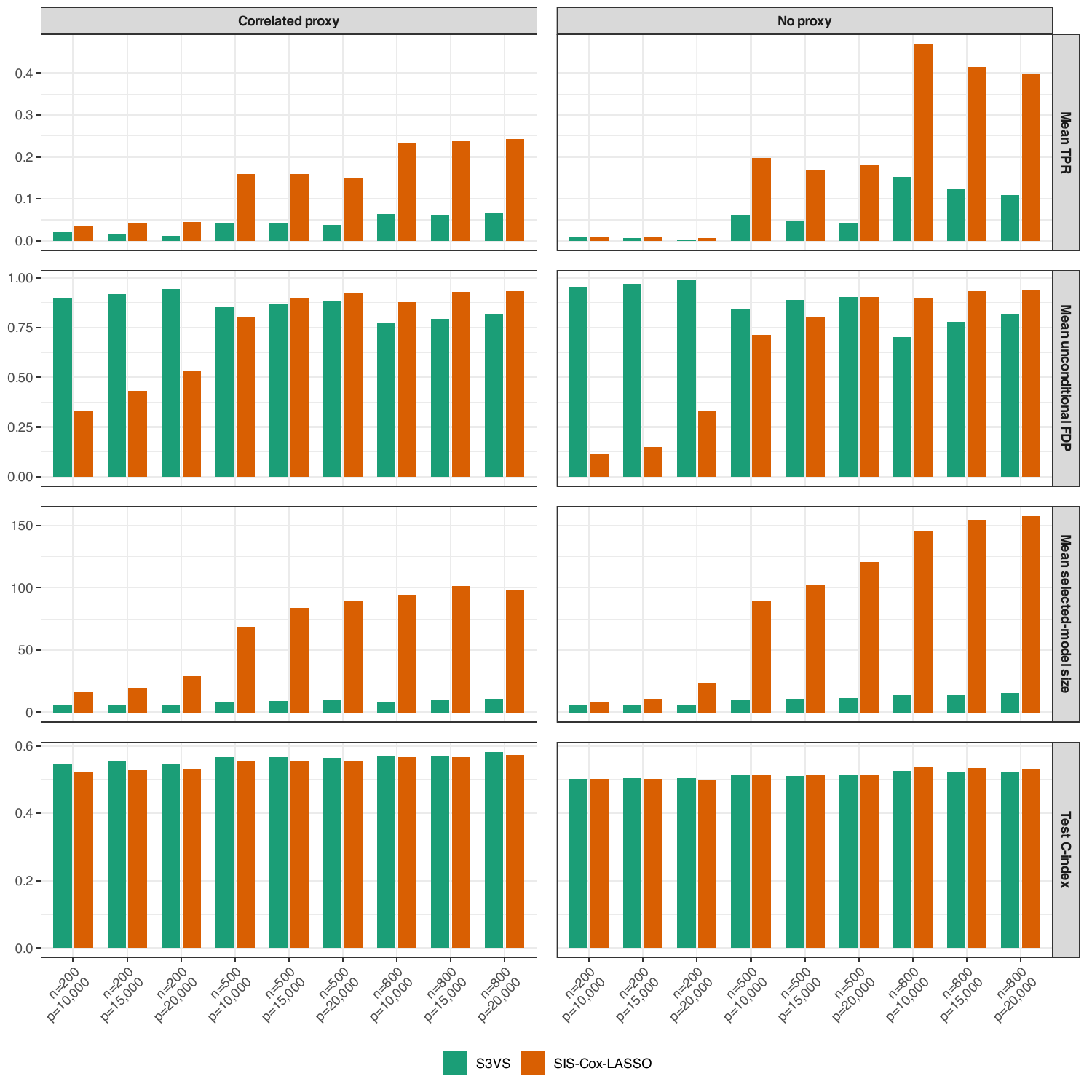}
\caption{Full-grid results for the full iterative Cox-model analysis over $n\in\{200,500,800\}$ and $p\in\{10{,}000,15{,}000,20{,}000\}$. Rows show mean TPR, mean unconditional FDP, mean selected-model size, and test concordance. Columns show the correlated-proxy and no-proxy scenarios.}
\label{fig:surv-full-grid}
\end{figure}

\subsection{Full grid for the one-iteration comparison}
The complete results for the matched one-iteration comparison are reported in Supplementary Table~\ref{tab:surv-one-iteration-full-grid} and Supplementary Figure~\ref{fig:surv-one-iteration-full-grid}. The displayed one-iteration S3VS fits have mean selected-model size $1.0$ and probability of a nonempty selection $1.000$ in every grid cell. S3VS has lower TPR than SIS--Cox--LASSO in every displayed cell. In the correlated-proxy setting, S3VS has higher test concordance in every cell; its FDP is higher than that of SIS--Cox--LASSO at $n=200$ but lower at $n=500$ and $n=800$. In the no-proxy setting, S3VS has higher test concordance at $n=200$ but lower test concordance at $n=500$ and $n=800$, with the same sample-size pattern for the FDP comparison. These results show that the primary one-iteration method is highly conservative in model size and that its recovery, purity, and predictive behavior depend on predictor structure and sample size.

\begin{table}[htbp]
\centering
\small
\caption{Complete one-iteration Cox-model simulation results over the $(n,p)$ grid. Entries are means with Monte Carlo standard errors in parentheses. FDP is the unconditional FDP with empty selections counted as zero.\\[-5pt]}
\label{tab:surv-one-iteration-full-grid}
\begin{adjustbox}{max width=\textwidth,max height=0.92\textheight,keepaspectratio}
\begin{tabular}{llrrrrrrrrrr}
\toprule
Scenario & Method & n & p & \shortstack[c]{Event rate,\\MCSE} & \shortstack[c]{Censoring rate,\\MCSE} & \shortstack[c]{Selection success,\\MCSE} & \shortstack[c]{P(nonempty),\\MCSE} & \shortstack[c]{TPR,\\MCSE} & \shortstack[c]{FDP,\\MCSE} & \shortstack[c]{Mean size,\\MCSE} & \shortstack[c]{Test C-index,\\MCSE} \\
\midrule
Correlated proxy & S3VS (one iteration) & 200 & 10,000 & 0.698 (0.003) & 0.302 (0.003) & 1.000 (0.000) & 1.000 (0.000) & 0.019 (0.002) & 0.520 (0.050) & 1.0 (0.0) & 0.565 (0.005) \\
Correlated proxy & SIS--Cox--LASSO (one iteration) & 200 & 10,000 & 0.698 (0.003) & 0.302 (0.003) & 1.000 (0.000) & 0.410 (0.049) & 0.035 (0.005) & 0.334 (0.045) & 16.5 (2.6) & 0.523 (0.004) \\
Correlated proxy & S3VS (one iteration) & 200 & 15,000 & 0.702 (0.003) & 0.298 (0.003) & 1.000 (0.000) & 1.000 (0.000) & 0.015 (0.002) & 0.620 (0.049) & 1.0 (0.0) & 0.568 (0.006) \\
Correlated proxy & SIS--Cox--LASSO (one iteration) & 200 & 15,000 & 0.702 (0.003) & 0.298 (0.003) & 1.000 (0.000) & 0.490 (0.050) & 0.044 (0.006) & 0.433 (0.046) & 19.6 (2.7) & 0.529 (0.004) \\
Correlated proxy & S3VS (one iteration) & 200 & 20,000 & 0.702 (0.003) & 0.298 (0.003) & 1.000 (0.000) & 1.000 (0.000) & 0.010 (0.002) & 0.740 (0.044) & 1.0 (0.0) & 0.558 (0.006) \\
Correlated proxy & SIS--Cox--LASSO (one iteration) & 200 & 20,000 & 0.702 (0.003) & 0.298 (0.003) & 1.000 (0.000) & 0.590 (0.049) & 0.046 (0.006) & 0.532 (0.047) & 29.1 (3.0) & 0.531 (0.005) \\
Correlated proxy & S3VS (one iteration) & 500 & 10,000 & 0.702 (0.002) & 0.298 (0.002) & 1.000 (0.000) & 1.000 (0.000) & 0.028 (0.002) & 0.290 (0.046) & 1.0 (0.0) & 0.576 (0.004) \\
Correlated proxy & SIS--Cox--LASSO (one iteration) & 500 & 10,000 & 0.702 (0.002) & 0.298 (0.002) & 1.000 (0.000) & 0.920 (0.027) & 0.159 (0.009) & 0.804 (0.030) & 68.8 (3.8) & 0.553 (0.005) \\
Correlated proxy & S3VS (one iteration) & 500 & 15,000 & 0.702 (0.002) & 0.298 (0.002) & 1.000 (0.000) & 1.000 (0.000) & 0.029 (0.002) & 0.280 (0.045) & 1.0 (0.0) & 0.577 (0.004) \\
Correlated proxy & SIS--Cox--LASSO (one iteration) & 500 & 15,000 & 0.702 (0.002) & 0.298 (0.002) & 1.000 (0.000) & 0.990 (0.010) & 0.159 (0.008) & 0.896 (0.019) & 83.7 (3.3) & 0.553 (0.004) \\
Correlated proxy & S3VS (one iteration) & 500 & 20,000 & 0.699 (0.002) & 0.301 (0.002) & 1.000 (0.000) & 1.000 (0.000) & 0.025 (0.002) & 0.380 (0.049) & 1.0 (0.0) & 0.571 (0.004) \\
Correlated proxy & SIS--Cox--LASSO (one iteration) & 500 & 20,000 & 0.699 (0.002) & 0.301 (0.002) & 1.000 (0.000) & 0.990 (0.010) & 0.150 (0.008) & 0.923 (0.014) & 89.0 (3.3) & 0.554 (0.004) \\
Correlated proxy & S3VS (one iteration) & 800 & 10,000 & 0.700 (0.002) & 0.300 (0.002) & 1.000 (0.000) & 1.000 (0.000) & 0.032 (0.002) & 0.200 (0.040) & 1.0 (0.0) & 0.578 (0.004) \\
Correlated proxy & SIS--Cox--LASSO (one iteration) & 800 & 10,000 & 0.700 (0.002) & 0.300 (0.002) & 1.000 (0.000) & 1.000 (0.000) & 0.234 (0.010) & 0.877 (0.021) & 94.3 (3.6) & 0.567 (0.005) \\
Correlated proxy & S3VS (one iteration) & 800 & 15,000 & 0.700 (0.002) & 0.300 (0.002) & 1.000 (0.000) & 1.000 (0.000) & 0.031 (0.002) & 0.220 (0.042) & 1.0 (0.0) & 0.576 (0.004) \\
Correlated proxy & SIS--Cox--LASSO (one iteration) & 800 & 15,000 & 0.700 (0.002) & 0.300 (0.002) & 1.000 (0.000) & 1.000 (0.000) & 0.238 (0.009) & 0.931 (0.004) & 101.3 (2.9) & 0.567 (0.004) \\
Correlated proxy & S3VS (one iteration) & 800 & 20,000 & 0.698 (0.001) & 0.302 (0.001) & 1.000 (0.000) & 1.000 (0.000) & 0.034 (0.001) & 0.150 (0.036) & 1.0 (0.0) & 0.585 (0.003) \\
Correlated proxy & SIS--Cox--LASSO (one iteration) & 800 & 20,000 & 0.698 (0.001) & 0.302 (0.001) & 1.000 (0.000) & 1.000 (0.000) & 0.242 (0.008) & 0.932 (0.003) & 98.0 (2.5) & 0.572 (0.004) \\
No proxy & S3VS (one iteration) & 200 & 10,000 & 0.702 (0.003) & 0.298 (0.003) & 1.000 (0.000) & 1.000 (0.000) & 0.004 (0.001) & 0.890 (0.031) & 1.0 (0.0) & 0.505 (0.003) \\
No proxy & SIS--Cox--LASSO (one iteration) & 200 & 10,000 & 0.702 (0.003) & 0.298 (0.003) & 1.000 (0.000) & 0.120 (0.033) & 0.010 (0.003) & 0.116 (0.032) & 8.4 (2.4) & 0.501 (0.001) \\
No proxy & S3VS (one iteration) & 200 & 15,000 & 0.701 (0.003) & 0.299 (0.003) & 1.000 (0.000) & 1.000 (0.000) & 0.004 (0.001) & 0.900 (0.030) & 1.0 (0.0) & 0.505 (0.003) \\
No proxy & SIS--Cox--LASSO (one iteration) & 200 & 15,000 & 0.701 (0.003) & 0.299 (0.003) & 1.000 (0.000) & 0.150 (0.036) & 0.008 (0.003) & 0.147 (0.035) & 10.6 (2.6) & 0.501 (0.001) \\
No proxy & S3VS (one iteration) & 200 & 20,000 & 0.699 (0.003) & 0.301 (0.003) & 1.000 (0.000) & 1.000 (0.000) & 0.001 (0.001) & 0.970 (0.017) & 1.0 (0.0) & 0.502 (0.003) \\
No proxy & SIS--Cox--LASSO (one iteration) & 200 & 20,000 & 0.699 (0.003) & 0.301 (0.003) & 1.000 (0.000) & 0.330 (0.047) & 0.007 (0.002) & 0.328 (0.047) & 23.6 (3.6) & 0.497 (0.002) \\
No proxy & S3VS (one iteration) & 500 & 10,000 & 0.700 (0.002) & 0.300 (0.002) & 1.000 (0.000) & 1.000 (0.000) & 0.013 (0.002) & 0.670 (0.047) & 1.0 (0.0) & 0.509 (0.002) \\
No proxy & SIS--Cox--LASSO (one iteration) & 500 & 10,000 & 0.700 (0.002) & 0.300 (0.002) & 1.000 (0.000) & 0.760 (0.043) & 0.198 (0.013) & 0.714 (0.040) & 88.8 (5.5) & 0.513 (0.002) \\
No proxy & S3VS (one iteration) & 500 & 15,000 & 0.701 (0.002) & 0.299 (0.002) & 1.000 (0.000) & 1.000 (0.000) & 0.013 (0.002) & 0.670 (0.047) & 1.0 (0.0) & 0.510 (0.002) \\
No proxy & SIS--Cox--LASSO (one iteration) & 500 & 15,000 & 0.701 (0.002) & 0.299 (0.002) & 1.000 (0.000) & 0.840 (0.037) & 0.168 (0.010) & 0.801 (0.035) & 101.8 (5.1) & 0.511 (0.002) \\
No proxy & S3VS (one iteration) & 500 & 20,000 & 0.701 (0.002) & 0.299 (0.002) & 1.000 (0.000) & 1.000 (0.000) & 0.010 (0.002) & 0.760 (0.043) & 1.0 (0.0) & 0.510 (0.002) \\
No proxy & SIS--Cox--LASSO (one iteration) & 500 & 20,000 & 0.701 (0.002) & 0.299 (0.002) & 1.000 (0.000) & 0.940 (0.024) & 0.181 (0.009) & 0.905 (0.023) & 120.8 (3.4) & 0.514 (0.002) \\
No proxy & S3VS (one iteration) & 800 & 10,000 & 0.700 (0.002) & 0.300 (0.002) & 1.000 (0.000) & 1.000 (0.000) & 0.028 (0.002) & 0.300 (0.046) & 1.0 (0.0) & 0.522 (0.002) \\
No proxy & SIS--Cox--LASSO (one iteration) & 800 & 10,000 & 0.700 (0.002) & 0.300 (0.002) & 1.000 (0.000) & 0.980 (0.014) & 0.468 (0.011) & 0.900 (0.013) & 145.8 (2.7) & 0.539 (0.002) \\
No proxy & S3VS (one iteration) & 800 & 15,000 & 0.699 (0.002) & 0.301 (0.002) & 1.000 (0.000) & 1.000 (0.000) & 0.022 (0.002) & 0.440 (0.050) & 1.0 (0.0) & 0.518 (0.002) \\
No proxy & SIS--Cox--LASSO (one iteration) & 800 & 15,000 & 0.699 (0.002) & 0.301 (0.002) & 1.000 (0.000) & 1.000 (0.000) & 0.414 (0.010) & 0.932 (0.002) & 154.2 (1.7) & 0.533 (0.002) \\
No proxy & S3VS (one iteration) & 800 & 20,000 & 0.699 (0.002) & 0.301 (0.002) & 1.000 (0.000) & 1.000 (0.000) & 0.018 (0.002) & 0.560 (0.050) & 1.0 (0.0) & 0.512 (0.002) \\
No proxy & SIS--Cox--LASSO (one iteration) & 800 & 20,000 & 0.699 (0.002) & 0.301 (0.002) & 1.000 (0.000) & 1.000 (0.000) & 0.396 (0.009) & 0.937 (0.001) & 157.6 (0.8) & 0.531 (0.002) \\
\bottomrule
\end{tabular}
\end{adjustbox}
\end{table}

\begin{figure}[htbp]
\centering
\includegraphics[width=\textwidth]{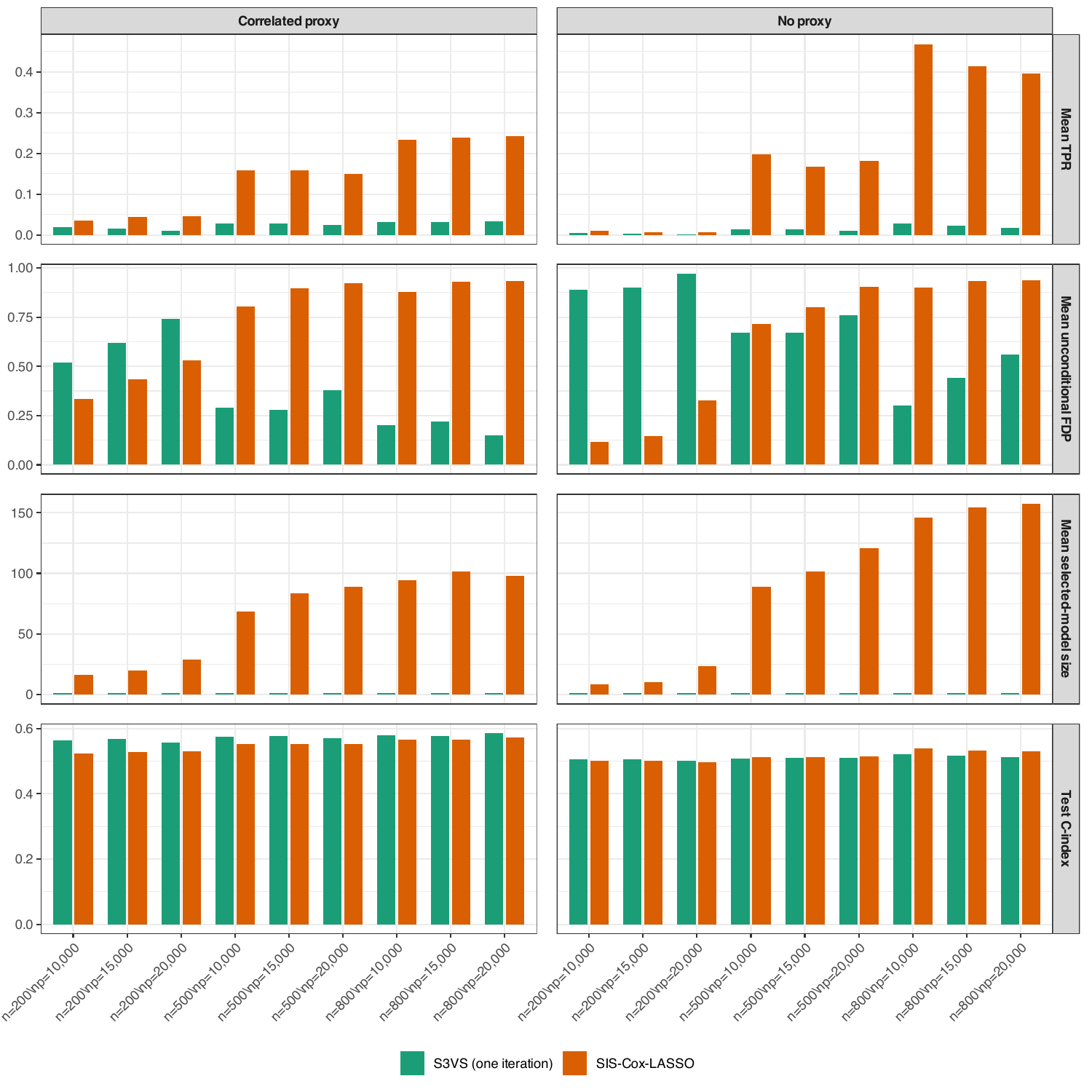}
\caption{Full-grid results for the one-iteration comparison of S3VS and SIS--Cox--LASSO. Rows show mean TPR, mean unconditional FDP, mean selected-model size, and test concordance. Columns show the correlated-proxy and no-proxy scenarios.}
\label{fig:surv-one-iteration-full-grid}
\end{figure}

\end{document}